\documentclass[10pt,a4paper,twocolumn]{article}
\usepackage[margin=17mm,top=20mm,bottom=20mm,columnsep=6mm]{geometry}
\usepackage{cite}
\usepackage[T1]{fontenc}
\usepackage{newtxtext,newtxmath}
\usepackage{amsmath,bm,mathtools}
\usepackage{graphicx,booktabs,tabularx,array,multirow}
\newcolumntype{P}[1]{>{\raggedright\arraybackslash}p{#1}}
\newcolumntype{Y}{>{\raggedright\arraybackslash}X}
\usepackage{tikz}
\usetikzlibrary{arrows.meta,positioning,fit,calc,backgrounds}
\usepackage{algorithm,algpseudocode}
\usepackage[font=small,labelfont=bf]{caption}
\usepackage{subcaption}
\usepackage{microtype,xurl,xcolor} 
\usepackage{stfloats,flushend}
\usepackage{fancyhdr,enumitem}
\usepackage[hidelinks,bookmarksopen=true,bookmarksopenlevel=1,pdftitle={QNETSIM for memory assisted design of satellite and hybrid quantum networks},pdfauthor={Seid Koudia, Francesco Vista, Intidhar Bedhief, Ilora Maity, Heyang Peng, Symeon Chatzinotas}]{hyperref}
\allowdisplaybreaks[1]
\setlist{nosep,leftmargin=*}
\newcommand{\ket}[1]{\lvert #1\rangle}
\newcommand{\bra}[1]{\langle #1\rvert}
\newcommand{\proj}[1]{\ket{#1}\!\bra{#1}}
\newcommand{\tr}{\operatorname{Tr}}
\newcommand{\id}{\mathbb{I}}
\newcommand{\E}{\mathbb{E}}

\newcommand{\code}[1]{\texttt{\small #1}}

\newcommand{\dd}{\mathrm{d}}
\newtheorem{proposition}{Proposition}

\newenvironment{proof}{\par\noindent\textit{Proof.}\ }{\hfill$\square$\par}

\newcommand{\AuthorNames}{Seid Koudia\textsuperscript{*}, Francesco Vista, Heyang Peng, Ilora Maity, Intidhar Bedhief, Juan Carlos Hernandez Hernandez, Symeon Chatzinotas}
\newcommand{\AuthorAffiliation}{Signal Processing and Communications (SIGCOM),\\Interdisciplinary Centre for Security, Reliability and Trust (SnT),\\University of Luxembourg, L-1855 Luxembourg, Luxembourg}
\newcommand{\CorrespondingAuthor}{\textsuperscript{*}Corresponding author: \href{mailto:seid.koudia@uni.lu}{seid.koudia@uni.lu}}


\newcommand{\PaperTitle}{QNETSIM for memory assisted design of satellite and hybrid quantum networks}

\begin{document}
\twocolumn[{
\begin{center}
{\LARGE\bfseries \PaperTitle\par}
\vspace{8pt}
{\normalsize\AuthorNames\par}
\vspace{4pt}
{\small\AuthorAffiliation\par}
\vspace{3pt}
{\small\CorrespondingAuthor\par}
\end{center}
\vspace{6pt}
\noindent\textbf{Abstract.} Space-based quantum information networks require simulation tools that capture orbital dynamics, probabilistic quantum operations, and classical coordination. We introduce QNETSIM, a modular discrete-event simulator for the design and performance evaluation of satellite and hybrid space–ground quantum networks. The framework integrates orbital geometry, free-space optical and fiber channels, quantum memories, entanglement generation, and Bell-state measurements within a common event-driven architecture. Physical models account for optical loss, atmospheric turbulence, detector imperfections, finite memory capacity, and storage-induced decoherence. Space–time routing and segmented nesting support entanglement distribution across changing topologies, while a software-defined control plane coordinates generation, storage, and swapping. Configurable scenarios encompass onboard repeaters, satellite constellations, ground-based repeaters, and direct satellite distribution. Through explicit interfaces between physical models, network protocols, and orchestration, QNETSIM links component characteristics and operational policies to entanglement fidelity, throughput, latency, availability, and memory utilization, enabling systematic exploration of architectural and resource-management trade-offs in future space-based quantum information networks.
\vspace{12pt}
}]

\section{Introduction}
\label{sec:introduction}

Quantum information networks aim to connect quantum processors, memories, sensors, and photonic devices through the distribution of quantum states and entanglement. By extending quantum resources beyond individual laboratories, these networks can support quantum key distribution (QKD), distributed quantum computation, and entanglement-enhanced sensing and metrology \cite{wehner2018,parny2023}. Their development requires more than the transmission of isolated photons: useful quantum correlations must be established between selected endpoints, preserved until they are needed, and consumed by applications under appropriate timing and quality constraints. Space-based quantum information networks offer a route to extending these capabilities across continental and intercontinental distances, complementing terrestrial optical-fiber infrastructure with satellite links.

The motivation for a space segment originates in the physical limitations of long-distance quantum transmission. Optical-fiber transmissivity decreases exponentially with distance, imposing a fundamental rate--loss trade-off on repeaterless quantum communication \cite{pirandola2017}. Unknown quantum states cannot be noiselessly copied or regenerated using conventional optical amplification. Satellite links offer a different propagation geometry: most of the path lies outside the atmosphere, avoiding the material absorption accumulated along a comparable fiber connection. Their performance is instead governed by diffraction, finite receiving apertures, atmospheric propagation, and terminal pointing \cite{pirandola_sat}. The satellite distribution of entanglement demonstrated by Yin et al. established an experimental foundation for this approach \cite{yin2017}. Its extension to network services requires connecting satellite-generated resources to ground-based users and, where appropriate, to quantum memories and repeater nodes.

This hybrid perspective is particularly relevant to the European Quantum Communication Infrastructure (EuroQCI). The European Commission and the European Space Agency envisage an infrastructure combining national terrestrial quantum networks with a space segment, supporting the protection of sensitive communications and critical infrastructure across Europe \cite{euroqci2025}. EAGLE-1 is being developed to demonstrate satellite-to-ground QKD and provide operational experience for this wider infrastructure \cite{eagle1}. Such initiatives make the integration of optical ground stations, terrestrial networks, and their control systems an immediate engineering concern. QKD delivers classical secret keys, whereas general entanglement services additionally require the distribution and manipulation of shared quantum states. Moving toward memory-assisted networking therefore raises further questions about storage, swapping, resource allocation, and interfaces between independently operated network segments.

Space-based quantum links face several coupled physical limitations. Orbital motion changes slant range, elevation, propagation delay, and the duration of useful contact, while atmospheric extinction, turbulence, clouds, and background illumination influence the optical opportunities within that contact \cite{pirandola_sat,khatri2021}. Finite apertures and pointing errors determine how much light is collected; wavefront distortion and adaptive-optics performance influence coupling into the accepted receiving modes \cite{ao}. Source brightness, multipair emission, detector efficiency, dark counts, timing jitter, and dead time then determine which arrivals contribute to usable correlations \cite{ma2007,szlachetka,bonato}. Consequently, increasing the emitted photon rate or extending geometrical visibility does not necessarily improve application performance. Wavelength, receiver acceptance, detector technology, and operating conditions must be assessed together.

Quantum memories and Bell-state measurements introduce a further set of constraints. Memories permit asynchronous elementary links to be joined, but their usefulness depends jointly on write and retrieval efficiencies, coherence time, bandwidth, and the number of independently addressable modes. Ensemble optical depth and control-field dynamics influence storage and retrieval, while environmental noise degrades the stored logical state \cite{gorshkov_prl,gorshkov_free}. Entanglement swapping additionally requires compatible optical modes and a successful measurement, with passive linear-optical Bell analyzers subject to an intrinsic discrimination limit under their standard resource assumptions \cite{bsm_limit}. Repeater placement consequently trades reduced synchronization demands against storage loss, decoherence, and additional probabilistic operations. For onboard memories, these trade-offs also involve payload integration and space qualification, motivating quantitative capability studies rather than assumptions of ideal storage \cite{gundogan2021}.

The same physical constraints shape network control and management. A geometrically available route is useful only if its elementary resources can be generated, stored, and joined before their quality or the contact opportunities expire. Routing must therefore account for time-dependent connectivity, memory occupancy, and the ordering of swaps. Classical communication is integral to this process: heralding, resource identification, reservation, measurement outcomes, and completion notifications must remain associated with the correct quantum resources \cite{rfc9340}. Signaling and processing delays can consume part of the available storage interval. A software-defined architecture separates network-wide scheduling and admission policies from local execution, while fast pointing, synchronization, and optical correction remain terminal functions. Management must also reconcile heterogeneous device capabilities, service requirements, and control interfaces across the space and terrestrial segments \cite{parny2023,rfc9340}.

These interdependencies make system-level simulation an essential instrument for mission planning. Before committing to payload and ground-segment designs, mission studies must assess orbit selection, ground-station placement, optical apertures, wavelengths, memory provisioning, and operating policies against a common set of service objectives. A longer contact can be offset by higher propagation loss; additional memory capacity can remain underused when loading is inefficient; and a shorter route can perform poorly when its resource-generation times are incompatible. Link-budget calculations provide necessary physical inputs, but do not by themselves resolve stochastic waiting, finite-capacity contention, failed swaps, or delayed control actions. Discrete-event simulation provides a natural means of representing these dependencies and relating component assumptions to delivered-pair rate, conditional fidelity, latency, availability, and resource utilization \cite{netsquid,sequence,meister}. Configuration-controlled studies can then support requirement allocation and sensitivity analysis before experimental deployment.

Existing simulation and architecture studies provide important foundations. NetSquid~\cite{netsquid}, QuNetSim~\cite{qunetsimdocs}, QNET~\cite{qnetdocs}, SeQUeNCe~\cite{sequence}, and QuISP~\cite{quisp} offer complementary approaches to device modeling, event execution, and quantum-network protocols. Satellite-specific work has examined ground-memory architectures, orbiting repeaters, constellation coverage, and space-borne memory requirements \cite{boone2015,liorni2021,khatri2021,gundogan2021}. Event-based multi-satellite studies, satellite-connected quantum cities, and coupled orbital--optical simulations further demonstrate the importance of temporal and physical modeling \cite{wallnofer2022,yehia,meister}. The QUASAR simulator also addresses dynamic satellite architectures and routing \cite{shi2026}. Table~\ref{tab:comparison} summarizes the simulator abstractions relevant to the present work. These contributions motivate an integrated mission-design environment in which physical assumptions, finite quantum resources, and programmable policies can be varied through explicit interfaces.

The roadmap of de Forges de Parny et al. distinguishes a direct entanglement-resource demonstrator from a networking stage that assembles elementary links through storage and swapping \cite{parny2023}. QNETSIM addresses the system-level design questions raised by this transition. We introduce it as a modular discrete-event simulator for space-based and hybrid space--ground quantum information networks. Its architecture combines orbital geometry, free-space optical and fiber channels, quantum memories, entanglement generation and swapping, network algorithms, and software-defined orchestration on a common event timeline. Physical models are separated from the scheduling core and expose operation probabilities and conditional quality descriptors to node-level resource management. The accompanying physical formulation relates matter--light interactions and open-system evolution to effective storage and measurement channels, providing a basis for selecting and refining the network-level descriptions. Space--time-based distribution and segmented nesting build on the framework of Koudia et al. \cite{stbd2024}. Within QNETSIM, these mechanisms connect evolving link opportunities to the scheduling and consumption of stored entanglement resources.

Three scenarios organize the evaluation around complementary mission-design questions. In S1, an onboard repeater connects two optical ground stations, exposing the interaction between geometry, asynchronous loading, and memory coherence. In S2, a low-Earth-orbit constellation introduces time-dependent routing and segmented swapping. In S3, a ground-based repeater at Windhof joins a satellite downlink to a fiber connection toward Kirchberg, examining heterogeneous resource supply. A separate direct Matera--Helmos calculation provides a memory-free optical reference for receiver and source operating-point studies. Together, these cases connect architectural choices to observable network behavior and establish a framework for investigating the transition from individual satellite quantum links to managed entanglement services. The Results present the service indicators, architecture, and scenario findings; the Methods detail the physical models, algorithms, configurations, and statistical treatment, including the derivations in Sections~\ref{app:memories}--\ref{app:stbd}.

\begin{table*}[!tp]
\centering\small
\caption{Representative simulator abstractions. Entries identify modeling and interface scope, rather than an accuracy or speed ranking.}
\label{tab:comparison}
\begin{tabularx}{\textwidth}{P{1.65cm}P{2.65cm}YY}
\toprule
Tool & Execution / interface & State and component representation & Protocol and satellite scope\\
\midrule
QuNetSim \cite{qunetsimdocs} & Python, event-driven & Backend quantum states; configurable connections & Quantum/classical messages, routing, teleportation, and entanglement protocols.\\
NetSquid \cite{netsquid,netsquiddocs} & Python; optimized C/Cython DES & Ket, stabilizer, and density matrices; components with time-dependent noise & Asynchronous protocols and classical control; satellite applications include \cite{yehia}.\\
QNET \cite{qnetdocs} & Python DES; QCompute-related interfaces & Modular devices, channels, and physical models & Resource/protocol simulation and computation interfaces; documented Micius example.\\
SeQUeNCe \cite{sequence} & Python DES & Hardware and quantum-resource components & Layered entanglement, resource, network, and application management.\\
QuISP \cite{quisp} & C++ / OMNeT++ & Scalable error-based resource representation & Heterogeneous repeater-network protocols.\\
QNETSIM & Native token DES; Python workflows; REST control & Analytical FSO/fiber and platform-specific memories; finite modes; field-to-instrument reference models & Single repeater, constellation, and hybrid workflows; STBD, nesting, and request/QoS control.\\
\bottomrule
\end{tabularx}
\end{table*}
\section{Results}
\subsection{Service requirements and performance measures}
\label{sec:services}
The memory-assisted service is a pair shared between specified endpoints, with a quality descriptor and the classical information required to interpret its measurement frame. A direct detected coincidence is a different product: detection consumes the arriving photons. Table~\ref{tab:kpis} maps the service indicators of Ref.~\cite{parny2023} to quantities recorded by the simulator and to the associated design decisions.

\begin{table*}[!tp]
\centering\small
\caption{Service requirements and corresponding observables. Observation time, contact time, and event count are kept distinct when reporting rates.}
\label{tab:kpis}
\begin{tabularx}{\textwidth}{P{2.3cm}P{7.3cm}Y}
\toprule
Requirement & Observable & Design decision\\
\midrule
Resource rate & Delivered pairs per observation second; separately pairs per pass or contact second & Source allocation and optical/memory provisioning.\\
State quality & Conditional Bell fidelity of a declared delivery population; accepted counts at a target & Memory coherence, age cutoff, and swap depth.\\
Latency & Request-to-delivery time, separating cold start from stocked inventory and physical readiness from known frame & Proactive storage and classical signaling.\\
Availability & Fraction of requests meeting quality and deadline targets, distinct from geometrical visibility & Temporal routing and coverage.\\
Resource cost & Occupied memory-mode seconds per delivered pair; overflow and expiry counts & Bank capacity and buffer policy.\\
Timing and control & Relative photon-arrival uncertainty, gate acceptance, control traffic and delay & Synchronization and controller placement.\\
\bottomrule
\end{tabularx}
\end{table*}

A usable-pair rate counts only deliveries meeting the specified fidelity and deadline. Failed requests remain in the reliability denominator. Conditional fidelity is undefined when there are no deliveries, and a bin average cannot determine how many individual pairs exceed a new quality threshold. These distinctions prevent a high raw count from being mistaken for a high-quality service. Full KPI definitions and the retained bin-level statistics are given in Sections~\ref{si:records} and \ref{sec:extended-results}.

\subsection{A common execution and resource architecture}
\label{sec:architecture}
Figure~\ref{fig:architecture} shows the separation of application requests, network decisions, local agents, and physical models. The event core supplies one ordered simulation timeline. Source, storage, and repeater agents evaluate physical operations when their events execute. A resource ledger records which excitations exist and where they are stored. A controller authorizes future actions using information received through classical messages; it does not directly modify a physical state.
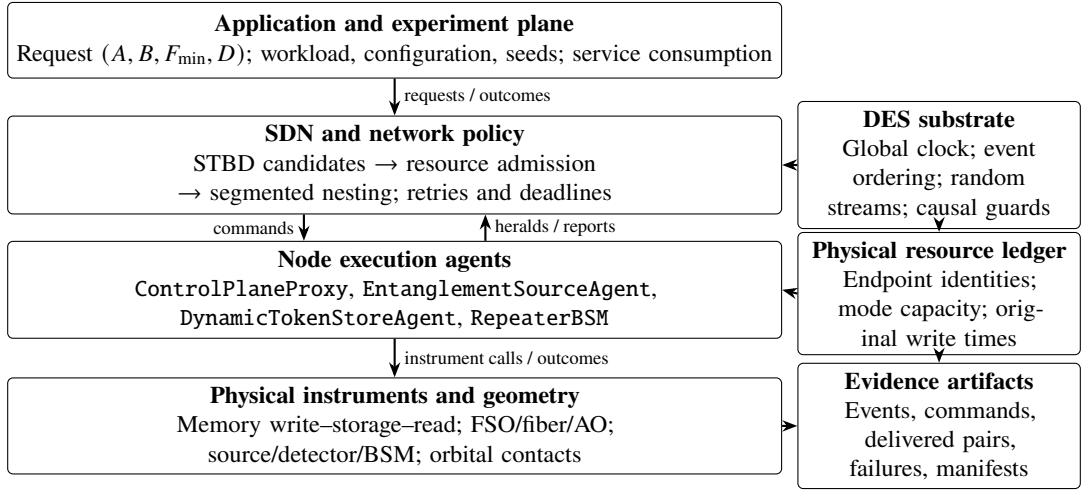
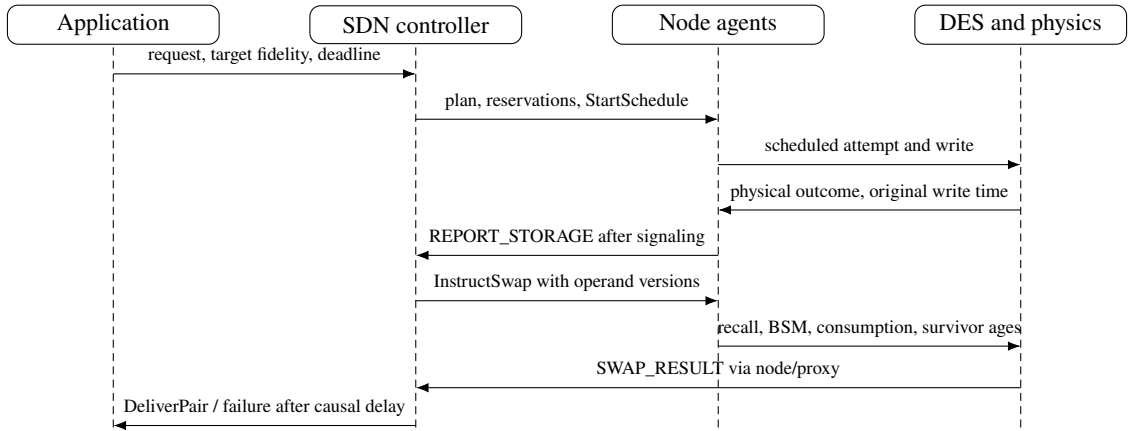
\begin{figure*}[!tbp]
\centering
\begin{subfigure}[t]{\textwidth}
\centering
\begin{tikzpicture}[font=\small,box/.style={draw,rounded corners=2pt,align=center,minimum height=1.0cm,text width=10.0cm},side/.style={draw,rounded corners=2pt,align=center,text width=3.5cm,minimum height=1.2cm},arr/.style={-{Stealth[length=2mm]},thick}]
\node[box] (app) at (0,0) {\textbf{Application and experiment plane}\\Request $(A,B,F_{\min},D)$; workload, configuration, seeds; service consumption};
\node[box] (ctl) at (0,-1.65) {\textbf{SDN and network policy}\\STBD candidates $\rightarrow$ resource admission $\rightarrow$ segmented nesting; retries and deadlines};
\node[box] (agents) at (0,-3.3) {\textbf{Node execution agents}\\\texttt{ControlPlaneProxy}, \texttt{EntanglementSourceAgent},\\\texttt{DynamicTokenStoreAgent}, \texttt{RepeaterBSM}};
\node[box] (phys) at (0,-5.1) {\textbf{Physical instruments and geometry}\\Memory write--storage--read; FSO/fiber/AO; source/detector/BSM; orbital contacts};
\node[side] (des) at (7.2,-1.65) {\textbf{DES substrate}\\Global clock; event ordering; random streams; causal guards};
\node[side] (bank) at (7.2,-3.35) {\textbf{Physical resource ledger}\\Endpoint identities; mode capacity; original write times};
\node[side] (logs) at (7.2,-5.1) {\textbf{Evidence artifacts}\\Events, commands, delivered pairs, failures, manifests};
\draw[arr] (app.south) -- node[right,font=\scriptsize]{requests / outcomes} (ctl.north);
\draw[arr] ([xshift=-12mm]ctl.south) -- node[left,font=\scriptsize]{commands} ([xshift=-12mm]agents.north);
\draw[arr] ([xshift=12mm]agents.north) -- node[right,font=\scriptsize]{heralds / reports} ([xshift=12mm]ctl.south);
\draw[arr] (agents.south) -- node[right,font=\scriptsize]{instrument calls / outcomes} (phys.north);
\draw[arr] (des.west) -- (ctl.east);
\draw[arr] (des.south) -- (bank.north);
\draw[arr] (bank.west) -- (agents.east);
\draw[arr] (phys.east) -- (logs.west);
\draw[arr] (bank.south) -- (logs.north);
\end{tikzpicture}

\caption{Layered execution and resource architecture.}
\label{fig:architecture}
\end{subfigure}
\par\vspace{3mm}
\begin{subfigure}[t]{\textwidth}
\centering
\begin{tikzpicture}[x=1cm,y=1cm,font=\small,>=Latex]
\foreach \x/\label in {0/Application,4/SDN controller,8/Node agents,12/DES and physics}{
\node[draw,rounded corners,minimum width=2.8cm,minimum height=.5cm] at (\x,0) {\label};
\draw[densely dashed] (\x,-.3)--(\x,-5.35);}
\draw[->] (0,-.65)--node[above,font=\scriptsize]{request, target fidelity, deadline}(4,-.65);
\draw[->] (4,-1.25)--node[above,font=\scriptsize]{plan, reservations, StartSchedule}(8,-1.25);
\draw[->] (8,-1.85)--node[above,font=\scriptsize]{scheduled attempt and write}(12,-1.85);
\draw[->] (12,-2.45)--node[above,font=\scriptsize]{physical outcome, original write time}(8,-2.45);
\draw[->] (8,-3.05)--node[above,font=\scriptsize]{REPORT\_STORAGE after signaling}(4,-3.05);
\draw[->] (4,-3.65)--node[above,font=\scriptsize]{InstructSwap with operand versions}(8,-3.65);
\draw[->] (8,-4.25)--node[above,font=\scriptsize]{recall, BSM, consumption, survivor ages}(12,-4.25);
\draw[->] (12,-4.8)--node[above,font=\scriptsize]{SWAP\_RESULT via node/proxy}(4,-4.8);
\draw[->] (4,-5.3)--node[above,font=\scriptsize]{DeliverPair / failure after causal delay}(0,-5.3);
\end{tikzpicture}

\caption{Event/action information flow.}
\label{fig:controlflow}
\end{subfigure}
\caption{\textbf{QNETSIM architecture and orchestration.} (a) Physical models are evaluated at discrete-event boundaries, node agents execute local operations, and the resource ledger records endpoint occupancy and state descriptors. Solid arrows denote software interactions, not quantum-state copying; the instrument-valued formulation specifies the meaning of compact physical interfaces. (b) Request, generation, storage, swap, and delivery messages are scheduled operations rather than instantaneous transfers. Physical outcome, controller knowledge, and endpoint frame readiness have distinct timestamps. Local BSM operands are consumed while the original ages of surviving remote memories are preserved. The external-control transport carries telemetry and commands; effective execution times remain simulation events.}
\label{fig:architecture-control}
\end{figure*}

A pair has a unique identity, a request identifier, two endpoint records, a state representation, and a known or pending Pauli frame. Each endpoint retains its memory bank, mode, original write time, last evolution time, and allocation version. A successful swap consumes the two local operands and creates a new pair identity for the two surviving endpoints. Their original write times are preserved. Capacity counts both occupied and committed slots, including operations awaiting a herald. Multiple labels referring to the same physical bank share its capacity.

The physical interface separates outcome probability from the conditional output state. Compact configurations represent quality by a scalar fidelity or Werner parameter. Bell-probability vectors preserve independent Pauli errors at almost the same state-storage cost; general state-dependent filtering requires a density matrix or a more general quantum instrument. This choice is recorded with the simulation configuration. Sections~\ref{app:memories} and \ref{app:bsm} derive the accepted memory and measurement operations underlying those representations.

Events distinguish emission, arrival, memory loading, herald reception, swapping, and final notification. Time ordering uses the scheduled time followed by declared priority and insertion order. Source rates are allocated across modes rather than multiplied by memory capacity. Poisson thinning can skip unsuccessful attempts when those failures have no effect on the system state; periodic sources instead require geometric skipping. Detector busy periods, failed-attempt reservations, and source multipair statistics require explicit treatment where relevant. These conditions define the acceleration boundary without changing the resource budget.

\subsection{Orchestration and information flow}
\label{sec:control}
QNETSIM supports autonomous local execution, an in-process controller, and an external Ryu-hosted REST controller. Node proxies report timestamped storage and swap outcomes and retrieve pending actions. The controller maintains observed inventories, topology, request deadlines, and swap dependencies. Figure~\ref{fig:controlflow} follows the request, generation, storage, swap, and notification sequence; Table~\ref{tab:api} gives the message fields.

The separation between control and execution matters when messages are delayed. A decision based on a valid storage report can arrive after an operand has expired or been consumed. Local validation of the allocation version and request epoch rejects such stale instructions. Swaps reserve and consume their operands atomically. Quality estimates may be propagated from original write times, while future random channel realizations remain outside the controller's information set.

\begin{figure*}[!tp]
\centering
\includegraphics[width=.98\textwidth]{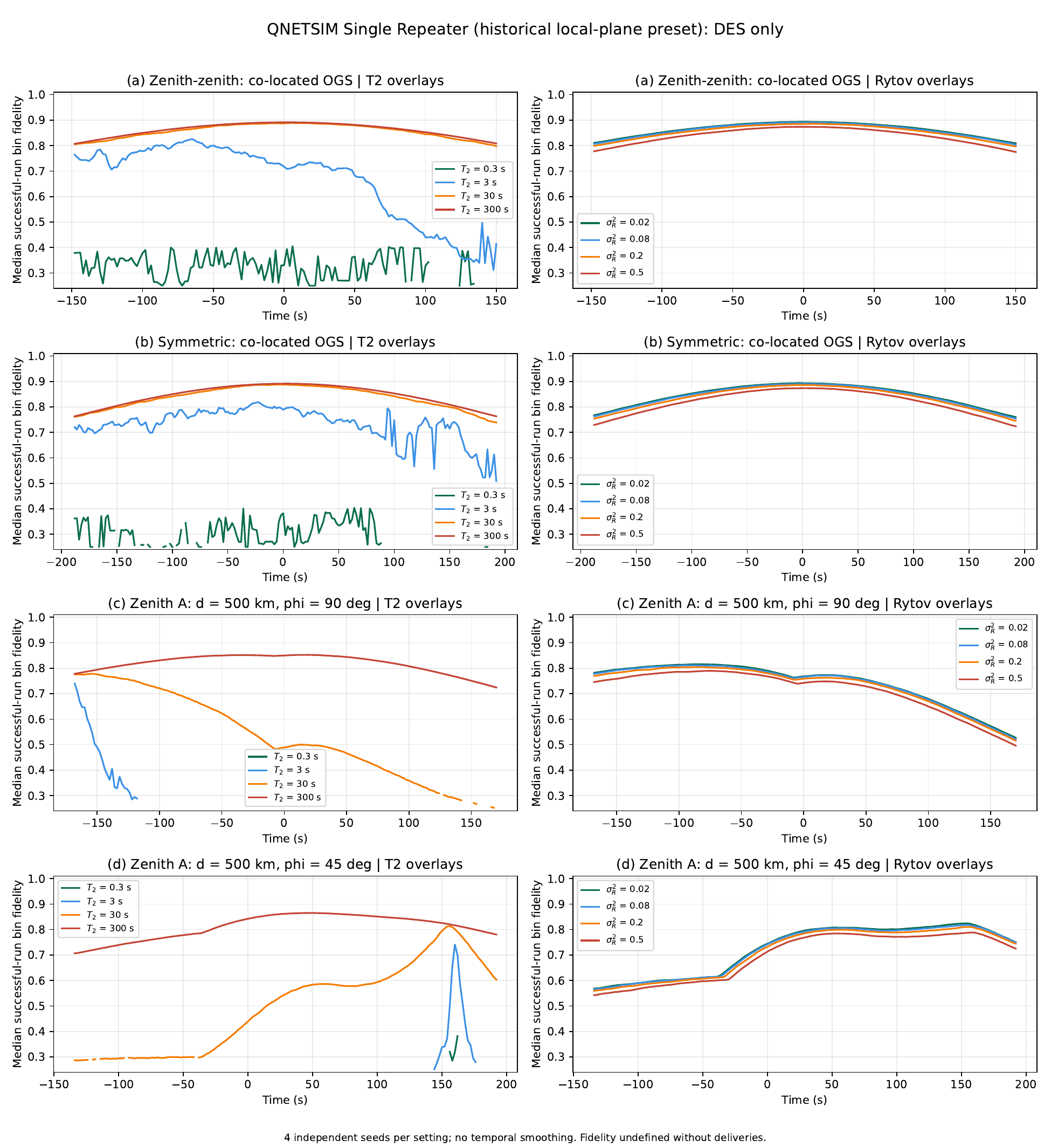}
\caption{Complete S1 conditional-fidelity sweep. Rows are cases A--D; the left column varies $T_2$, and the right varies $\sigma_R^2$. Each curve is the median of supported run-bin means over four repeats, with no temporal smoothing. Empty-support bins remain undefined. The first two rows have co-located stations, not a 1000\,km separation; the lower rows have 500\,km separation and an Alice-anchored pass.}
\label{fig:full-fidelity}
\end{figure*}

\begin{figure*}[!tp]
\centering
\includegraphics[width=.98\textwidth]{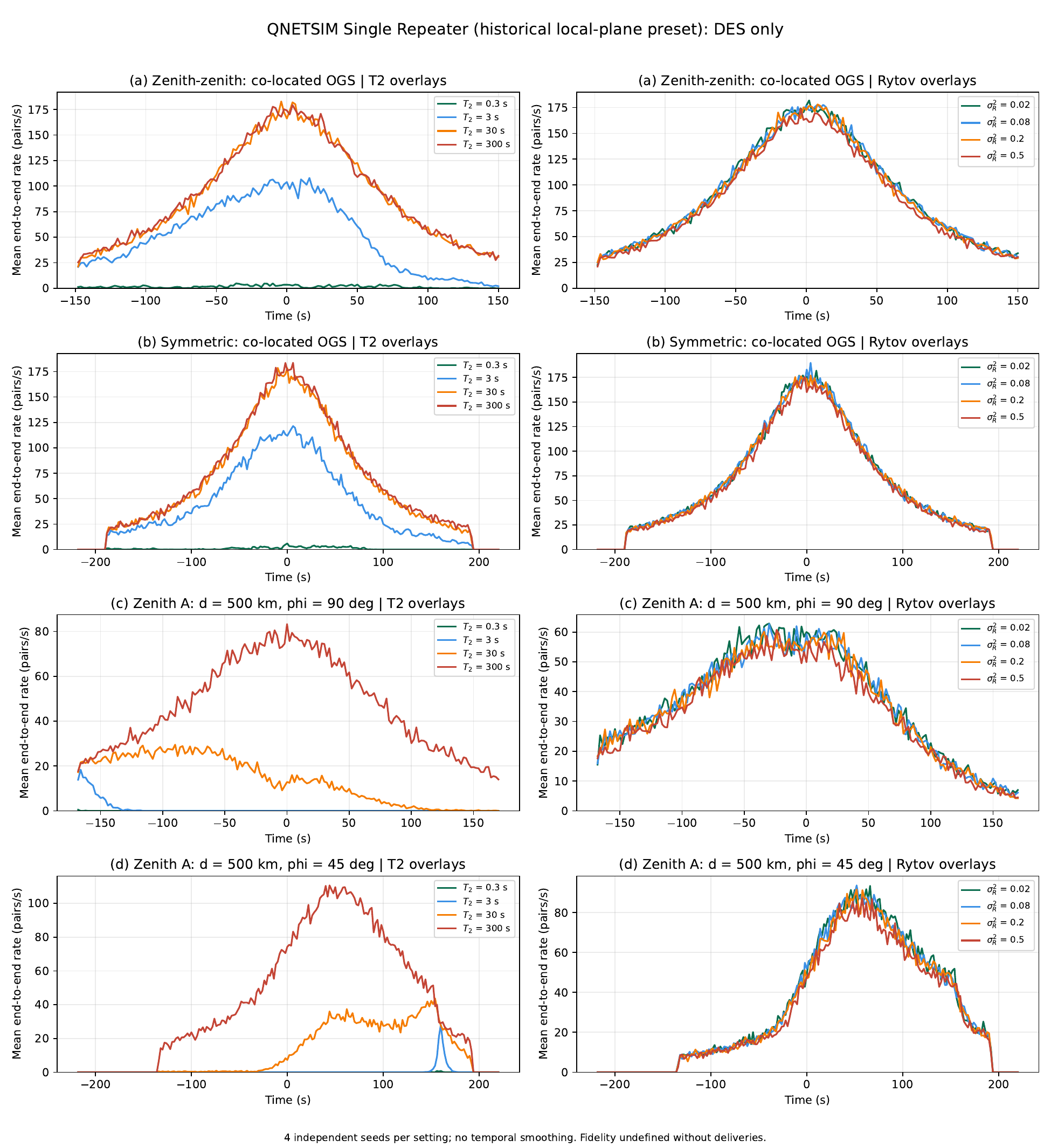}
\caption{Complete S1 rate sweep for the same 32 settings and 128 runs as Fig.~\ref{fig:full-fidelity}. The ordinate is the mean recorded end-to-end delivery rate in pairs/s, obtained from 2\,s bins; no rate multiplier or temporal smoothing is applied. The geometrical and observation-window definitions are given in Section~\ref{sec:s1}.}
\label{fig:full-rate}
\end{figure*}

Classical propagation, serialization, queueing, processing, and polling are separate modeled delays. Their sum determines when an instruction can become eligible on the simulation timeline; wall-clock HTTP execution time is a software measurement rather than automatically a simulated network delay. Optical wavefront correction, pointing, and source synchronization remain local control functions. The network controller sets slower scheduling objectives and responds to telemetry.

Ryu is retained as a compatibility adapter, while the scientific message contract is transport independent. Its upstream maintenance status motivates this separation \cite{ryu2026}. The scenario results below characterize the selected execution configurations; an external-controller equivalence  would additionally require matched input and action traces. The resource and event contracts provide the common comparison boundary.

\subsection{S1: Onboard storage and geometry-dependent memory requirements}
\label{sec:scenarios}\label{sec:s1}
Scenario S1 uses two elementary resources whose local halves occupy separate memories on a satellite; their remote halves belong to two OGSs. The onboard BSM consumes the local excitations and, on success, connects the ground endpoints. In a satellite-source realization, one photon from each elementary pair travels down its corresponding arm and its partner is retained locally. A stored-pair service also requires heralded loading at the remote endpoint. Optical conversion and loading are therefore distinct from photon detection and must be included in a hardware realization.

The controlled campaign comprises four local-plane geometries at 500\,km altitude and 7.2\,km/s projected speed. Cases A and B have co-located terminals and form limiting cases with $0^\circ$ and $90^\circ$ track directions. Cases C and D have 500\,km station separation and an Alice-anchored pass at $90^\circ$ and $45^\circ$. Their observation windows are 300, 440, 340, and 440\,s. These definitions prevent the co-located presets from being interpreted as separated-city links or their different-duration yields as a paired crossing-angle comparison.

The recorded cold-atom settings are $d=20$, dephasing exponent $n=1.3$, initial mode efficiency 0.9, and retrieval-timescale parameter $T_{\rm mem}=1000$\,s. The sweeps allocate 8192 modes per label, use a nominal source rate $5\times10^5$\,s$^{-1}$, and record counts in 2\,s bins. Optical defaults include 810\,nm wavelength, a 0.15\,m transmitter, a 0.50\,m receiver, and 12\,$\mu$rad divergence. Table~\ref{tab:campaign-parameters} records the numerical configuration fields.

Each geometry has four coherence settings at fixed Rytov variance $\sigma_R^2=0.08$ and four turbulence settings at the default $T_2=100$\,s. Four runs per setting give 32 settings and 128 runs. Table~\ref{tab:t2-yields} reports the mean integrated count and sample standard deviation between runs. Figures~\ref{fig:full-fidelity} and \ref{fig:full-rate} show the complete conditional-fidelity and delivery-rate sweeps for cases A--D, including the co-located limiting cases.

\begin{table*}[!tp]
\centering\small
\caption{S1 integrated delivered pairs across the coherence sweep. Entries are mean $\pm$ sample standard deviation over four runs, not confidence intervals. A and B are co-located limits; C and D have 500\,km station separation. Observation durations are 300, 440, 340, and 440\,s, respectively. All cases use $\sigma_R^2=0.08$.}
\label{tab:t2-yields}
\begin{tabular}{lrrrr}\toprule Case & $T_2=0.3$\,s & $T_2=3$\,s & $T_2=30$\,s & $T_2=300$\,s\\ \midrule
A & $446.2\,\pm\,383.7$ & $15,423.5\,\pm\,4,614.9$ & $27,201.5\,\pm\,242.6$ & $27,570.8\,\pm\,204.2$ \\
B & $362.5\,\pm\,335.3$ & $18,580.5\,\pm\,5,049.0$ & $28,838.8\,\pm\,621.7$ & $29,842.5\,\pm\,217.4$ \\
C & $1.2\,\pm\,1.0$ & $284.8\,\pm\,106.2$ & $4,750.2\,\pm\,267.6$ & $15,921.8\,\pm\,41.5$ \\
D & $4.2\,\pm\,3.3$ & $213.2\,\pm\,17.4$ & $5,264.2\,\pm\,137.2$ & $18,689.0\,\pm\,148.7$ \\
\bottomrule\end{tabular}

\end{table*}

\begin{table*}[!tp]
\centering\footnotesize
\caption{All 32 S1 settings. $\bar N$ and $s_N$ are the mean and sample standard deviation of integrated deliveries over four runs. $R_{\rm peak}$ is the maximum ensemble mean rate; $t_R$ is its bin-end time relative to the preset origin. $F_{\rm peak}$ is the maximum median of supported run-bin mean fidelities; $\bar F_{\rm pair}$ is the success-count-weighted fidelity over all bins and repeats. These are reanalyses of the recorded bin tables.}
\label{tab:all-settings}
\begin{tabular}{llrrrrrr}\toprule Case & Setting & $\bar N$ & $s_N$ & $R_{\rm peak}$ (s$^{-1}$) & $t_R$ (s) & $F_{\rm peak}$ & $\bar F_{\rm pair}$\\ \midrule
A & $T_2=0.3$ s & 446.25 & 383.74 & 4.750 & -36 & 0.4053 & 0.3665 \\
A & $T_2=3$ s & 15,423.50 & 4,614.93 & 108.000 & +16 & 0.8246 & 0.7538 \\
A & $T_2=30$ s & 27,201.50 & 242.63 & 182.875 & -4 & 0.8889 & 0.8690 \\
A & $T_2=300$ s & 27,570.75 & 204.23 & 179.375 & +4 & 0.8915 & 0.8741 \\
A & $\sigma_R^2=0.02$ & 27,673.50 & 54.76 & 181.875 & +2 & 0.8935 & 0.8761 \\
A & $\sigma_R^2=0.08$ & 27,583.50 & 135.74 & 177.750 & +8 & 0.8911 & 0.8729 \\
A & $\sigma_R^2=0.2$ & 27,212.25 & 173.48 & 178.000 & +8 & 0.8855 & 0.8671 \\
A & $\sigma_R^2=0.5$ & 26,047.75 & 73.68 & 174.000 & -4 & 0.8740 & 0.8535 \\
B & $T_2=0.3$ s & 362.50 & 335.31 & 6.000 & +0 & 0.4043 & 0.3603 \\
B & $T_2=3$ s & 18,580.50 & 5,048.96 & 121.375 & +6 & 0.8182 & 0.7697 \\
B & $T_2=30$ s & 28,838.75 & 621.74 & 178.625 & -10 & 0.8878 & 0.8611 \\
B & $T_2=300$ s & 29,842.50 & 217.43 & 183.625 & +6 & 0.8915 & 0.8682 \\
B & $\sigma_R^2=0.02$ & 29,653.50 & 242.08 & 181.875 & +10 & 0.8932 & 0.8694 \\
B & $\sigma_R^2=0.08$ & 29,574.50 & 221.04 & 189.875 & +4 & 0.8901 & 0.8666 \\
B & $\sigma_R^2=0.2$ & 29,217.50 & 169.95 & 177.000 & +4 & 0.8858 & 0.8605 \\
B & $\sigma_R^2=0.5$ & 28,031.00 & 182.33 & 175.375 & -2 & 0.8738 & 0.8468 \\
C & $T_2=0.3$ s & 1.25 & 0.96 & 0.500 & -168 & 0.2500 & 0.2735 \\
C & $T_2=3$ s & 284.75 & 106.19 & 18.375 & -166 & 0.7397 & 0.6237 \\
C & $T_2=30$ s & 4,750.25 & 267.65 & 29.500 & -90 & 0.7786 & 0.6428 \\
C & $T_2=300$ s & 15,921.75 & 41.55 & 83.250 & +0 & 0.8526 & 0.8327 \\
C & $\sigma_R^2=0.02$ & 12,474.75 & 180.66 & 62.875 & -30 & 0.8158 & 0.7708 \\
C & $\sigma_R^2=0.08$ & 12,255.75 & 176.35 & 62.750 & -30 & 0.8115 & 0.7676 \\
C & $\sigma_R^2=0.2$ & 12,020.00 & 121.91 & 60.125 & -34 & 0.8045 & 0.7606 \\
C & $\sigma_R^2=0.5$ & 11,382.00 & 80.03 & 60.750 & -24 & 0.7897 & 0.7450 \\
D & $T_2=0.3$ s & 4.25 & 3.30 & 0.750 & +156 & 0.3819 & 0.3157 \\
D & $T_2=3$ s & 213.25 & 17.44 & 28.125 & +160 & 0.7412 & 0.6310 \\
D & $T_2=30$ s & 5,264.25 & 137.15 & 43.625 & +154 & 0.8130 & 0.6298 \\
D & $T_2=300$ s & 18,689.00 & 148.65 & 110.625 & +44 & 0.8651 & 0.8343 \\
D & $\sigma_R^2=0.02$ & 14,342.25 & 151.59 & 93.375 & +66 & 0.8243 & 0.7737 \\
D & $\sigma_R^2=0.08$ & 14,142.25 & 128.11 & 93.625 & +52 & 0.8182 & 0.7699 \\
D & $\sigma_R^2=0.2$ & 13,798.50 & 101.16 & 91.625 & +52 & 0.8102 & 0.7630 \\
D & $\sigma_R^2=0.5$ & 13,075.25 & 67.70 & 87.250 & +58 & 0.7886 & 0.7462 \\
\bottomrule\end{tabular}

\end{table*}

\begin{figure*}[!tp]
\centering
\includegraphics[width=.91\textwidth]{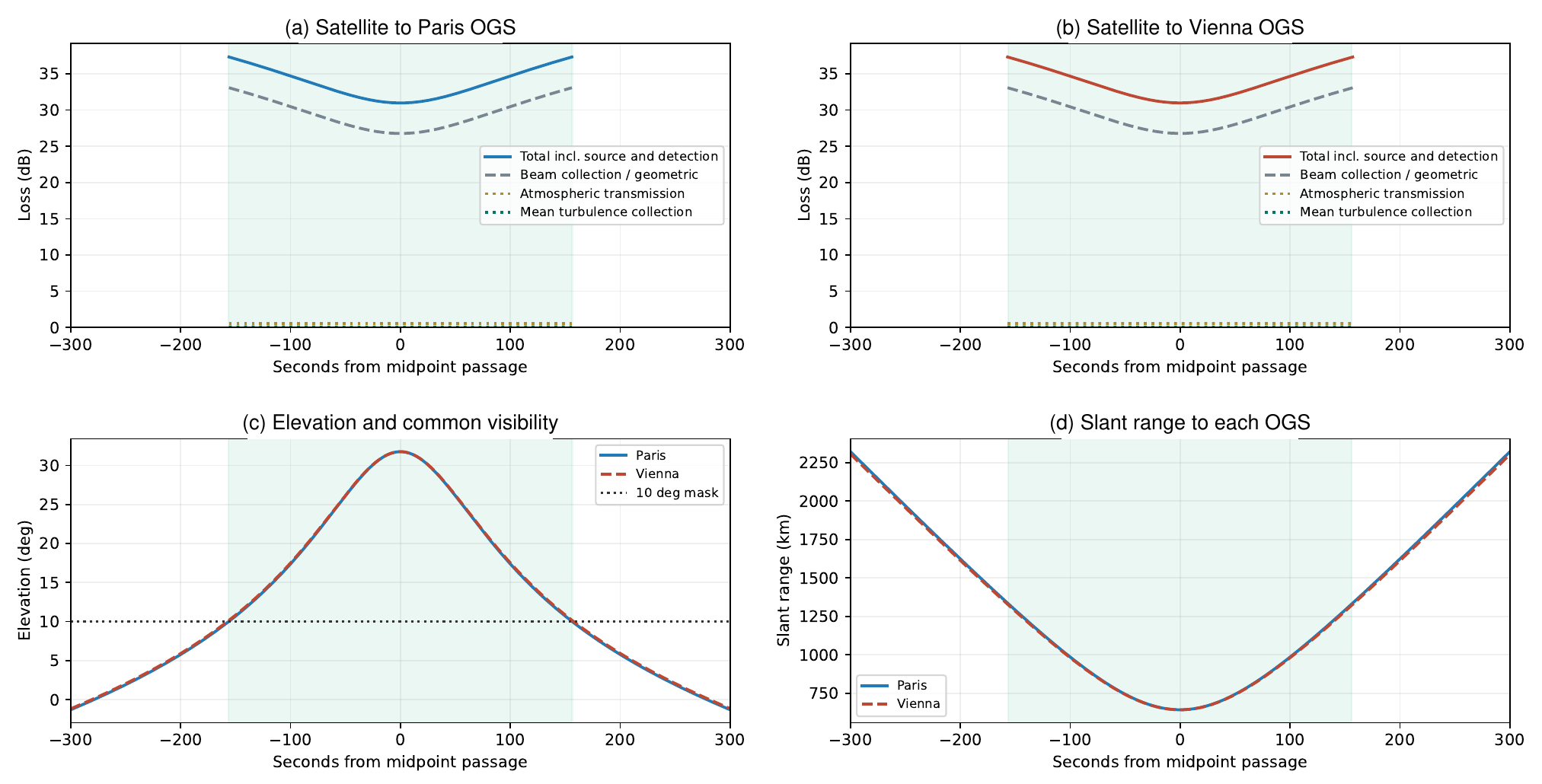}\par\vspace{3pt}
\includegraphics[width=.47\textwidth]{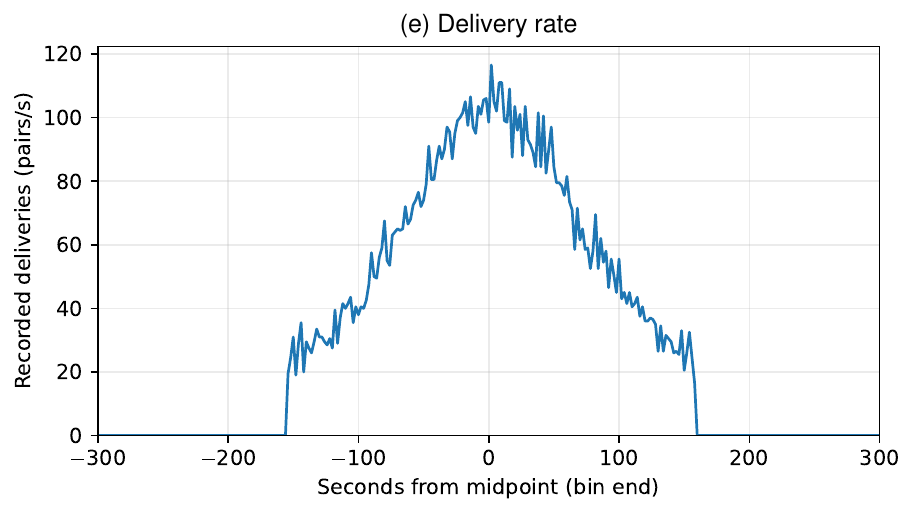}\hfill
\includegraphics[width=.47\textwidth]{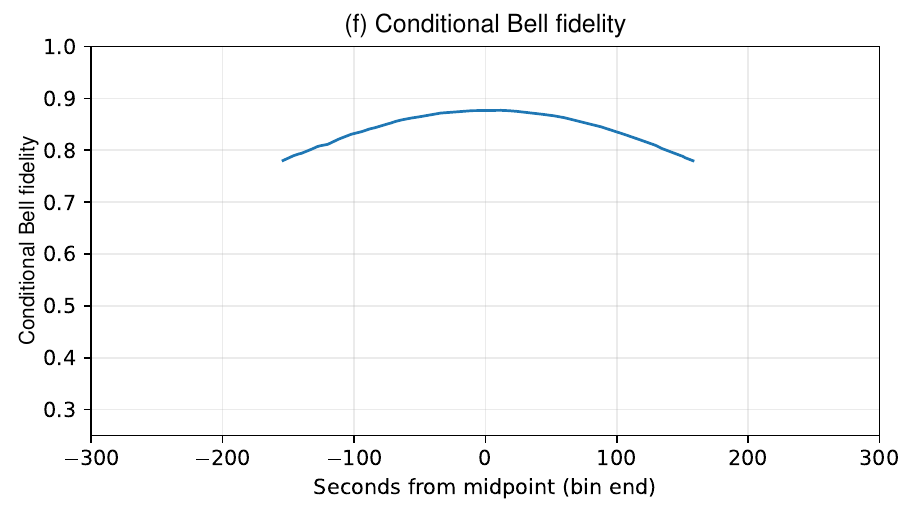}
\caption{\textbf{Dated S1 optical conditions and delivery dynamics.} (a,b) Paris and Vienna downlink budgets. (c,d) Common visibility and slant range from the geometry used for the run. (e,f) Delivery rate and conditional mean fidelity from the 2\,s counting bins. Time zero is the assumed pass midpoint. Rate and fidelity use bin-end times; attached optical geometry is sampled at bin midpoints. Fidelity gaps contain no deliveries. Bank occupancy is summarized in Section~\ref{sec:dated-s1}; the component budget and illumination context are given in Section~\ref{si:dated-context}.}
\label{fig:dated-metrics}\label{fig:dated-budgets}
\end{figure*}

Coherence requirements depend strongly on loading geometry. In case A, increasing $T_2$ from 30 to 300\,s increases mean delivery by only 1.36\%. In separated cases C and D, the corresponding factors are 3.35 and 3.55. Thus a coherence value that is already near a delivery plateau in a co-located test remains limiting when different link histories must be bridged. The large run-to-run spread at short coherence also shows why one favorable realization is an inadequate basis for a provisioning threshold.

Turbulence produces a smaller change within the particular configured range. For case A, integrated mean delivery decreases from 27\,673.5 to 26\,047.8 as $\sigma_R^2$ increases from 0.02 to 0.5; in case D it decreases from 14\,342.3 to 13\,075.3. These within-campaign sensitivities do not establish a universal ordering of turbulence and memory impairments. The complete numerical grid, including fidelity aggregation and rate maxima, is retained in Table~\ref{tab:all-settings}.

The onboard case is a capability study. The synchronization concern in Ref.~\cite{parny2023} concerns independent incoming photons without suitable space-qualified memories. Controlled recall from local storage changes the synchronization problem, but introduces memory, cooling, vacuum, optical-control, and qualification requirements. Space-borne memory proposals motivate assessing those requirements explicitly \cite{gundogan2021}. The cold-atom numerical settings are not a flight-qualified payload specification; selecting a rare-earth crystal would instead require its own cryogenic and transition-matching assumptions.

\subsubsection{Complete single-repeater parameter grid}
\label{sec:extended-results}
The full parameter grid consists of cases A--D, both control variables, four settings per variable, and four repetitions per setting. Table~\ref{tab:all-settings} places integrated yield beside maximum ensemble rate, maximum plotted median-bin fidelity, and success-count-weighted fidelity. The maximum rate and fidelity need not occur in the same time bin. Figures~\ref{fig:full-fidelity} and \ref{fig:full-rate} preserve every panel of the original grids.

Case windows are 300, 440, 340, and 440\,s, respectively. Differences between cases consequently include both loading geometry and exposure duration. The two co-located cases also use distinct seeds and windows and are not a paired test of crossing angle. The table reports between-run sample dispersion, not confidence intervals estimated from adjacent bins. Neither the median-bin fidelity nor its count-weighted counterpart reconstructs a new individual-pair threshold count.

\subsubsection{A dated Paris--Vienna pass: from optical geometry to delivered pairs}
\label{sec:dated-s1}

A companion S1 run uses Paris and Vienna city-center proxies, a circular 360\,km design orbit, and an assumed midpoint of 25 September 2026 at 17:00 UTC. A $10^\circ$ elevation mask gives approximately 312.7\,s of common visibility during a 600\,s observation. Both stations reach approximately $31.75^\circ$ elevation. This run uses 256 modes per label and $T_2=100$\,s, separating finite-bank behavior from the high-capacity controlled sweep.

Figure~\ref{fig:dated-metrics} combines the two downlink budgets with the delivery records generated for the same geometry. The per-arm mean loss near the midpoint is approximately 30.98\,dB, including source, collection, atmosphere, configured mean turbulence, static pointing, coupling, and detection. Memory recall and BSM are separate operations. Of the 300 two-second bins, 157 contain deliveries, totaling 19\,502 pairs: 32.503 pairs/s over the full observation, or 62.375 pairs/s when normalized by common visibility. The peak bin rate is 116.5 pairs/s, and success-count weighting of the bin mean fidelities gives 0.8550.

\begin{figure*}[!tp]
\centering
\includegraphics[width=.98\textwidth]{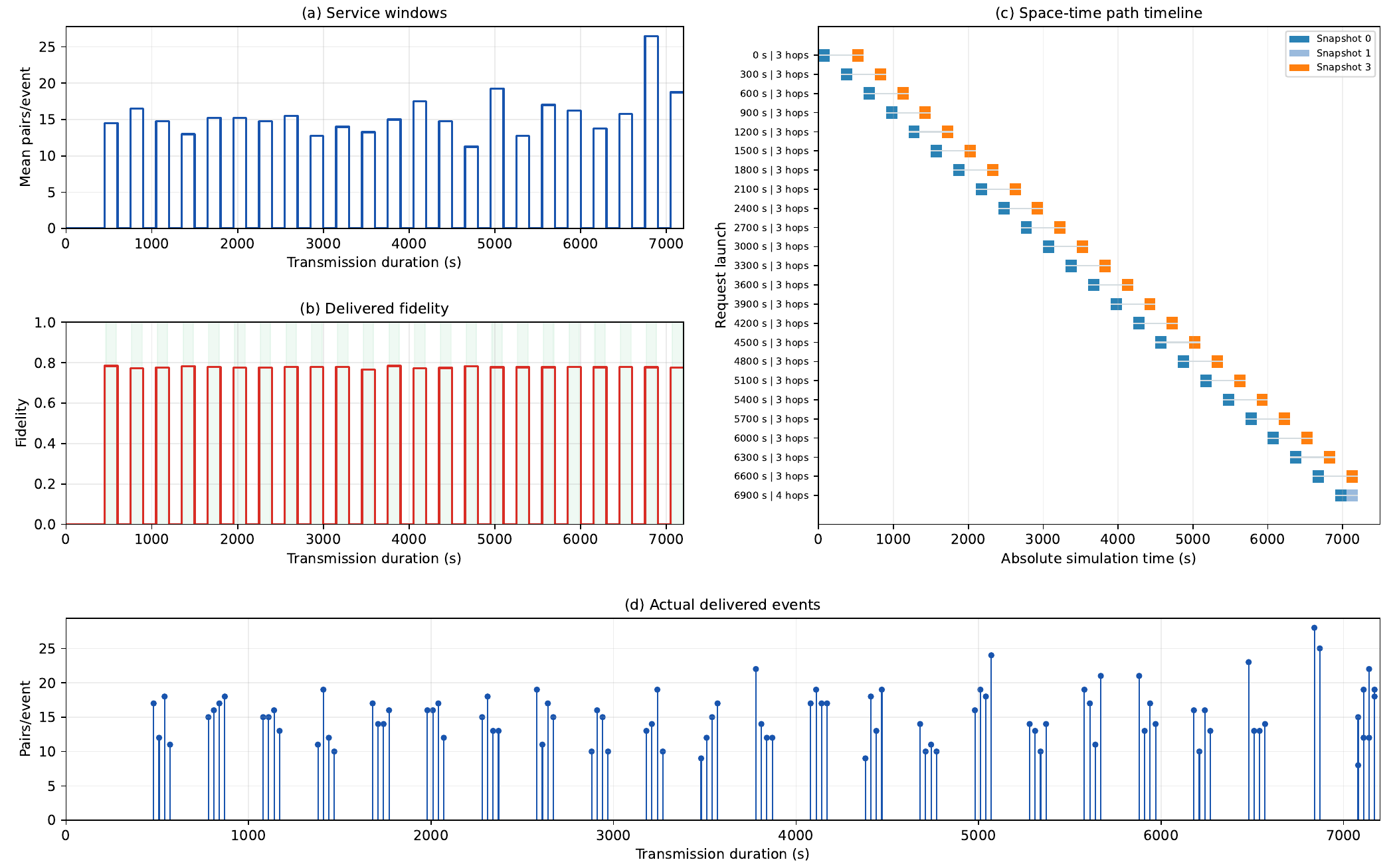}
\caption{\textbf{S2 space--time distribution between Luxembourg/Windhof and Singapore.} (a) Mean delivered pairs per event. (b) Quality during service windows; zero intervals mark no delivery. (c) Route timeline for successive requests. (d) Recorded delivery events. Snapshot duration, request spacing, planning horizon, and delivery cadence are 150, 300, 600, and 30\,s, respectively. Counts remain in their original per-event units.}
\label{fig:s2-campaign}
\end{figure*}

The peak Paris- and Vienna-side occupancies are 251 and 256 records, respectively, and the maximum combined occupancy is 321. The Vienna-side bank reaches its 256-mode capacity. At the end, the Paris-side bank is empty and the Vienna-side bank still holds 256 records. Those residuals are excluded from delivered counts. Capacity and expiry are separate controls: retained records need not be retrievable, and coherence decay does not itself release a slot. The successful-wait statistic averages 0.529\,s after count weighting.

This case shows why bank-level monitoring complements an optical link budget. Nearly balanced pass geometry does not guarantee balanced inventory throughout generation and consumption. Admission, consumption, and clearing decisions must operate at each bank, rather than only on total memory capacity.

\subsection{S2: Time-windowed service through a LEO constellation}
\label{sec:s2}
Scenario S2 connects Luxembourg/Windhof to Singapore through time-varying satellite paths. The timing configuration uses 150\,s snapshots, requests every 300\,s, a 600\,s planning horizon, and 30\,s delivery cadence. Figure~\ref{fig:s2-campaign} shows the resulting service windows, fidelity, route chronology, and individual counted events.

The route timeline spans launches from 0 to 6900\,s, mostly with three-hop routes and a final four-hop route. Quality remains near 0.78 during the plotted delivery windows. Between them, the original quality trace uses zero to indicate no delivery; it does not describe a zero-fidelity quantum state. Both count panels retain pairs-per-event units. A continuous-time throughput must instead sum these counts over a stated observation interval.

\begin{figure*}[!tp]
\centering
\begin{minipage}{.94\textwidth}
\begin{subfigure}[t]{\linewidth}
\centering
\includegraphics[width=\linewidth]{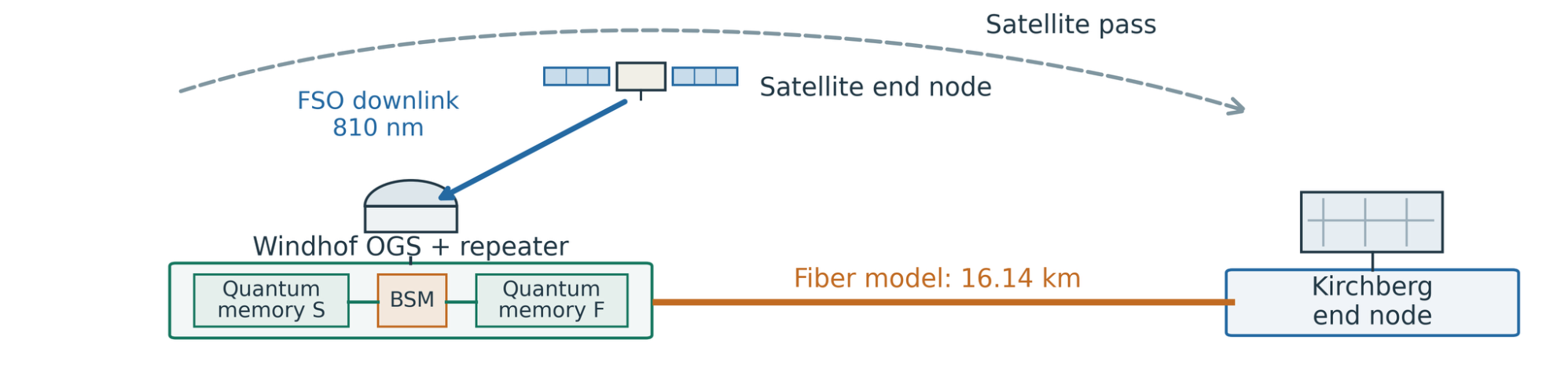}
\caption{Satellite-to-ground repeater architecture.}
\label{fig:s3-architecture-panel}
\end{subfigure}
\par\vspace{3mm}
\begin{subfigure}[c]{.56\linewidth}
\centering
\includegraphics[width=\linewidth]{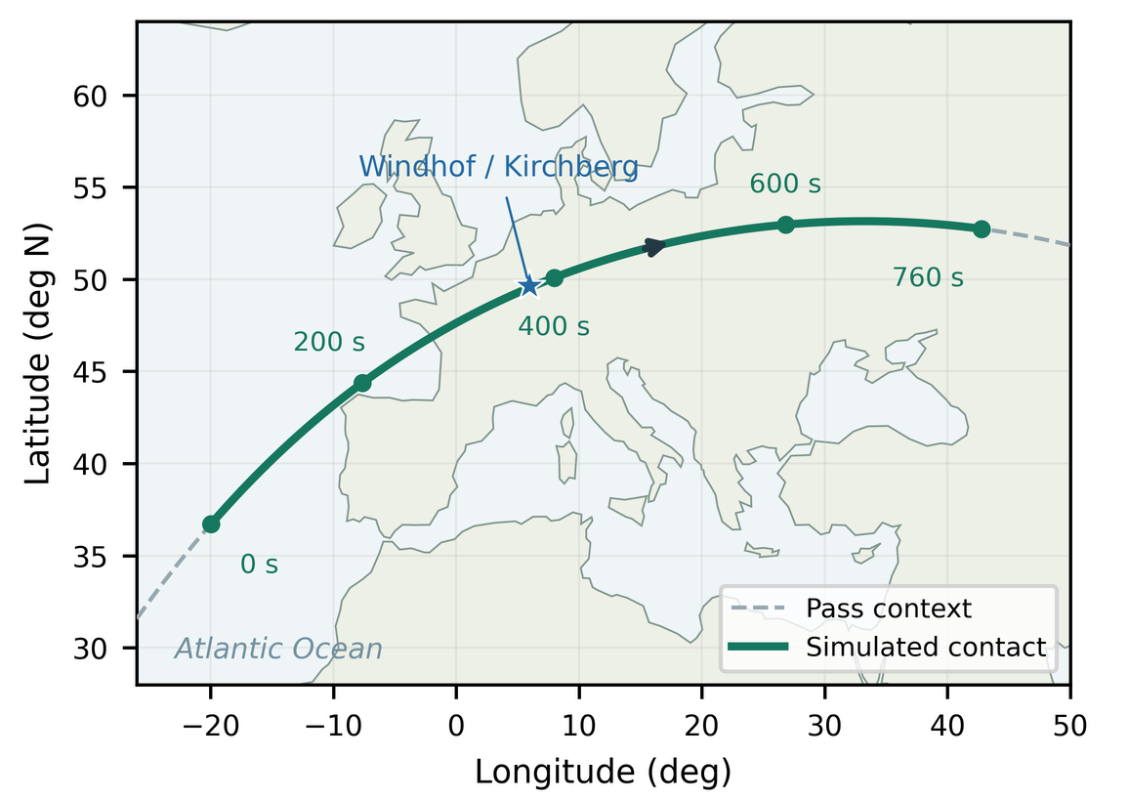}
\caption{Satellite ground track.}
\label{fig:s3-ground-track-panel}
\end{subfigure}\hfill
\begin{subfigure}[c]{.412\linewidth}
\centering
\includegraphics[width=\linewidth]{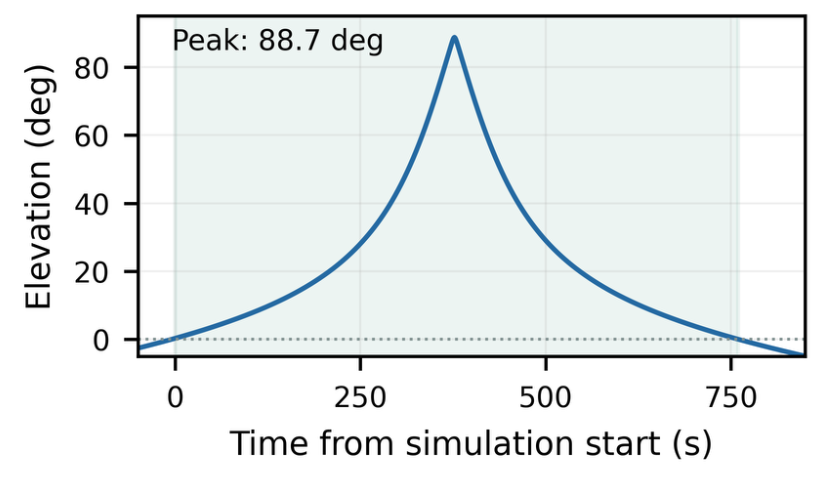}
\caption{Windhof contact geometry.}
\label{fig:s3-contact-panel}
\end{subfigure}
\end{minipage}
\caption{\textbf{S3 architecture and contact geometry.} (a) An 810\,nm satellite downlink and a 16.14\,km modeled fiber link meet at the Windhof repeater, which has separate memories and a Bell-state measurement (BSM). (b) Satellite ground track over the simulated contact. (c) Windhof elevation, reaching $88.7^\circ$. The track and times are scenario coordinates, rather than an operational acquisition schedule.}
\label{fig:s3-layout}
\end{figure*}

Connectivity becomes a service through both the contact plan and execution schedule. The four views expose this transformation more clearly than a path-availability value alone. A comparison of routing policies should hold the planning horizon, optical resources, source budgets, and traffic fixed. In particular, changing the graph horizon together with $T_2$ would mix the benefit of a larger search space with that of longer physical storage. The present trace demonstrates temporal service organization; the common policy-comparison contract is specified in Section~\ref{app:stbd}.

\subsection{S3: Ground-based swapping between a satellite downlink and fiber}
\label{sec:s3}
Scenario S3 places the repeater at Windhof, between a satellite endpoint and a terrestrial endpoint at Kirchberg. Figure~\ref{fig:s3-layout} defines the topology and its pass geometry. An 810\,nm downlink loads the satellite-facing memory, while a 16.14\,km fiber model connects the second elementary segment. A local Bell-state measurement joins the two resources. The numerical geometry has a near-zenith passage with peak elevation $88.7^\circ$ and an approximately 760\,s contact. Times are measured from simulation start.

This is a heterogeneous supply problem: the satellite segment is intermittent and strongly time dependent, whereas the fiber segment remains available with approximately constant attenuation. The two memory banks allow their elementary resources to become ready at different times before local swapping. Locating the repeater at the OGS also places receiving optics, storage controls, and spectral matching equipment on the ground. The fiber length is a simulation input, not a surveyed infrastructure length.

\begin{figure*}[!t]
\centering
\includegraphics[width=.98\textwidth]{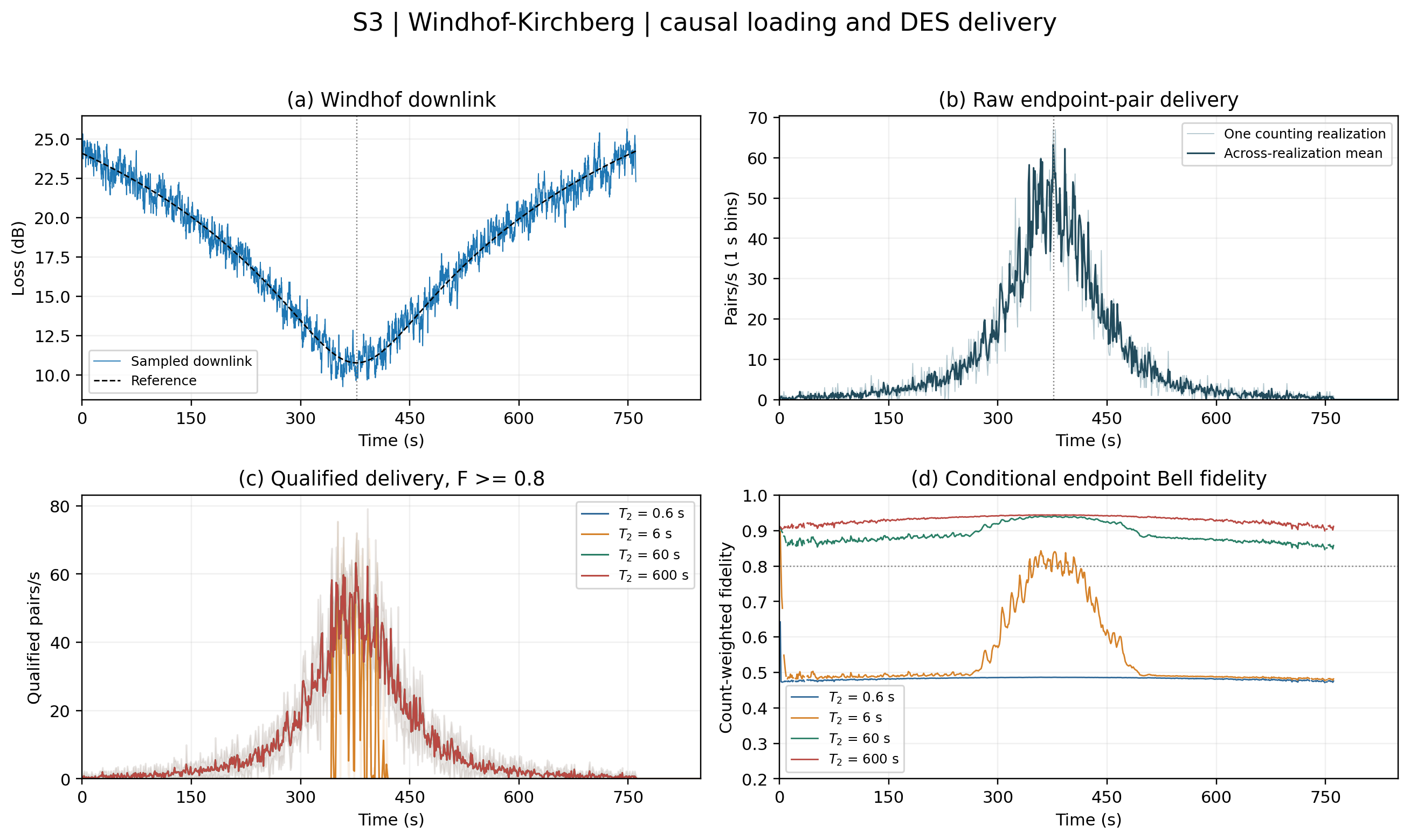}
\caption{\textbf{S3 causal loading and endpoint delivery, Part I.} (a) Sampled and reference Windhof downlink loss. (b) Raw endpoint-pair delivery in 1\,s bins, showing one counting realization and the across-realization mean. (c) Delivery satisfying $F\geq0.8$ for $T_2=0.6,6,60,$ and $600$\,s. (d) Count-weighted conditional endpoint Bell fidelity; the dotted line marks the qualification threshold. (e) Mean actual storage ages at BSM readout for the satellite-link and fiber-link OGS memories, on a logarithmic age axis. (f) Reserved and stored OGS modes in each bank. Part~II continues with cumulative qualified counts, age distributions, and resource/policy controls.}
\label{fig:s3-campaign}
\end{figure*}

Figure~\ref{fig:s3-campaign}  separates raw endpoint-pair delivery from delivery qualified by $F\geq0.8$. The downlink loss falls from approximately 24\,dB to 11\,dB near closest approach. The raw delivery envelope follows this changing link opportunity, with its maximum near 380\,s and an across-realization mean of roughly 50--60 pairs/s near the peak. The 1\,s counting realization resolves fluctuations around that envelope rather than a contact-wide constant rate.

The coherence sweep uses $T_2=0.6,6,60,$ and $600$\,s. At 0.6\,s, the conditional endpoint fidelity remains close to 0.49 and the displayed qualified yield is essentially zero. For 6\,s, fidelity exceeds 0.8 only around the middle of the pass, producing a narrow qualified-delivery window. The 60 and 600\,s cases stay above the threshold throughout the plotted contact, although their conditional fidelities differ away from closest approach. Their cumulative qualified counts nearly overlap at approximately $7.9\times10^3$ pairs, compared with approximately $2.0\times10^3$ for 6\,s. These are rounded readings of the plotted outputs.

\begin{figure*}[!tp]
\centering
\includegraphics[width=.48\textwidth]{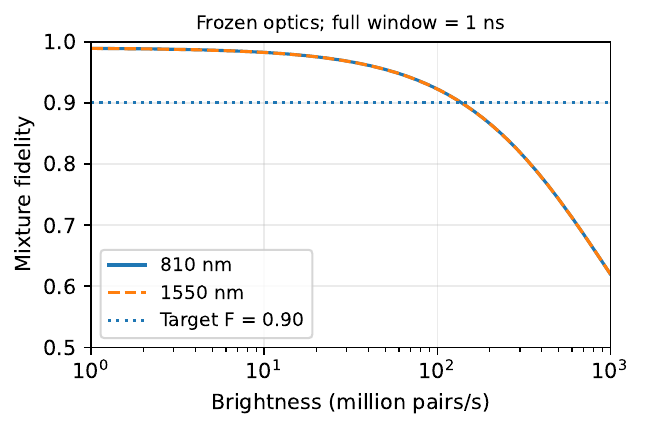}\hfill
\includegraphics[width=.48\textwidth]{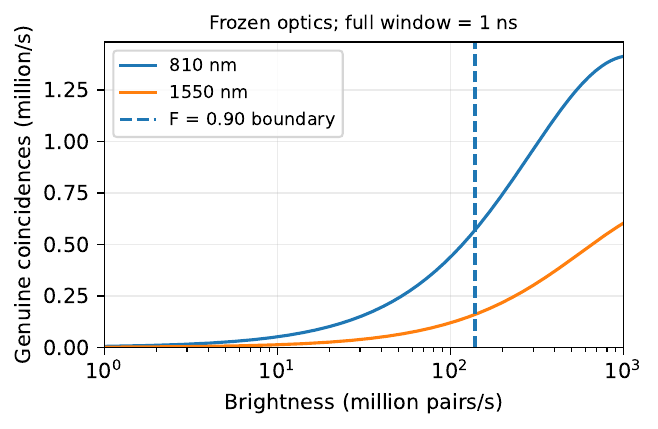}
\caption{\textbf{Direct Matera--Helmos operating-point trade-off.} Left: accepted-mixture fidelity against source brightness for a full 1\,ns relative-time window. Right: genuine-coincidence contribution with four detectors per station and a 50\,ns nonparalyzable dead-time approximation. Optical probabilities are frozen at the joint-rate peak. Equal detection efficiency and jitter are used at both wavelengths. The displayed rates are conditional receiver calculations, not stored-pair network delivery rates.}
\label{fig:direct-operation}\label{fig:brightness-fidelity}\label{fig:brightness-rate}
\end{figure*}

\subsection{Matera--Helmos: Direct distribution and a fidelity-constrained optical operating point}
\label{sec:quasar}

The Matera--Helmos case sends the two photons of one onboard entangled-pair source through separate 0.25\,m transmitters to 1.5\,m and 2.3\,m ground receivers. It provides the memory-free resource-production reference. A circular two-body orbit with semimajor axis 6939.13\,km and inclination $97.64^\circ$ gives about 293\,s of shared visibility above $20^\circ$ in the selected pass. Table~\ref{tab:quasar-inputs} states the epoch, coordinate conventions, obscurations, and receiver assumptions.

The optical calculation uses an equivalent Gaussian beam, $M^2=1$, 1\,$\mu$rad per-axis pointing uncertainty, and equal detection efficiency 0.60 at 810 and 1550\,nm. Each detector has 50\,ps RMS timing jitter. This compares wavelength-dependent propagation under equal receiver assumptions, rather than particular detector technologies. At the joint-rate peak the two-arm loss including detection is approximately 22.78\,dB at 810\,nm and 28.80\,dB at 1550\,nm. These optimistic optical assumptions are stated separately from the 12\,$\mu$rad divergence used in S1.

The conditional peak rates at a literal $10^9$ emitted pairs/s are $5.27\times10^6$ and $1.32\times10^6$ genuine coincidences/s for ideal linear detectors. That high-brightness reference is not the selected service operating point: different-pair accidentals at a 1\,ns coincidence window reduce the modeled mixture fidelity to approximately 0.62. Figure~\ref{fig:direct-operation} shows how brightness, finite detector live time, and the fidelity target jointly restrict operation.

At target fidelity 0.90 and fixed 1\,ns window, the upper allowed brightness is approximately 138.4 million pairs/s. With the four-detector, 50\,ns dead-time approximation, the genuine-coincidence contributions are approximately 570\,000/s at 810\,nm and 160\,000/s at 1550\,nm. At target 0.95, they fall to approximately 271\,000/s and 71\,000/s. Joint brightness/window optimization gives conditional 0.90-target operating points near 490 million pairs/s and 266\,ps at 810\,nm, and 534 million pairs/s and 234\,ps at 1550\,nm. The corresponding genuine contributions are approximately $1.14\times10^6$/s and $4.02\times10^5$/s. Tables~\ref{tab:target-rates} and \ref{tab:joint-optimum} retain the numerical values and assumptions.

The receiver field stop must affect signal and background consistently. In the receiver sensitivity in Methods, the same angular radius defines the encircled signal fraction and accepted sky solid angle. At 810\,nm and $20^\circ$ elevation, a 10\,$\mu$rad half-angle produces a combined field-stop penalty of approximately 0.50\,dB for the prescribed sea-level-referenced turbulence profile, but 5.17\,dB when its surface exponential is rebased to each station. This is a conditional, phase-only receiver sensitivity, not a new turbulent pass simulation. It shows why a field of view cannot be selected solely from a background count estimate; its signal acceptance must also be evaluated. Section~\ref{sec:fieldstop} gives the calculation.

The comparison with experiment is likewise conditional. Yin et al.\ measured dual-downlink attenuation of 64--82\,dB \cite{yin2017}. Their reported fine-tracking uncertainty was smaller than the pointing uncertainty assumed here, so improved pointing cannot explain the low reference loss. Divergence, aperture illumination, receiver throughput, and accepted optical modes must be compared at the same budget boundaries. For fixed small receivers in the far field, two-arm collection scales approximately as the inverse fourth power of divergence; changing 2 to 10\,$\mu$rad alone costs about 28\,dB under that approximation. The direct case therefore illustrates a coupled operating-point calculation, rather than a mission-performance advantage.


\section{Discussion}
\label{sec:discussion}

The scenarios reveal different design constraints rather than one universal memory requirement. In S3, the 60 and 600\,s coherence settings yield nearly the same number of pairs satisfying $F\geq0.8$, despite different conditional fidelities away from closest approach. The 6\,s case meets the target only near the pass peak, while the subsecond case produces essentially no qualified deliveries. By contrast, the separated-terminal S1 geometries continue to benefit strongly between 30 and 300\,s. A memory selected from a balanced, co-located test can therefore be underprovisioned for asynchronous link histories. Beyond a target-dependent count plateau, retrieval efficiency, loading balance, and BSM service time deserve comparison with further coherence improvements; S3's 50\,ms control explicitly exposes a measurement-service bottleneck. These results favor joint dimensioning of optical links, memory banks, and operating policies rather than selecting each component independently.

The dated S1 pass connects optical budgets to counted deliveries and exposes a second constraint: per-bank inventory can remain imbalanced even when the pass geometry is nearly symmetric. In S3, millisecond satellite-link storage contrasts with second-scale fiber-link storage, while the fiber-facing OGS bank remains heavily occupied. The relevant control variables are not only total capacity and generation rate, but also which bank accepts, consumes, or clears each excitation and the quality of the resources still held. At constellation scale, the timing of a viable route and the cadence of its execution jointly determine when the endpoints receive a service. Evaluating resource production and consumption on a common timeline is therefore necessary to distinguish geometrical connectivity from usable network availability.

The architecture supports this evaluation through separate physical, execution, and policy interfaces. Loss determines which resources survive, storage histories determine their quality, swaps consume their operands, and classical messages determine when an outcome is known. A richer physical model can replace a compact one at this interface without changing the definition of the delivered service. The reference instruments and detailed derivations in Methods specify this model boundary and support progressive calibration against experimental observations.

\subsection*{From simulation to testbed orchestration}

The existing APIs provide a concrete integration point for orchestration protocols in quantum-network testbeds. QNETSIM exposes an HTTP/JSON event--action interface through which external controllers receive resource and execution reports and issue scheduling commands. The implemented messages include \texttt{ENT\_REQUEST}, \texttt{REPORT\_STORAGE}, and \texttt{SWAP\_RESULT}, with corresponding actions for generation, storage, swapping, delivery, and request termination. Request and resource identifiers associate commands with their outcomes, while timestamps preserve the distinction between physical completion and controller knowledge.

This separation allows orchestration logic to be developed against simulated nodes and subsequently connected to physical-node adapters through the same logical service contract. The adapter translates device-specific observations and commands without placing hardware details inside the routing policy. A staged integration can retain simulated satellite links while incorporating a real controller or ground-node management stack, enabling repeatable tests of admission, scheduling, delayed acknowledgements, and failure recovery. Hardware-in-the-loop operation additionally requires explicit mapping between simulation time and laboratory clocks, authenticated interfaces, and device-specific timing validation. Fast optical synchronization, pointing, and wavefront correction remain local functions; the orchestration interface coordinates service-level operations rather than replacing those control loops.

\subsection*{KMS integration for satellite-assisted QKD networks}

An important extension is coupling QNETSIM to a key management system (KMS) for satellite-assisted quantum key distribution (QKD). This direction is relevant to EuroQCI's integration of terrestrial and space segments and to the ground-service interfaces required by initiatives such as EAGLE-1~\cite{euroqci2025,eagle1}. The integration objective is to connect mission-dependent key production to key availability at authorized applications, while preserving the separation between quantum-network simulation, the satellite provider's service, and cryptographic key management.

The proposed coupling uses a provider-service adapter between QNETSIM and the KMS. QNETSIM supplies contact opportunities, link conditions, and time-resolved key-production estimates from an explicitly selected QKD and post-processing model. In emulation, the adapter represents processing delays, the availability of matching key blocks at the intended endpoints, and the provider's key-identifier semantics. The KMS then manages key pooling, authorization, reservation, retrieval, and consumption. In a live deployment, the satellite or terrestrial QKD service remains authoritative for actual final-key availability; QNETSIM supplies predictions and receives operational metadata for model calibration. Coincidence counts or entanglement fidelities are not themselves secret-key blocks: secret-key accounting requires protocol-specific parameter estimation, reconciliation, privacy amplification, and the associated security assumptions~\cite{ma2007}.

Standardized interfaces provide complementary integration targets. ETSI GS QKD 014 specifies application-facing key delivery, ETSI GS QKD 015 specifies the interface between an SDN-enabled QKD node and its controller, and ETSI GS QKD 020 supports horizontal key transfer between interoperating key-management entities within a common trusted node~\cite{etsi_qkd014,etsi_qkd015,etsi_qkd020}. Device and service adapters can map QNETSIM's orchestration abstractions to these interfaces while respecting their different scopes. The orchestrator requires availability, demand, allocation status, and completion metadata, rather than access to secret key values. Test key material must remain isolated from operational cryptographic services, and provider-specific internal key-relay mechanisms can remain behind their service interfaces.

This coupling would extend mission planning from estimating optical yield to evaluating sustained key-service availability. Predicted satellite contacts can be compared with application demand and the inventories of the participating KMSs to assess reserve requirements, depletion risk between passes, request latency, and prioritization under limited supply. Delayed post-processing, weather-related contact loss, authentication costs, and trusted-relay consumption can then be included in the same operational study. Unlike stored entanglement, classical key inventories do not undergo quantum decoherence; their availability is governed by production, allocation, consumption, and security policy. These distinctions make the integration useful both for EuroQCI-oriented testbeds and for other heterogeneous satellite--terrestrial QKD networks.

\subsection*{Scope and limitations}
\label{sec:limitations}

The numerical campaigns characterize their declared configurations. The S1 runs use compact state evolution; the exact Bell-vector and general instrument formulations are physical reference models, not retrospective reinterpretations of those curves. The S3 presentation distinguishes raw and fidelity-qualified delivery, storage ages, memory occupancy, and BSM resource/policy controls. It includes one counting realization and an across-realization mean, but no confidence interval or replication count is specified for these panels. Routing-policy gains, native-engine versus independent-simulator agreement, and local versus external-controller equivalence require matched experiments beyond these outputs. The direct-link study covers one pass with assumed receiver parameters; its frozen-peak optimization and separate field-stop sensitivity are not a seasonal, detector-specific, or finite-key mission prediction. Space-qualified hardware, spectral conversion, and heralded loading require platform-specific calibration. The event--action APIs establish the software integration points; operation with physical testbed devices and a production KMS, including clock synchronization, interface conformance, and security controls, requires separate integration and validation.

QNETSIM thus connects the resource and networking stages of the satellite-quantum-internet roadmap~\cite{parny2023}, while providing an interface-based path from mission studies to testbed orchestration and key-service evaluation. Its memory-assisted scenarios show why temporal resource accounting matters, and the direct case illustrates how optical performance must be qualified by receiver statistics and state quality. Releasing configurations, counting conventions, interface definitions, and figure-generation data alongside the simulator makes these design decisions inspectable and supports their subsequent assessment in hybrid experimental networks.

\section{Methods}
\label{sec:methods}\label{page:short-methods}
The physical interfaces, derivations, numerical configurations, and verification procedures are collected here. The overview establishes the common calculation boundaries; the subsequent subsections provide the detailed memory, measurement, optical, and routing models and the records underlying the scenario results.
\subsection{Overview of physical interfaces and calculation conventions}
\label{methods:overview}

\subsubsection{Quantum memory operations and measurement models}
\label{sec:memories}\label{sec:bell-reference}
The physical reference begins with an optical encoding and an open matter--light model. Projection onto accepted output modes defines a completely positive, trace-nonincreasing memory operation
\begin{equation}
\mathcal J_s(\rho)=\sum_\mu B_\mu\rho B_\mu^\dagger,\qquad
Q_s=\sum_\mu B_\mu^\dagger B_\mu\leq\id.
\end{equation}
Its trace is the probability of successful recall. An orthogonal erasure outcome completes the operation into a trace-preserving channel. Balanced mode efficiencies and conditional pure dephasing give the reduced channel
\begin{align}
\mathcal M_\tau(\rho)&=\eta(\tau)\mathcal D_{\lambda_L(\tau)}(\rho)
+[1-\eta(\tau)]\tr(\rho)\proj{\emptyset},\label{eq:main-memory}\\
\mathcal D_{\lambda_L}(\rho)&=\frac{1+\lambda_L}{2}\rho+
\frac{1-\lambda_L}{2}Z\rho Z.
\end{align}
Here $\eta$ describes non-erasure probability and $\lambda_L$ describes surviving logical coherence. The Kraus operators are the appropriately embedded $\sqrt{\eta(1+\lambda_L)/2}\id$, $\sqrt{\eta(1-\lambda_L)/2}Z$, and $\sqrt{1-\eta}\ket{\emptyset}\bra{j}$ for $j=0,1$. After compensating a deterministic unitary, the conditional entanglement fidelity is $(1+\lambda_L)/2$, whereas the unconditional overlap is $\eta(1+\lambda_L)/2$.

Section~\ref{app:memories} develops separate reductions for cold atoms, rare-earth-ion-doped crystals, and trapped ions. Cold-atom storage starts from Maxwell--Bloch equations for the signal field, optical polarization, and spin wave, with optical depth and control pulses determining the write/read kernels \cite{fleischhauer,gorshkov_prl,gorshkov_free}. The ideal backward-CRIB expression $(1-e^{-d})^2$ belongs to its specified echo model \cite{crib_kraus,crib_sangouard}. Dynamical decoupling modifies the spin-noise filter and its coherence function \cite{dd_filter}. Trapped-ion interfaces retain finite-dimensional internal dynamics, optical coupling, and effective dissipative operators \cite{ion_memory,reiter}. A fitted $T_2$ law is distinguished from separately modeled population loss or retrieval leakage.

For monotone $\lambda_L(\tau)=e^{-\chi(\tau)}$, an already evolved state aged from $a$ to $b$ receives
\begin{equation}
\lambda_L(b\mid a)=e^{-\chi(b)+\chi(a)}.
\label{eq:main-aging}
\end{equation}
This preserves the physical age when events subdivide storage. Replacing it by $e^{-\chi(b-a)}$ resets a nonexponential history.

In the Bell-diagonal representation, let $p_g$ denote probabilities labeled by $g\in\mathbb Z_2^2$. Ideal swapping of independent elementary resources, followed by the known Pauli correction, gives
\begin{equation}
p_g^{\rm out}=\sum_h p_h^{(1)}p_{g\oplus h}^{(2)}.
\label{eq:main-bell-convolution}
\end{equation}
For Werner inputs this reduces to $F_{\rm out}=F_1F_2+(1-F_1)(1-F_2)/3$. Pure phase damping remains distinct from isotropic depolarization. Section~\ref{app:bsm} derives the click-pattern instrument, its conditioning, and the scope of compact alternatives.

For ideal indistinguishability and the stated two-of-four passive analyzer, a reduced genuine-herald probability is
\begin{equation}
p_b=\frac12\eta_{r,1}\eta_{r,2}\eta_{o,1}\eta_{o,2}\eta_{d,1}\eta_{d,2}A_t,
\label{eq:main-bsm-prob}
\end{equation}
where recall, post-recall optics, detection, and timing acceptance have separate boundaries. The factor $1/2$ reflects the restricted linear-optical analyzer \cite{bsm,bsm_limit}; any loss already sampled upstream is omitted here. General distinguishability and false-herald effects require the accepted measurement instrument rather than assigning one visibility parameter to both success and conditional fidelity.

\subsubsection{Optical channels, receiver coupling, and orbital geometry}
\label{sec:optics}
For a Gaussian beam of radius $w_L$ at a circular receiver of diameter $D_r$, geometrical collection is $1-\exp[-D_r^2/(2w_L^2)]$. Extinction, pointing, aperture collection, mode coupling, and detection define distinct accounting boundaries. Random pointing can be integrated into an effective spot size or sampled as an offset, but not counted twice. Annular receivers require integration over the clear annulus rather than a uniform-area correction to a centrally peaked beam \cite{pirandola_sat}.

The detailed atmospheric model uses the Hufnagel--Valley profile, Fried coherence length, and Rytov diagnostics with explicit elevation and altitude conventions \cite{andrews2005,fried1966}. The Gamma--Gamma description follows its normalized-irradiance formulation \cite{alhabash2001}. A unit-mean fading factor alone does not imply an additional mean attenuation. Finite-aperture propagation, wavefront correction, and accepted-mode coupling determine that loss; the field-based construction and reduced model choices are given in Section~\ref{app:channels} \cite{vasylyev,ao}.

For receiver angular acceptance $\theta$, the conditional signal fraction is $E_j(\theta)=\int_{|\boldsymbol\alpha|\leq\theta}P_j(\boldsymbol\alpha)\,\mathrm d^2\boldsymbol\alpha$, with a normalized receiving point-spread function $P_j$. The same angle sets the sky solid angle $2\pi(1-\cos\theta)$. Fiber intensity transmission is $10^{-[\alpha(\lambda_0)L+\ell_{\rm ins}]/10}$, with attenuation and insertion loss in dB. Its spectral and polarization transfer, conversion noise, and matching to the memory transition are additional interface properties.

Geometry supplies position, range, elevation, visibility, and propagation time. The dated cases state epoch, coordinate and time conventions, Earth model, propagator, and station positions. A circular design orbit and a controlled local-plane pass answer different questions; neither is relabeled as a reconstructed mission ephemeris \cite{vallado,skyfield}. Original memory-write times are independent of the snapshot bins used for routing.

\subsubsection{Space--time routing and causal resource execution}
\label{sec:distribution}
STBD represents changing connectivity by node--time vertices, communication opportunities, and memory-waiting edges \cite{stbd2024}. For declared fidelity and memory utilities,
\begin{align}
U_e&=(U_e^{\rm fid})^{\omega_f}(U_e^{\rm mem})^{\omega_m},\quad \omega_f+\omega_m=1,\\
U(P)&=\prod_{e\in P}U_e,\qquad c_e=-\ln U_e.
\end{align}
The additive cost ranks routes. A realizable plan must additionally reserve both endpoints of each stored pair, respect optical-terminal and source budgets, provide two local operands for every BSM, and propagate final quality. A communication edge creates a two-ended resource; it does not move an unknown qubit through the graph as if it were a classical packet. Heralded storage must preserve that excitation \cite{heralded_memory}.

Binary nesting joins adjacent ready segments. Segmented nesting completes snapshot-local segments early and carries surviving endpoints across later opportunities. An $\ell$-link route needs $\ell-1$ successful swaps, with failures consuming the measured operands. The full algorithm, resource-state updates, and computational boundaries appear in Section~\ref{app:stbd}. The usable-pair estimator is
\begin{equation}
R_u=\frac{1}{T_{\rm obs}}\sum_j\mathbf1\{j\text{ delivered},F_j\geq F_{\min},t_j\leq D_j\}.
\label{eq:main-service-rate}
\end{equation}
Per-link threshold checks do not replace this final boundary: four independent Werner links of fidelity 0.9 yield approximately 0.6731 after ideal swapping.

\subsubsection{Direct-link coincidence and live-time calculation}
\label{sec:direct-methods}
For CW Poisson pair brightness $R$, per-arm detection probabilities $\eta_M,\eta_H$, station noise rates $N_M,N_H$, and a full relative window $\tau$,
\begin{align}
g(\tau)&=\operatorname{erf}\!\left[\frac{\tau}{2\sqrt2\sigma_\Delta}\right],\\
C&=R\eta_M\eta_Hg(\tau),\\
A&=(R\eta_M+N_M)(R\eta_H+N_H)\tau,\\
F_{\rm mix}&=\frac{F_0C+A/4}{C+A}.
\label{eq:main-direct-mixture}
\end{align}
The accidental model assumes all-pairs association and a completely mixed effective accidental state \cite{szlachetka}. Sky counts follow entrance radiance times clear area, accepted solid angle, filter bandwidth, optical/detection efficiencies, and inverse photon energy \cite{bonato}. For target $1/4<F_\star<F_0$, the feasible brightness interval is obtained from $R^2-KR+b_Mb_H\leq0$, where $b_j=N_j/\eta_j$ and $K=[(F_0-F_\star)/(F_\star-1/4)]g/\tau-b_M-b_H$.

With four evenly loaded nonparalyzable detectors per station, the independent live fractions are
\begin{equation}
\ell_j=\left[1+(R\eta_j+N_j)t_d/4\right]^{-1},\qquad C_{\rm live}=C\ell_M\ell_H.
\end{equation}
The argument is detected singles before dead-time suppression. The optimization maximizes $C_{\rm live}$ subject to the mixture-fidelity constraint. Section~\ref{si:direct} specifies the numerical roots, timing range, equal-efficiency assumptions, and the separate field-stop sensitivity.




\subsection{Quantum memory models: field modes, open dynamics, and Kraus operators}
\label{app:memories}
\subsubsection*{Scope and conventions}
This subsection derives quantum-memory instruments for REID crystals, cold atomic ensembles, and trapped ions. Starting from the optical encoding and open matter--light dynamics, it defines accepted output modes, non-erasure probability, conditional state, and Kraus operators. Analytical retrieval and coherence laws are then identified as reduced models with stated assumptions, not unique microscopic solutions. Optical-depth conventions, independent-loss and coupled-recall alternatives, and the noise observable defining $T_2$ are kept explicit.

The optical carrier wavelength is denoted $\lambda_0$; the dimensionless logical coherence factor is $\lambda_L(\tau)$. The storage duration is $\tau$, not the orbital time. Optical depth is $d$; the number of optical interface modes is $M_{\rm int}$ and the number of simultaneously occupied memory slots is $M_{\rm cap}$. These are different quantities. A platform's native $T_2$ is the $1/e$ time of the specified coherence observable, not an automatic erasure lifetime or request timeout. The different native optical and memory parameters are highlighted in Table.5

\subsubsection*{Quantized modes and the encoded input}
For orthonormal transverse/polarization modes $u_\mu$, introduce frequency operators satisfying
\begin{equation}
[\hat a_\mu(\omega),\hat a_\nu^\dagger(\omega')]
=\delta_{\mu\nu}\delta(\omega-\omega').
\end{equation}
A normalized wavepacket and its one-photon state are
\begin{equation}
\hat a_\mu^\dagger[f]=\int \dd\omega\,f(\omega)\hat a_\mu^\dagger(\omega),
\qquad \int |f(\omega)|^2\dd\omega=1.
\label{eq:mode-packet}
\end{equation}
The positive-frequency field is an expansion in these modes,
\begin{equation}
\hat{\boldsymbol E}^{(+)}(\boldsymbol r,t)
=\sum_\mu \boldsymbol e_\mu\int\dd\omega\,
\mathcal E_\omega u_\mu(\boldsymbol r_\perp,\omega)
\hat a_\mu(\omega)e^{i[k(\omega)z-\omega t]}.
\end{equation}
The normalization factor $\mathcal E_\omega$ depends on the chosen continuum convention; its product with the transition dipole determines the coupling $g_\mu$. The logical input is encoded by the isometry
\begin{equation}
V_{\rm in}\ket j=\hat a_j^\dagger[f_j]\ket{\mathrm{vac}},
\quad j=0,1,
\quad \ket\psi=\alpha\ket0+\beta\ket1.
\end{equation}
Two polarization components, two time bins, or two rails require an interface that preserves their relative amplitude and phase. A single three-level transition, by itself, specifies only one optical coupling. An additional rail or a suitable multilevel implementation must be stated to obtain the full polarization-qubit memory. This is an encoding requirement, not extra multiplexing capacity.

\subsubsection*{Open-system evolution and the exact reduction}
Let $S$ include the accepted optical modes and the material degrees of freedom. Initially assume $\rho_{SB}=V_{\rm in}\rho V_{\rm in}^\dagger\otimes\rho_B$ and write
\begin{align}
H_{\rm tot}(t)&=H_S(t)+H_B+H_{SB},\\
H_B&=\sum_\alpha\hbar\omega_\alpha b_\alpha^\dagger b_\alpha,
\qquad H_{SB}=\sum_\alpha Q_\alpha\otimes B_\alpha.
\end{align}
The exact propagator is time ordered. For $\rho_B=\sum_\nu p_\nu\proj\nu$, a Kraus family for the reduced system evolution is
\begin{equation}
E_{\alpha\nu}(\tau)=\sqrt{p_\nu}\bra\alpha
\mathcal T e^{-i\int H_{\rm tot}(t)\dd t/\hbar}\ket\nu,
\quad \sum_{\alpha\nu}E_{\alpha\nu}^\dagger E_{\alpha\nu}=\id.
\label{eq:dilation-kraus}
\end{equation}
Preparation of the material ground state and the write/read intervals are included in this evolution. This equation is the formal field-plus-bath to Kraus connection.

When weak system--bath coupling, short bath correlations, and the secular approximation are appropriate, the reduced dynamics take the time-dependent master-equation form
\begin{align}
\dot\varrho&=-\frac{i}{\hbar}[H_S+H_{\rm LS},\varrho]
+\sum_\alpha\mathcal D[L_\alpha]\varrho,\\
\mathcal D[L]\varrho&=L\varrho L^\dagger
-\tfrac12\{L^\dagger L,\varrho\}.
\label{eq:gksl-app}
\end{align}
Classical colored detuning noise can instead be sampled in $H_S(t)$, followed by ensemble averaging. The same noise contribution must not also be included in a fitted Lindblad rate. A stretched-exponential coherence law is not evidence that the underlying bath is Markovian.

Let $V_{\rm out}$ embed the accepted logical output. Projection of the reduced evolution gives successful operators $B_r=V_{\rm out}^\dagger E_rV_{\rm in}$ and
\begin{equation}
\mathcal J_s(\rho)=\sum_r B_r\rho B_r^\dagger,
\qquad Q_s=\sum_r B_r^\dagger B_r\leq\id_2.
\label{eq:success-instrument-app}
\end{equation}
All rejected output sectors can be deliberately coarse-grained into an orthogonal flag $\ket\emptyset$. With $V$ embedding the two-dimensional accepted space into the three-dimensional flagged space, one complete Kraus family is
\begin{equation}
K_{s,r}=VB_r,\qquad
K_{f,j}=\ket\emptyset\bra j\sqrt{\id_2-Q_s},\quad j=0,1.
\label{eq:general-flagged}
\end{equation}
Indeed, the successful effects sum to $Q_s$ and the failure effects to $\id_2-Q_s$. The operational quantities are
\begin{equation}
p_s(\rho)=\tr(Q_s\rho),\qquad
\rho_{s}=\mathcal J_s(\rho)/p_s(\rho).
\end{equation}
Conditioning is undefined when $p_s=0$. The normalization is generally input dependent; the conditional transformation is a linear trace-preserving channel only in special cases, including $Q_s=\eta\id_2$.

\subsubsection*{From the field transfer kernel to a survival operator}
For a fixed classical realization $\xi$, a passive linearized, vacuum-noise memory can be written
\begin{equation}
\hat a^{\rm out}_\mu(t)=\sum_\nu\int\dd t'\,
\mathcal K_{\mu\nu}^{(\xi)}(t,t';\tau)\hat a^{\rm in}_\nu(t')
+\hat v_\mu(t).
\label{eq:memory-kernel-app}
\end{equation}
The Langevin/environment term $\hat v$ restores the output commutators and accounts for lost excitations. Projecting onto normalized incoming modes $f_\nu$ and accepted outgoing modes $g_\mu$ yields
\begin{equation}
A^{(\xi)}_{\mu\nu}(\tau)=\iint\dd t\,\dd t'\,
 g_\mu^*(t)\mathcal K_{\mu\nu}^{(\xi)}(t,t';\tau)f_\nu(t').
\label{eq:A-from-kernel}
\end{equation}
For a passive, number-preserving system--environment dilation with vacuum environment and a single input excitation, $A_\xi$ is a contraction. The successful map is
\begin{equation}
\mathcal J_s(\rho)=\E_\xi[A_\xi\rho A_\xi^\dagger],
\quad Q_s=\E_\xi[A_\xi^\dagger A_\xi].
\label{eq:average-cp}
\end{equation}
One averages the unnormalized operation and then conditions. In general $\E[A\rho A^\dagger]\neq\E[A]\rho\E[A]^\dagger$.

This reduction requires more than linear equations: the noise state and commutators must be specified. A mixture of Gaussian channels need not itself be Gaussian. A single ion is not a bosonic ensemble merely because a single-excitation calculation is linear. Added photons from four-wave mixing, fluorescence, thermal excitation, or imperfect optical pumping require further operators and detector-conditioned sectors; they are not captured by a vacuum-erasure model. The passive, vacuum-noise reduction is used only under these conditions.

\subsubsection*{Balanced-rail dephasing plus erasure: explicit Kraus construction}
Suppose the accepted operation has no rail mixing and write
\begin{align}
A_\xi&=U_{\rm det}\operatorname{diag}(a_{0,\xi},a_{1,\xi}),\\
\eta_j&=\E|a_{j,\xi}|^2,\qquad
\zeta=\E[a_{0,\xi}a_{1,\xi}^*].
\end{align}
For balanced efficiencies $\eta_0=\eta_1=\eta>0$ and a compensated mean phase, set $\lambda_L=\zeta/\eta\in[0,1]$. Then
\begin{equation}
\mathcal J_s(\rho)=\eta U_{\rm det}
\begin{pmatrix}\rho_{00}&\lambda_L\rho_{01}\\
\lambda_L\rho_{10}&\rho_{11}\end{pmatrix}U_{\rm det}^\dagger.
\end{equation}
Since $Z\rho Z$ changes the sign of the off-diagonal elements, this equals $\eta U_{\rm det}\mathcal D_{\lambda_L}(\rho)U_{\rm det}^\dagger$, where
\begin{equation}
\mathcal D_{\lambda_L}(\rho)=\frac{1+\lambda_L}{2}\rho+
\frac{1-\lambda_L}{2}Z\rho Z.
\label{eq:dephasing-app}
\end{equation}
Thus the four operators used in the balanced-rail reduced channel are
\begin{align}
K_0&=\sqrt{\eta(1+\lambda_L)/2}\,VU_{\rm det},&
K_1&=\sqrt{\eta(1-\lambda_L)/2}\,VU_{\rm det}Z,\\
K_2&=\sqrt{1-\eta}\ket\emptyset\bra0,&
K_3&=\sqrt{1-\eta}\ket\emptyset\bra1.
\label{eq:four-kraus-app}
\end{align}
Their completeness is $\sum_jK_j^\dagger K_j=\eta\id_2+(1-\eta)\id_2=\id_2$. Identifying $\lambda_L=e^{-\chi}$ gives $p_\phi=(1-e^{-\chi})/2$.

For a maximally entangled reference input, after correcting $U_{\rm det}$,
\begin{equation}
F_{e|s}=\frac{1+\lambda_L}{2},\qquad
F_e=\eta\frac{1+\lambda_L}{2},\qquad p_s=\eta.
\label{eq:memory-fidelity-app}
\end{equation}
The conditional average pure-state fidelity, a different quantity, is $F_{\rm av|s}=(2F_{e|s}+1)/3$. A phase-dephased ideal Bell pair approaches $F_{e|s}=1/2$, not $1/4$. The latter is the isotropic depolarization limit and belongs to a different channel.

\paragraph{Unequal rails.}
Cauchy--Schwarz gives $|\zeta|\leq\sqrt{\eta_0\eta_1}$. After removing the phase of $\zeta$, let $\lambda_u=|\zeta|/\sqrt{\eta_0\eta_1}$. Replace the successful operators in Eq.~\eqref{eq:four-kraus-app} by
\begin{align}
K_0&=\sqrt{(1+\lambda_u)/2}\,V\operatorname{diag}(\sqrt{\eta_0},\sqrt{\eta_1}),\\
K_1&=\sqrt{(1-\lambda_u)/2}\,V\operatorname{diag}(\sqrt{\eta_0},-\sqrt{\eta_1}),
\end{align}
and the failure factors by $\sqrt{1-\eta_j}$. Known output phases/unitaries can be appended. Now
\begin{align}
p_s(\rho)&=\eta_0\rho_{00}+\eta_1\rho_{11},\\
F_{e|s}&=\frac{\eta_0+\eta_1+2|\zeta|}{2(\eta_0+\eta_1)}
\label{eq:unbalanced-fidelity}
\end{align}
for the maximally entangled reference. A scalar efficiency is insufficient when the rails have unequal survival or couple differently to the bath.

\paragraph{Common attenuation is not logical dephasing.}
If $a_{0,\xi}=a_{1,\xi}=a_\xi$, then $\zeta=\eta$ and $\lambda_L=1$, even when $\eta\ll1$ or the common phase fluctuates. Conversely, a random relative phase can give $\lambda_L<1$ with no loss. This establishes why the retrieval envelope and conditional coherence must be separately identified before imposing any relation between them.

\subsubsection{REID spin-wave CRIB and dynamical decoupling}
\label{app:reid}
\paragraph{Material platform and matter--light Hamiltonian.}
Consider Pr$^{3+}$:Y$_2$SiO$_5$ with an optical interface near $605.977$\,nm, controlled reversible inhomogeneous broadening (CRIB), and a long-lived hyperfine spin wave protected by a clock/ZEFOZ operating point and dynamical decoupling. The storage level $\ket s$ and transfer control field connect optical excitation to the spin-wave stage:
\begin{align}
H/\hbar={}&H_{\rm field}/\hbar
+\sum_j[\Delta_j(t)\sigma_{ee}^{(j)}+\delta_j(t)\sigma_{ss}^{(j)}]
\nonumber\\
&-\sum_j\left[\sum_\mu g_{\mu j}\hat a_\mu(z_j)\sigma_{eg}^{(j)}
+\Omega_j(t)\sigma_{es}^{(j)}+\mathrm{h.c.}\right],\\
\Delta_j(t)&=\Delta_j^{(0)}+\Delta_{\rm CRIB}(z_j,t)+\delta\Delta_j(t).
\end{align}
Here $H_{\rm field}/\hbar=\sum_\mu\int\dd z\,\hat a_\mu^\dagger(-ic\partial_z)\hat a_\mu$. The sign of the coupling is a phase convention. Optical spontaneous decay, spin leakage, and fluctuating detunings enter through Eq.~\eqref{eq:gksl-app} or an explicit bath. With most ions in $\ket g$, the optical polarization obeys the linear optical Maxwell--Bloch equations; retaining $\Omega$ adds the spin-wave equation used below for ensembles. Immobility removes the atomic-diffusion term but not spectral diffusion or magnetic noise.

\paragraph{Optical depth and backward-echo amplitude.}
For ideal, phase-matched backward CRIB, complete detuning reversal, adequate prepared bandwidth, and lossless spin-transfer pulses, an excitation created at depth $z$ has input and return amplitude weights $e^{-\alpha z/2}$. The coherent source contribution is proportional to $\alpha\dd z$. Consequently the normalized retrieval amplitude is
\begin{equation}
A_{\rm CRIB}^{(0)}=\int_0^L\alpha e^{-\alpha z}\dd z=1-e^{-d},
\quad d=\alpha L,
\end{equation}
and
\begin{equation}
\eta_{\rm CRIB}^{(0)}(d)=(1-e^{-d})^2.
\label{eq:crib-eff-app}
\end{equation}
This is the ideal backward-CRIB limit under the stated assumptions \cite{crib_kraus,crib_sangouard}. Imperfect preparation, irreversible optical dephasing, insufficient bandwidth, imperfect transfer pulses, and collection must enter the actual kernel or additional efficiencies. Increasing $d$ does not remove these errors. Atomic-frequency-comb storage is a distinct echo protocol, not another name for CRIB \cite{afc}.

\paragraph{From the spin bath to the DD coherence factor.}
In a logical toggling frame, use $H_\phi(t)=\hbar y(t)\delta\omega(t)Z/2$, where ideal $\pi$ pulses change $y(t)$ between $+1$ and $-1$. For one realization,
\begin{equation}
\rho_{01}(\tau)=e^{-i\phi(\tau)}\rho_{01}(0),\qquad
\phi(\tau)=\int_0^\tau y(t)\delta\omega(t)\dd t.
\end{equation}
For zero-mean Gaussian noise with correlation $C(t-t')$, the Gaussian characteristic function gives
\begin{equation}
\lambda_L(\tau)=e^{-\chi(\tau)},\quad
\chi=\frac12\int_0^\tau\!\int_0^\tau y(t)y(t')C(t-t')\dd t\dd t'.
\label{eq:chi-time}
\end{equation}
With the two-sided spectrum $S_2(\omega)=\int_{-\infty}^{\infty}C(u)e^{-i\omega u}\dd u$ and $Y(\omega)=\int_0^\tau y(t)e^{i\omega t}\dd t$,
\begin{equation}
\chi=\frac{1}{4\pi}\int_{-\infty}^{\infty}S_2(\omega)|Y(\omega)|^2\dd\omega.
\label{eq:chi-psd}
\end{equation}
For a classical real process, the positive-frequency form is $\chi=(2\pi)^{-1}\int_0^\infty S_2(\omega)|Y(\omega)|^2\dd\omega$. An equivalent convention defines $F(\omega\tau)=\omega Y(\omega)$ and $S_F=S_2/2$, giving $\chi=\pi^{-1}\int_0^\infty S_F|F|^2/\omega^2\dd\omega$. The spectral definition must accompany the prefactor \cite{dd_filter}.

At a ZEFOZ point the first derivative of the transition frequency with respect to magnetic field vanishes. Residual quadratic magnetic shifts, pulse errors, and technical noise remain. A fitted coherence and efficiency model
\begin{equation}
\chi_R(\tau)=\left(\frac{\tau}{T_{2,R}^{\rm DD}}\right)^{n_R},
\quad \eta_{\rm DD}=e^{-(\tau/T_{\rm mem})^{m_R}}
\label{eq:reid-closure}
\end{equation}
therefore replaces a specified noise spectrum and pulse sequence over a declared fitting interval. $m_R$ is the efficiency-envelope exponent, not the mode capacity. Non-Gaussian noise may require higher cumulants than Eq.~\eqref{eq:chi-time}.

\paragraph{Independent-loss and coupled-recall variants.}
The independent-loss variant uses
\begin{equation}
\eta_R^{(0)}(\tau)=\eta_{\rm CRIB}^{(0)}(d)\eta_{\rm DD}(\tau),
\label{eq:reid-q0}
\end{equation}
whereas the coupled-recall variant uses
\begin{equation}
\eta_R^{(2)}(\tau)=\eta_R^{(0)}(\tau)e^{-2\chi_R(\tau)}.
\label{eq:reid-q2}
\end{equation}
Both are retained as explicitly different effective models. Each gives a CPTP instrument when inserted with $\lambda_L=e^{-\chi_R}$ into Eq.~\eqref{eq:four-kraus-app}. Only the first is implied by the scalar erasure-plus-dephasing form with independent envelopes. The second additionally assumes that recoverable collective-mode population decays with the square of the same fitted coherence amplitude. That extra identification needs a microscopic correlation model or calibration; the general field reduction does not impose it. Optional population loss adds $e^{-\tau/T_{\rm loss}}$ when it is genuinely leakage out of both logical rails. The same loss must not be charged again at recall.

\subsubsection{Cold-atom EIT/Raman memory}
\label{app:cold}
\paragraph{Hamiltonian, linearization, and noise.}
For an atomic $\Lambda$ system ($^{87}$Rb as its example), replace the REID detunings by a one-photon detuning $\Delta$ and two-photon detunings $\delta_j(t)$, and retain the quantum coupling and classical control $\Omega(t)$. Define coarse-grained collective polarization $P\propto\sum_j\sigma_{ge}^{(j)}$ and spin wave $S\propto\sum_j\sigma_{gs}^{(j)}$. Weak excitation lets the population differences be replaced by their ground-state values; an appropriate collective normalization gives
\begin{align}
(\partial_t+c\partial_z)\hat{\mathcal E}&=iG\hat P,\\
\partial_t\hat P&=-(\gamma+i\Delta)\hat P+iG\hat{\mathcal E}+i\Omega\hat S+\hat F_P,\\
\partial_t\hat S&=-(\gamma_s+i\delta)\hat S+i\Omega^*\hat P+\hat F_S.
\label{eq:MB-cold}
\end{align}
Here $\gamma$ is an \emph{amplitude} decay rate, and $G$ is the collective coupling in the selected normalization. The reservoir correlations, not the mere presence of $\hat F_P$, determine whether added noise is negligible. EIT and off-resonant Raman storage are different operating limits of these equations, although an ideal optimization can relate their maximum efficiencies \cite{fleischhauer,gorshkov_prl}.

\paragraph{A factor-of-two optical-depth convention.}
Set $\Omega=\Delta=0$ and eliminate the stationary polarization: $P=iG\mathcal E/\gamma$. Then
\begin{equation}
\partial_z\mathcal E=-\frac{G^2}{c\gamma}\mathcal E,
\quad \frac{I(L)}{I(0)}=e^{-2G^2L/(c\gamma)}.
\end{equation}
Matching the intensity optical-depth convention $I(L)/I(0)=e^{-d}$ gives $d=2G^2L/(c\gamma)$. The free-space solution of Gorshkov \emph{et al.}\ uses $d_G=G^2L/(c\gamma)$, so $d_G=d/2$ for these normalizations \cite{gorshkov_free}. This conversion must precede importing a storage kernel. Otherwise even a numerically exact kernel can be evaluated at the wrong physical optical depth.

\paragraph{EIT storage as a change of field character.}
On two-photon resonance with slowly varying, fixed-phase control, the dark excitation has negligible optical polarization:
\begin{equation}
\hat\Psi=\cos\vartheta\,\hat{\mathcal E}-\sin\vartheta\,\hat S,
\quad \tan\vartheta=G/|\Omega|,
\quad v_g=c\cos^2\vartheta.
\end{equation}
Reducing $|\Omega|$ converts the propagating dark excitation to a material spin wave. Reversing the control converts it back. Pulse bandwidth, finite optical depth, nonadiabatic coupling to bright states, and spin decoherence limit the efficiency; a dark-state expression alone does not determine the write/read probability.

\paragraph{Explicit kernel derivation in the adiabatic limit.}
Introduce retarded dimensionless time $s=\gamma(t-z/c)$, coordinate $x=z/L$ for uniform density, $\bar\Delta=\Delta/\gamma$, and $\bar\Omega=\Omega/\gamma$. Neglect spin decay only during the short write/read step. With normalized mode amplitudes the equations become
\begin{align}
\partial_x\mathcal E&=i\sqrt{d_G}P,\\
\partial_sP&=-(1+i\bar\Delta)P+i\sqrt{d_G}\mathcal E+i\bar\Omega S,\\
\partial_sS&=i\bar\Omega^*P.
\end{align}
Laplace transforming $x\mapsto u$ yields
\begin{align}
\mathcal E(u,s)&=i\sqrt{d_G}P(u,s)/u+\mathcal E_{\rm in}(s)/u,\\
D(u)&=1+i\bar\Delta+d_G/u,\\
\partial_sS&=-\frac{|\bar\Omega|^2}{D(u)}S
-\frac{\bar\Omega^*\sqrt{d_G}}{uD(u)}\mathcal E_{\rm in},
\label{eq:laplace-write}
\end{align}
where the last equation uses adiabatic elimination of $P$. This scalar inhomogeneous equation is solved with an integrating factor. Define $h(s_1,s_2)=\int_{s_1}^{s_2}|\bar\Omega(v)|^2\dd v$ and $q=1+i\bar\Delta$. The inverse Laplace transform gives the write kernel
\begin{align}
S(x,s_w)&=\int_0^{s_w} W(x,s)\mathcal E_{\rm in}(s)\dd s,\\
W(x,s)&=-\frac{\sqrt{d_G}\bar\Omega^*(s)}{q}
 e^{-[h(s,s_w)+d_Gx]/q}\nonumber\\
 &\quad\times I_0\!\left(\frac{2\sqrt{h(s,s_w)d_Gx}}{q}\right).
\label{eq:write-kernel}
\end{align}
Here $I_0$ is the modified Bessel function. The corresponding forward read kernel, with read time reset to zero and $S_r(x)$ the initial spin-wave profile, is
\begin{align}
\mathcal E_{\rm out}(s)&=\int_0^1 R(s,x)S_r(x)\dd x,\\
R(s,x)&=-\frac{\sqrt{d_G}\bar\Omega(s)}{q}
 e^{-[h(0,s)+d_G(1-x)]/q}\nonumber\\
 &\quad\times I_0\!\left(\frac{2\sqrt{h(0,s)d_G(1-x)}}{q}\right).
\label{eq:read-kernel}
\end{align}
A global phase convention is immaterial to efficiencies. Backward readout reverses the spin-wave coordinate and must also satisfy the wavevector-matching condition. These are free-space write and read kernels under the boundary and adiabatic assumptions just stated \cite{gorshkov_free}.

Let $D_{\rm sw}(x,x';\tau)$ describe storage evolution. The full write--store--read mode kernel is now explicitly
\begin{equation}
\mathcal K(s,s';\tau)=\int_0^1\!\int_0^1
R(s,x)D_{\rm sw}(x,x';\tau)W(x',s')\dd x\dd x'.
\label{eq:WSR}
\end{equation}
Projection by Eq.~\eqref{eq:A-from-kernel} and the open-system average by Eq.~\eqref{eq:average-cp} complete the path to the Kraus operators. A solver can precompute $\mathcal K$ or its singular modes offline; DES events need only the reduced instrument and its parameter identifier.

For a normalized initial spin wave, completed retrieval without spin decay also admits the independent efficiency check
\begin{align}
\eta_r&=\int_0^1\!\int_0^1 k_{d_G}(x,y)S_r(1-x)S_r^*(1-y)\dd x\dd y,\\
k_{d_G}(x,y)&=\frac{d_G}{2}e^{-d_G(x+y)/2}
 I_0(d_G\sqrt{xy}).
\label{eq:retrieval-kernel-check}
\end{align}
It follows by integrating the polarization/spin excitation-balance equation over readout time. For a flat unit-norm spin wave,
\begin{equation}
\eta_{r,\rm flat}=1-e^{-d_G}[I_0(d_G)+I_1(d_G)].
\label{eq:flat-retrieval}
\end{equation}
This is retrieval from a prescribed spin wave, not the total write--read efficiency. The total efficiency includes $\eta_w=\int_0^1|S(x,s_w)|^2\dd x$ and the overlap of the retrieved field with the accepted temporal mode. A tabulated $d$ alone cannot set all three.

\paragraph{Parametric efficiency and kernel-based efficiency.}
The compact law $\eta_{\rm EIT}(d)\simeq(1-e^{-d})^2$ is a \emph{parametric} backward-recall surrogate, not a general consequence of EIT. Equations~\eqref{eq:write-kernel}--\eqref{eq:flat-retrieval} define a kernel-based alternative. A simulation must select one efficiency model; the two are not interchangeable at fixed $d$.

\paragraph{Motion, magnetic inhomogeneity, and $T_2$.}
The collective spin wave has wavevector $\boldsymbol q=\boldsymbol k_s-\boldsymbol k_c$. A displacement $\Delta\boldsymbol r$ contributes $e^{i\boldsymbol q\cdot\Delta\boldsymbol r}$ to its retrieval overlap. Under Gaussian diffusion with $\E[\Delta r_a\Delta r_b]=2D\tau\delta_{ab}$,
\begin{equation}
C_{\rm mot}(\tau)=\E e^{i\boldsymbol q\cdot\Delta\boldsymbol r}
=e^{-D|\boldsymbol q|^2\tau}.
\end{equation}
Independent quasistatic two-photon detuning of variance $\sigma_\delta^2$ contributes $e^{-\sigma_\delta^2\tau^2/2}$; homogeneous spin decay contributes $e^{-\gamma_0\tau}$. The combined functional is therefore
\begin{equation}
\chi_C(\tau)=a\tau+b\tau^2,
\quad a=D|\boldsymbol q|^2+\gamma_0,\quad b=\sigma_\delta^2/2.
\label{eq:cold-chi}
\end{equation}
Its $1/e$ time obeys $aT_2+bT_2^2=1$, giving the stable expression
\begin{equation}
T_{2,C}=\frac{2}{a+\sqrt{a^2+4b}}.
\label{eq:cold-t2}
\end{equation}
For $a=b=0$, the coherence time is infinite. For ballistic thermal motion, the different model $C_{\rm mot}=\exp[-|\boldsymbol q|^2(k_BT/m_a)\tau^2/2]$ replaces the diffusion factor. A trapping model may require neither limit. These alternatives must not be superimposed for the same motional mechanism.

An alternative fitted model sets $\chi_C=(\tau/T_{2,C})^{n_C}$ and, for the coupled-recall variant,
\begin{equation}
\eta_C^{(2)}(\tau)=\eta_{\rm EIT}(d)\eta_{\rm mode}(\tau)e^{-2\chi_C(\tau)},
\quad \lambda_{L,C}=e^{-\chi_C(\tau)}.
\label{eq:cold-coupled}
\end{equation}
This coupled model identifies optical-mode survival and logical dephasing through the same functional. For two physical rails, that identification is an assumption about their noise correlations. In a kernel backend, motional loss of collective retrieval and differential logical noise are extracted from $\eta_j$ and $\zeta$. The mode factor $\eta_{\rm mode}=\tr(M^\dagger M)/M_{\rm int}$ is an average over a specified equally weighted input ensemble, not a universal state-independent efficiency when $M^\dagger M$ is non-scalar.

\paragraph{Native-parameter consistency test.}
Microscopic parameters and a fitted coherence time cannot be varied independently for the same mechanism. For example, $D=10^{-4}$--$10^{-3}$\,m$^2$/s and $|\boldsymbol q|=10^4$--$10^6$\,m$^{-1}$ would place $(D|\boldsymbol q|^2)^{-1}$ between $10^{-9}$ and $10^{-4}$\,s under the stated diffusion law. They cannot simultaneously describe millisecond-to-second motional coherence without changing the motion or signal--control geometry model. A run therefore uses either a consistent microscopic parameter set or a calibrated effective $T_2,n_C$ pair.

\subsubsection{Trapped-ion Raman interface and internal-state memory}
\label{app:ion}
\paragraph{Finite-dimensional matter dynamics.}
A single-ion $\Lambda$ interface couples optical excitation to a long-lived internal state through a Raman interaction. For a cavity-assisted realization, the optical input/output boundary is specified by
\begin{align}
H/\hbar={}&\delta_c a^\dagger a+\Delta\proj e+\delta\proj s\nonumber\\
&+\left(g a\ket e\bra g+\Omega(t)\ket e\bra s+\mathrm{h.c.}\right),\\
L_{\rm ex}&=\sqrt{\kappa_{\rm ex}}a,\quad L_{\rm loss}=\sqrt{\kappa_{\rm loss}}a,\\
L_g&=\sqrt{\Gamma_g}\ket g\bra e,\quad
L_s=\sqrt{\Gamma_s}\ket s\bra e.
\end{align}
Here $\kappa=\kappa_{\rm ex}+\kappa_{\rm loss}$ and $\Gamma=\Gamma_g+\Gamma_s$ are \emph{population} decay rates. The associated amplitude linewidths are $\kappa/2$ and $\Gamma/2$. Additional dark-state noise and leakage operators are specified separately. A polarization memory needs the multilevel/dual-rail extension identified above. No collective bosonic approximation is made for the ion.

\paragraph{An explicit single-photon capture and readout kernel.}
The incoming continuum wavepacket can be reduced to a single-excitation amplitude equation for a fixed classical-noise realization. Let $\boldsymbol v=(c,e,s)^T$ denote the cavity-photon, optical-excited-state, and stored-state amplitudes. The coherent/no-jump branch obeys
\begin{equation}
\dot{\boldsymbol v}=\mathsf A(t)\boldsymbol v+
\sqrt{\kappa_{\rm ex}}\boldsymbol e_c f_{\rm in}(t),
\label{eq:ion-amplitude-ode}
\end{equation}
where $\boldsymbol e_c=(1,0,0)^T$, $\boldsymbol e_s=(0,0,1)^T$, and
\begin{equation}
\mathsf A(t)=
\begin{pmatrix}
-\kappa/2-i\delta_c&-ig^*&0\\
-ig&-\Gamma/2-i\Delta&-i\Omega(t)\\
0&-i\Omega^*(t)&-\gamma_\ell/2-i\delta(t)
\end{pmatrix}.
\end{equation}
Here $\gamma_\ell$ is population leakage out of the selected storage state, not a relabeling of logical $1/T_2$. Input/output normalization gives $f_{\rm out}=f_{\rm in}-\sqrt{\kappa_{\rm ex}}c$. For a normalized incident photon and initially empty cavity/unexcited ion, let $\mathsf U_w(t,t')$ be the time-ordered propagator of $\mathsf A$ during writing. The coherent storage amplitude is
\begin{equation}
s_w=\sqrt{\kappa_{\rm ex}}\int_0^{t_w}
\boldsymbol e_s^T\mathsf U_w(t_w,t')\boldsymbol e_c
f_{\rm in}(t')\dd t',\qquad \eta_w^{\rm coh}=|s_w|^2.
\end{equation}
For retrieval from the unit storage excitation with vacuum optical input,
\begin{equation}
f_r(t)=-\sqrt{\kappa_{\rm ex}}\boldsymbol e_c^T
\mathsf U_r(t,0)\boldsymbol e_s,\quad
\eta_r^{\rm coh}=\int|f_r(t)|^2\dd t.
\end{equation}
The accepted read-mode amplitude is $\int g_{\rm out}^*(t)f_r(t)\dd t$, rather than $\sqrt{\eta_r^{\rm coh}}$ unless the modes match. Multiplying the write amplitude, the storage propagator, and this read amplitude supplies $a_{j,\xi}$ for each logical rail and hence the quantities $\eta_j,\zeta$ in the Kraus reduction. This makes the dependence on $g,\Omega,\Delta,\kappa,\Gamma$, and the field mode explicit instead of treating $\eta_{\rm TI}$ as a function of $T_2$ alone.

The amplitude equations describe the coherent/no-jump contribution. Spontaneous jumps that return population to an accepted internal state can produce incoherent, unheralded errors; their trajectories or full master-equation contributions must be added through Eq.~\eqref{eq:success-instrument-app}. They may be classified as erasure only when the specified herald rejects them. The equation is therefore a physical kernel construction.

\paragraph{Eliminating the excited state without discarding dissipation.}
Write all coherent operators in frequency units, $h=H/\hbar$. Split the excitation and de-excitation couplings as $v_+$ and $v_-=v_+^\dagger$, and define $h_{\rm NH}=h_e-i\sum_jL_j^\dagger L_j/2$. Adiabatic elimination gives
\begin{align}
h_{\rm eff}&=h_g-\tfrac12v_-\left[h_{\rm NH}^{-1}+(h_{\rm NH}^{-1})^\dagger\right]v_+,\\
L_{{\rm eff},j}&=L_jh_{\rm NH}^{-1}v_+.
\label{eq:effective-ion}
\end{align}
For $h_{\rm NH}=(\Delta-i\Gamma/2)\proj e$ and $v_+=\ket e(g a\bra g+\Omega\bra s)$, the coherent part contains
\begin{align}
H_{\rm eff}/\hbar\supset-
\frac{\Delta}{\Delta^2+(\Gamma/2)^2}
\big[&|g|^2a^\dagger a\proj g+|\Omega|^2\proj s\nonumber\\
&+g^*\Omega a^\dagger\ket g\bra s+\mathrm{h.c.}\big],
\end{align}
where the Hermitian-conjugate term applies to the off-diagonal coupling. The effective spontaneous-emission operators include
\begin{equation}
L_{{\rm eff},g}=
\frac{\sqrt{\Gamma_g}}{\Delta-i\Gamma/2}
\ket g(g a\bra g+\Omega\bra s),
\end{equation}
and the analogous $\ket s$ channel. Thus the large-detuning limit gives the Raman coupling magnitude $|g\Omega/\Delta|$, together with AC Stark shifts and scattering rates proportional to $\Gamma|\Omega/\Delta|^2$. Retaining only the coherent coupling would omit precisely the open-system errors that the memory model is meant to capture \cite{reiter}.

\paragraph{Cooperativity rather than ensemble optical depth.}
With the above linewidth convention, define
\begin{equation}
C=\frac{4|g|^2}{\kappa\Gamma}.
\end{equation}
In a resonant bad-cavity limit with the required bandwidth assumptions, cavity elimination gives a desired-mode emission rate $\Gamma_c=4|g|^2/\kappa$. The branching fraction into the external port is then
\begin{equation}
\beta_{\rm ex}=\frac{\kappa_{\rm ex}}{\kappa}
\frac{\Gamma_c}{\Gamma+\Gamma_c}
=\frac{\kappa_{\rm ex}}{\kappa}\frac{C}{1+C}.
\label{eq:ion-beta}
\end{equation}
This is a limit derived for a specified cavity, not a universal write--read efficiency. The actual $\eta_{\rm TI}$ is obtained from capture, control-pulse, retrieval, spectral, and spatial-mode matching. The incoming photon can be treated as a one-photon wavepacket in a continuum or a source-cavity dilation, and Eq.~\eqref{eq:dilation-kraus} then produces the corresponding instrument. Free-space collection uses its optical-mode overlap instead of cavity $C$. An effective optical coupling parameter must not be interpreted simultaneously as a dimensionless cooperativity and a collection probability, and an ensemble optical depth $d$ is not an intrinsic single-ion parameter.

\paragraph{Dark-storage coherence and effective retrieval.}
During storage the control is off unless the protocol explicitly applies optical fields. For an additive effective rate model,
\begin{equation}
\frac1{T_{2,I}}=\gamma_0+\Gamma_B+|\Delta k|^2D_{\rm mot}+\alpha_\phi\gamma_{\rm sc},
\quad \lambda_{L,I}=e^{-\tau/T_{2,I}}.
\end{equation}
Off-resonant optical scattering contributes only over intervals in which the relevant light is present. A Ramsey/echo/DD measurement can instead provide the effective coherence law \cite{ion_memory}. A coupled-recall closure additionally assumes
\begin{equation}
\eta_I^{(2)}(\tau)=\eta_{\rm TI}e^{-2\tau/T_{2,I}}.
\end{equation}
This is a phenomenological retrieval envelope. Pure phase noise of an internal ion qubit does not alone imply that envelope; an independent optical survival measurement or microscopic derivation is needed to identify it with a physical device.

\paragraph{Population relaxation within the qubit is not erasure.}
When relaxation takes $\ket1\to\ket0$ inside the logical space, its Kraus operators are
\begin{equation}
A_0=\operatorname{diag}(1,\sqrt u),\quad
A_1=\sqrt{1-u}\ket0\bra1,\quad u=e^{-\tau/T_1}.
\end{equation}
Compose these with pure-dephasing operators $D_0=\sqrt{(1+\nu)/2}\id$ and $D_1=\sqrt{(1-\nu)/2}Z$, $\nu=e^{-\tau/T_\phi}$. Then
\begin{equation}
\rho_{01}(\tau)=\sqrt u\nu\rho_{01}(0),\qquad
\frac1{T_2}=\frac1{2T_1}+\frac1{T_\phi}.
\end{equation}
An independent optical erasure efficiency $\eta$ gives successful operators $\sqrt\eta VD_jA_k$ and the same two erasure operators as Eq.~\eqref{eq:four-kraus-app}. Its conditional entanglement fidelity is
\begin{equation}
F_{e|s}=\frac{1+u+2\sqrt u\nu}{4}.
\label{eq:ion-relax-fidelity}
\end{equation}
This additional channel is distinct from the pure-dephasing-plus-erasure closure. Replacing internal relaxation by a flag would incorrectly make an unheralded error heralded.

\subsubsection{From memory instruments to DES parameters}
\begin{table*}[!tbp]
\label{tabel:QM}
\caption{Memory-platform parameters and reduced-model alternatives. Efficiencies refer to the full declared write--store--read primitive. All rows use Eq.~\eqref{eq:four-kraus-app} only under balanced-rail and conditional pure-dephasing assumptions. Splitting total efficiency into write and conditional recall requires an explicit factorization.}
\centering\small
\begin{tabularx}{\textwidth}{P{2.4cm}YYY}
\toprule
Platform & Optical interface and native controls & Conditional coherence & Effective retrieval efficiency\\
\midrule
REID: spin-wave CRIB + DD & $d$, prepared spectrum, transfer pulses, $T_{\rm mem},m_R$, DD sequence & $\lambda_L=e^{-(\tau/T_{2,R}^{\rm DD})^{n_R}}$ & Independent loss: $(1-e^{-d})^2\eta_{\rm DD}$; coupled recall: multiply by $e^{-2\chi_R}$.\\
Cold atoms: EIT/Raman & $d$, $\Omega(t)$, detuning, bandwidth, $M_{\rm int}$, $D,\Delta k,\sigma_\delta,\gamma_0$ & $e^{-a\tau-b\tau^2}$ or fitted $e^{-(\tau/T_{2,C})^{n_C}}$ & Coupled recall: $\eta_{\rm EIT}(d)\eta_{\rm mode}(\tau)e^{-2\chi_C}$; distinguish parametric $\eta_{\rm EIT}$ from kernel-derived efficiency.\\
Trapped ion: Raman & $g,\Omega,\Delta$, cavity linewidth/collection, $C$, internal-state noise & $e^{-\tau/T_{2,I}}$ & Coupled recall: $\eta_{\rm TI}e^{-2\tau/T_{2,I}}$.\\
\bottomrule
\end{tabularx}
\end{table*}
The reduced model identifier must specify the optical-depth convention, logical encoding, recall variant, noise observable defining $T_2$, and the location of every efficiency. A write event samples $p_w$; recall samples $p_{r|w}(\tau)$, with $p_wp_{r|w}(\tau)=\eta(\tau)$ for the chosen primitive. It is not valid to sample the total $\eta$ both at write and at recall.

For a token already aged to $a$ and a monotone fitted $\lambda_L(a)>0$, an event at age $a+\Delta$ applies $\lambda_L(a+\Delta)/\lambda_L(a)$. The original write time must survive swaps at other nodes. Non-Markovian bath correlations between tokens or interventions require an augmented noise history; one-time Kraus maps alone do not determine multi-time process statistics. The  $F/W$ metadata can reproduce its chosen isotropic approximation, whereas exact composition of the above dephasing channels requires at least Bell-diagonal probabilities for Pauli-diagonal pairs. The richer state record defines the instrument-valued reference model and the state representation required by the native validation contract.

\subsection{Bell-state measurements: optical modes, detector instruments, and swapping}
\label{app:bsm}
\subsubsection{Model scope}
The model proceeds from polarization-encoded elementary pairs through photon retrieval, beam-splitter interference, detector response, accepted click patterns, and the remote conditional state. The accepted click-pattern instrument determines both herald probability and conditional output state. Bell-diagonal closure gives an efficient exact reduction within the independent Pauli-noise domain.

\subsubsection{The two incoming photons and the interferometer}
Label the remote endpoints $A,B$ and the two local repeater qubits $R_1,R_2$. Before swapping the state is $\rho_{AR_1}\otimes\rho_{R_2B}$ when the elementary resources are independent. An optical realization first retrieves $R_1,R_2$ into input ports $a,b$. With normalized auxiliary wavepackets $f,g$, a balanced beam splitter acts, up to known phases, as
\begin{align}
\hat a_p^\dagger[f]&\longmapsto
(\hat c_p^\dagger[f]+\hat d_p^\dagger[f])/\sqrt2,\\
\hat b_p^\dagger[g]&\longmapsto
(\hat c_p^\dagger[g]-\hat d_p^\dagger[g])/\sqrt2,
\quad p\in\{H,V\}.
\label{eq:BS-transform}
\end{align}
For identical auxiliary modes, the antisymmetric polarization Bell state produces the cross-port, orthogonal-polarization signature. A suitable polarization-resolved analyzer also distinguishes one further Bell state. The standard ancilla-free passive analyzer cannot unambiguously resolve all four equiprobable Bell states; the usual two-of-four device has success probability $1/2$ for independent Bell-pair inputs \cite{bsm,bsm_limit}. Detector architecture and the accepted click patterns must be specified to justify accepting both states; accepting only one gives $1/4$ in this ideal input case.

\paragraph{HOM overlap is an auxiliary-mode quantity.}
For two independent single photons with the same tested polarization and auxiliary density operators $\sigma_1,\sigma_2$, beam-splitter expansion gives
\begin{equation}
P_{\rm cc}=\tfrac12(1-V_{\rm HOM}),\qquad
V_{\rm HOM}=\tr(\sigma_1\sigma_2).
\end{equation}
For pure auxiliary states this is $|\langle f|g\rangle|^2$. The density operators in this expression describe the unresolved temporal, spectral, or spatial states involved in the interference test. Inserting the logical polarization marginals of perfect Bell pairs would instead give $\tr[(\id_2/2)^2]=1/2$ and misidentify perfect indistinguishability as 50\%. The auxiliary-mode and logical-polarization spaces must therefore be distinguished.

When spatial, spectral/temporal, and residual alignment variables are separable, their overlaps can factorize as $V_sV_tV_p$. Spectral--spatial correlations or frequency-dependent polarization transformations require the joint density operators. A mode overlap cannot be inferred from received intensity alone.

\subsubsection{Temporal modes, physical delay, and timestamp jitter}
For normalized spectral envelopes the field overlap is
\begin{equation}
V_t=\left|\int\dd\omega\,f^*(\omega)g(\omega)e^{-i\omega\delta t}\right|^2.
\end{equation}
As a derived example, let each temporal intensity have variance $\sigma_p^2$, with a relative delay $\delta t$ and carrier offset $\delta\omega$. Gaussian integration gives
\begin{equation}
V_t=\exp[-\delta t^2/(4\sigma_p^2)-\sigma_p^2\delta\omega^2].
\label{eq:gaussian-overlap-app}
\end{equation}
Physical emission-time or synchronization fluctuations randomize $\delta t$ and can reduce the averaged field overlap. Detector timestamp jitter occurs after interference: it does not retroactively change the incident wavepackets. It changes the accepted coincidence ensemble. Time-resolved postselection can couple these effects, but that coupling must be evaluated with a joint detection model rather than counted twice.

For a full relative-time coincidence window $w$, suppose the registered difference is Gaussian with mean $\bar\delta$ and variance $\sigma_{\rm rel}^2$. Then
\begin{align}
A_t(w)&=\Phi\!\left(\frac{w/2-\bar\delta}{\sigma_{\rm rel}}\right)
-\Phi\!\left(\frac{-w/2-\bar\delta}{\sigma_{\rm rel}}\right),\\
\sigma_{\rm rel}^2&=\sigma_{d,1}^2+\sigma_{d,2}^2+\sigma_{\rm sync}^2
+\sigma_{\rm intrinsic}^2,
\label{eq:relativegate}
\end{align}
where the final term is present only when an intrinsic arrival-difference spread is modeled. For independent absolute gates about two trigger times, the acceptance is instead the product of two single-detector probabilities. A squared single-detector acceptance corresponds to the latter interpretation. It must not silently be reused for a single relative coincidence window.

For example, two independent 50\,ps detector RMS jitters, zero other spread, and $w=200$\,ps give $A_t=\operatorname{erf}(1)\simeq0.842701$ for a relative gate. The product of two independent absolute-gate acceptances is $\operatorname{erf}(\sqrt2)^2\simeq0.911070$. This is an algebraic timing-convention check, not a detector calibration.

\subsubsection{Detector POVMs and the accepted-outcome instrument}
An on/off detector with efficiency $\eta_d$ and independent dark-click probability $p_d$ in its specified gate has effects
\begin{equation}
\Pi_0=(1-p_d)\sum_{n=0}^\infty(1-\eta_d)^n\proj n,
\qquad \Pi_1=\id-\Pi_0.
\label{eq:threshold-detector}
\end{equation}
The probability $p_d=1-e^{-D_dw}$ follows from a Poisson dark process of rate $D_d$ when the detector gate is $w$. A coincidence window and the physical detector-active gate need not be identical. Photon-number-resolving detectors, dead time, afterpulsing, and saturation require their own response models.

For a resolved click pattern $c$, tensor products of detector effects and the time-acceptance rule define $\Pi_c$. Let $U_{\rm opt}$ include beam splitting, polarization analysis, optical loss dilations, and all auxiliary modes, and let $\varrho$ include the two input resources, retrieved fields, and environment states. The unnormalized remote operation is
\begin{equation}
\mathcal J_c(\rho)=\tr_{R,\mathrm{opt},E}
\left[(\id_{AB}\otimes\sqrt{\Pi_c})U_{\rm opt}\varrho U_{\rm opt}^\dagger
(\id_{AB}\otimes\sqrt{\Pi_c})\right].
\label{eq:click-instrument}
\end{equation}
A bath/detector basis expansion supplies Kraus operators, exactly as in Eq.~\eqref{eq:dilation-kraus}. Its trace gives the accepted-pattern probability. The sum over accepted patterns, with the appropriate remote Pauli corrections, gives the heralded output. This formulation includes false heralds from dark counts, multi-pair emission, or one missing photon without confusing them with successful two-photon Bell projections. Empty endpoint sectors must remain flags unless a justified model replaces them with an effective qubit state.

In the restricted no-background, single-photon, indistinguishable-auxiliary-mode, balanced-loss case with two accepted Bell states and independent Bell-diagonal input resources, the genuine-herald probability reduces to
\begin{equation}
p_{\rm true,BSM}=\tfrac12\eta_{r,1}\eta_{r,2}
\eta_{o,1}\eta_{o,2}\eta_{d,1}\eta_{d,2}A_t.
\label{eq:BSMtrue}
\end{equation}
Here $\eta_{r,j}$ are conditional memory-recall probabilities, and $\eta_{o,j}$ are optical efficiencies from the recall boundary to the detectors. Efficiencies already sampled by the caller are set to unity inside this primitive. HOM visibility affects the accepted-state operation and, for a particular pattern rule, may also affect its probability. It is not generally interchangeable with $p_{\rm BSM}$, nor does $V_{\rm HOM}=1$ remove the $1/2$ analyzer factor. A phenomenological probability $p_{\rm BSM}\simeq V_{\rm HOM}$ does not represent this two-of-four optical analyzer.

\subsubsection{Exact swapping and the Bell-diagonal reduction}
For ideal resolved Bell outcome $r$, define $M_r=\bra{B_r}_{R_1R_2}$. The unnormalized endpoint state is
\begin{equation}
\tilde\rho_{AB}^{(r)}=(\id_A\otimes M_r\otimes\id_B)
(\rho_{AR_1}\otimes\rho_{R_2B})
(\id_A\otimes M_r^\dagger\otimes\id_B).
\label{eq:ideal-swap-operation}
\end{equation}
Its trace is the outcome probability, and the known Pauli correction $P_r$ acts on an endpoint before comparing to the target state. Optical loss and accepted-pattern weights multiply or generalize this trace-decreasing operation.

Let $G=\mathbb Z_2^2$ label Bell states $\ket{B_{xz}}=(\id\otimes X^xZ^z)\ket{\Phi^+}$, and let $p_g,q_g$ be independent input Bell probabilities. Pauli labels add by bitwise XOR after the outcome correction. Therefore
\begin{equation}
r_g=\sum_{h\in G}p_hq_{g\oplus h}.
\label{eq:bell-convolution-app}
\end{equation}
This equation follows either by inserting the Bell expansion into Eq.~\eqref{eq:ideal-swap-operation}, or by moving each input Pauli through the Bell projection; unobservable Pauli phases cancel in the density operator. A local Pauli-noise channel updates $p$ by the same convolution. The four probabilities are consequently a closed representation for independent Pauli noise and ideal corrected swaps. Correlated resources and non-Pauli operations are outside that closure.

For the Werner/isotropic parametrization,
\begin{equation}
p=(F_1,(1-F_1)/3,(1-F_1)/3,(1-F_1)/3),
\end{equation}
and similarly for $q$, Eq.~\eqref{eq:bell-convolution-app} gives
\begin{equation}
F_{\rm ideal}=F_1F_2+\frac{(1-F_1)(1-F_2)}{3},
\qquad W_{\rm out}=W_1W_2,
\label{eq:wernerexact-app}
\end{equation}
with $W_j=(4F_j-1)/3$. The second term represents matching nonidentity Bell errors that cancel under swapping.

\paragraph{Scalar comparator and failure of isotropic closure.}
For the diagnostic comparison in Section~\ref{sec:bell-reference}, define
\begin{equation}
F_{\rm sc}=\tfrac12[1+V(2F_1F_2-1)].
\label{eq:scalar-fswap}
\end{equation}
This expression is excluded from the physical reference backend. At $V=1$, two inputs with $F_1=F_2=1/2$ give $F_{\rm sc}=1/4$, whereas Eq.~\eqref{eq:wernerexact-app} gives $1/3$. Two completely mixed inputs give $F_{\rm sc}=1/16$ instead of $1/4$. The inconsistency concerns the stated isotropic input family and local swapping operation; a fidelity below $1/4$ is possible for other two-qubit states. Independent phase damping is represented by its Bell-error kernel rather than by forcing the state into a Werner family.

\subsubsection{Partial distinguishability and false-herald mixtures}
A mixture of indistinguishable and distinguishable optical sectors can be written at the unnormalized-instrument level as $\mathcal J=V\mathcal J_{\rm ind}+(1-V)\mathcal J_{\rm dis}$ when that decomposition is physically justified. Its conditional fidelity is
\begin{equation}
F_s=\frac{Vp_{\rm ind}F_{\rm ind}+(1-V)p_{\rm dis}F_{\rm dis}}
{Vp_{\rm ind}+(1-V)p_{\rm dis}}.
\end{equation}
A convex interpolation with weights $V,1-V$ is valid only when the branch success probabilities coincide or when $V$ already denotes the conditional weight. A universal distinguishable-branch fidelity of $1/2$ is not determined by mode overlap alone.

Similarly, for true and accidental accepted-event rates,
\begin{equation}
F_{\rm obs}=\frac{C_{\rm true}F_{\rm true}+C_{\rm acc}F_{\rm acc}}
{C_{\rm true}+C_{\rm acc}}.
\label{eq:accidental-mixture}
\end{equation}
For an explicitly assumed completely mixed false-qubit state, $F_{\rm acc}=1/4$. Different false-event mechanisms can yield other states. The scalar expression $C_{\rm true}F_{\rm swap}/(C_{\rm true}+C_{\rm acc})$ corresponds instead to assigning zero target overlap to the false branch. That choice is not white qubit noise, and the assumed false-event state must accompany any fidelity comparison.

\subsubsection{Event semantics at the repeater}
A swap reserves two distinct local qubits, evaluates their accumulated storage evolution, attempts recall, and applies the accepted detector instrument. Both operands are consumed by the destructive attempt, including an inconclusive outcome. A successful resource joins the two surviving remote endpoints and carries the outcome/frame information; it does not create a new material qubit at the repeater. Remote qubits retain their original memory-write times. Physical completion, arrival of the herald, and availability of a known Pauli frame are separate DES events. The SDN controller receives the reported outcome only after the modeled classical delay.

\subsection{FSO, adaptive optics, fiber, and orbital geometry}
\label{app:channels}
The analytical turbulence descriptions use the propagation conventions of Andrews and Phillips \cite{andrews2005}; the Gamma--Gamma irradiance model additionally follows Al-Habash, Andrews, and Phillips \cite{alhabash2001}. Receiver long-exposure phase statistics are treated using Fried's formulation \cite{fried1966}.
\subsubsection{From propagation of field modes to a quantum channel}
An end-to-end optical trial efficiency separates source preparation, geometric collection, atmospheric propagation, pointing, coupling, and detection. These categories define accounting boundaries, not necessarily independent random variables. A propagation channel is distinguished from the complete trial in which it is used. The field-mode construction below connects a mode-resolved backend to analytical, event-level approximations \cite{pirandola_sat,vasylyev,ao}.

For a scalar paraxial mode in a medium with a small refractive-index perturbation $\delta n$, a representative propagation equation is
\begin{equation}
\partial_z u=\frac{i}{2k}\nabla_\perp^2u+ik\delta n(\boldsymbol r_\perp,z,t)u
-\frac{\alpha_{\rm ext}(z,\lambda_0)}2u,
\quad k=2\pi/\lambda_0.
\label{eq:paraxial-app}
\end{equation}
The time argument denotes slow environmental evolution relative to optical flight; dynamic propagation requires its retarded value when that approximation fails. A split-step backend alternates Fresnel propagation with random phase screens and extinction. The screen spectrum, outer/inner scales, sampling resolution, aperture truncation, and temporal correlation determine its validity. A label such as ``Rytov variance'' alone does not define that backend.

For environmental realization $\xi$, let $\mathcal U_\xi$ propagate the input field to the receiver, $P_A$ restrict it to the aperture, and $C_{\rm AO,\xi}$ apply the receiver correction. In accepted input/output mode bases,
\begin{equation}
T_{mn}^{(\xi)}=\langle g_m|C_{\rm AO,\xi}P_A\mathcal U_\xi|f_n\rangle.
\label{eq:channel-mode-matrix}
\end{equation}
Passivity implies $T_\xi^\dagger T_\xi\leq\id$. A vacuum-environment dilation can be written $\hat{\boldsymbol b}=T_\xi\hat{\boldsymbol a}+L_\xi\hat{\boldsymbol e}$ with $T_\xi T_\xi^\dagger+L_\xi L_\xi^\dagger=\id$ on the output modes. Restricting to a single photon gives the successful map $\rho\mapsto T_\xi\rho T_\xi^\dagger$ and the erasure completion of Eq.~\eqref{eq:general-flagged}. For fading, average the unnormalized maps before conditioning. A polarization-dependent transfer needs a joint polarization--spatial matrix; scalar turbulence propagation alone does not establish polarization decoherence.

An offline lookup can expose $p(\eta,\mathcal E_{\rm pol},\delta t\mid\mathrm{geometry},\mathrm{weather},\mathrm{AO})$ to the DES. The joint distribution matters: conditioning on received photons favors high-transmission realizations, and averaging a conditional fidelity independently of transmission generally gives a different ensemble.

\subsubsection{Gaussian spreading and aperture collection}
The equivalent-Gaussian model uses the $1/e^2$ intensity radius
\begin{equation}
w_L=\sqrt{w_0^2+(\theta_{\rm div}L)^2},\qquad w_0=D_t/2.
\label{eq:gaussian-radius}
\end{equation}
The relation $w_0=D_t/2$ is a chosen launch filling convention, not a universal identity for a truncated telescope beam. For a diffraction-limited Gaussian, $\theta_{\rm div}=\lambda_0/(\pi w_0)$; an $M^2$ factor or a measured divergence can replace this relation. Consequently an 810/1550\,nm comparison must declare whether divergence is fixed by optics, measured separately, or held artificially equal.

For normalized intensity $I(r)=2e^{-2r^2/w_L^2}/(\pi w_L^2)$, a centered aperture of radius $a=D_r/2$ collects
\begin{equation}
\eta_{\rm geo}=\int_0^a2\pi rI(r)\dd r
=1-e^{-2a^2/w_L^2}=1-e^{-D_r^2/(2w_L^2)}.
\label{eq:aperture-app}
\end{equation}
This gives the centered-aperture collection law. It approaches unity for a large aperture and scales as $D_r^2/(2w_L^2)$ when the aperture is small.

\subsubsection{Pointing: instantaneous displacement versus mean jitter loss}
For beam-center displacement $\boldsymbol r_0$ relative to the aperture center, the exact Gaussian collection is
\begin{align}
\eta_{\rm ap}(r_0)&=\int_{|\boldsymbol r|\leq a}
\frac{2e^{-2|\boldsymbol r-\boldsymbol r_0|^2/w_L^2}}{\pi w_L^2}\dd^2r\nonumber\\
&=F_{\chi'^2_2(4r_0^2/w_L^2)}(4a^2/w_L^2).
\label{eq:offset-aperture}
\end{align}
The last form is the noncentral chi-square cumulative distribution with two degrees of freedom. It follows because each transverse coordinate of the Gaussian intensity distribution has variance $w_L^2/4$. It is directly usable for numerical sampling or quadrature.

In the small-aperture approximation, $\eta_{\rm ap}(r_0)/\eta_{\rm ap}(0)\simeq e^{-2r_0^2/w_L^2}$. If pointing displacement has independent coordinate variance $\sigma_r^2=L^2\sigma_\theta^2$ and mean offset of magnitude $b$, Gaussian integration gives
\begin{equation}
\E\left[e^{-2r_0^2/w_L^2}\right]
=\frac{\exp[-2b^2/(w_L^2+4\sigma_r^2)]}{1+4\sigma_r^2/w_L^2}.
\label{eq:mean-pointing}
\end{equation}
For zero mean and a \emph{radial} RMS displacement $r_{\rm rms}=\sqrt2\sigma_r$, the denominator is $1+2r_{\rm rms}^2/w_L^2$. The expression $e^{-2(\sigma_\theta L/w_L)^2}$ is instead a representative-offset penalty obtained by substituting one displacement into the instantaneous loss law. Using the displaced-aperture integral already accounts for geometric collection and pointing, so a second geometric or pointing factor would double count them.

\subsubsection{Extinction, elevation, and the turbulence profile}
Absorption and scattering attenuate power according to
\begin{equation}
T_{\rm ext}=T_{\rm abs}T_{\rm sca}
=\exp\left[-\int_{\rm ray}(\alpha_{\rm abs}+\alpha_{\rm sca})\dd s\right].
\end{equation}
For an approximately plane-parallel atmosphere of effective vertical thickness $h_0$ and \emph{elevation} $\epsilon$, $L_{\rm atm}\simeq h_0/\sin\epsilon$. A curved atmosphere and a minimum elevation mask are preferable near the horizon. Clouds are a separate availability/extinction process; a clear-air turbulence parameter does not model cloud blockage.

A Hufnagel--Valley profile is
\begin{align}
C_n^2(h)={}&0.00594(v/27)^2(10^{-5}h)^{10}e^{-h/1000}\nonumber\\
&+2.7\times10^{-16}e^{-h/1500}+A_0e^{-h/100},
\label{eq:HV-app}
\end{align}
with $h$ in meters and the conventional profile coefficients. $A_0$ and $v$ parameterize a selected atmospheric profile, not a measurement unless calibrated. In the plane-wave slant-path approximation, the Fried coherence diameter is
\begin{equation}
r_0=\left[0.423k^2\frac{1}{\sin\epsilon}
\int C_n^2(h)\dd h\right]^{-3/5}.
\label{eq:fried-app}
\end{equation}
The factor is cosecant when the argument is elevation. Equivalently, it is secant for the zenith angle $z=\pi/2-\epsilon$. At zenith, Eq.~\eqref{eq:fried-app} remains finite. An implementation must not interchange these angle conventions.

For a homogeneous plane wave in weak turbulence,
\begin{equation}
\sigma_R^2=1.23C_n^2k^{7/6}L^{11/6}
\end{equation}
is a standard limiting parameterization. Slant paths, spherical/Gaussian waves, finite apertures, and uplink versus downlink propagation require the corresponding path weighting. The same value of $\sigma_R^2$ does not specify beam wander, wavefront residuals, or temporal correlations uniquely. In particular, a scalar intensity distribution is not a substitute for a spatial random field \cite{vasylyev,ao}.

\subsubsection{Normalized fading and empirical mean collection}
\paragraph{Weak-fluctuation surrogate.}
A normalized lognormal irradiance is
\begin{align}
H&=e^{-s_\ell^2/2+s_\ell Z},\quad Z\sim\mathcal N(0,1),\\
\E H&=1,\qquad \operatorname{Var}(H)=e^{s_\ell^2}-1.
\label{eq:lognormal}
\end{align}
Hence $s_\ell^2=\ln(1+\sigma_I^2)$ if $\sigma_I^2$ is the scintillation index. In the weak plane-wave limit $\sigma_I^2\simeq\sigma_R^2$. A separately fitted relation such as $\sigma_I^2\simeq1.23\sigma_R^{12/5}$ is an alternative surrogate, not that weak-limit identity. Likewise, a switch near $\sigma_R^2=0.3$ is a model-selection rule, not a universal atmospheric phase boundary.

\paragraph{Gamma--Gamma surrogate.}
Let $X\sim\mathrm{Gamma}(\alpha,1/\alpha)$ and $Y\sim\mathrm{Gamma}(\beta,1/\beta)$ independently, and $H=XY$. Integrating $p_X(x)p_Y(h/x)/x$ gives
\begin{equation}
p_H(h)=\frac{2(\alpha\beta)^{(\alpha+\beta)/2}}{\Gamma(\alpha)\Gamma(\beta)}
h^{(\alpha+\beta)/2-1}K_{\alpha-\beta}(2\sqrt{\alpha\beta h}),\quad h>0.
\end{equation}
The normalization follows from the product construction; $\E H=1$ and $\operatorname{Var}(H)=\alpha^{-1}+\beta^{-1}+(\alpha\beta)^{-1}$. In terms of $s_R=\sigma_R^2$, one analytical shape prescription is
\begin{align}
\alpha^{-1}&=\exp\!\left[\frac{0.49s_R}{(1+1.11s_R^{6/5})^{7/6}}\right]-1,\\
\beta^{-1}&=\exp\!\left[\frac{0.51s_R}{(1+0.69s_R^{6/5})^{5/6}}\right]-1.
\label{eq:GG-shapes}
\end{align}
These expressions specify the shape parameters of this surrogate. Their use for a particular finite-aperture slant link still requires justification. $H$ is relative irradiance and can exceed unity. If $\eta=\bar\eta H$ is clipped to $[0,1]$, the clipping changes the mean and must be recorded. A bounded distribution of collected power or a wave-optics lookup avoids interpreting unbounded irradiance as transmissivity.

\paragraph{An explicit limitation of the mean-loss equation.}
For fixed extinction and a normalized irradiance fluctuation,
\begin{equation}
\E[T_{\rm ext}H]=T_{\rm ext}.
\end{equation}
A turbulence-dependent exponential reduction of the mean therefore does not follow from this expectation. Such a reduction can instead represent a separately calibrated \emph{mean collected-mode loss}, for example
\begin{equation}
\bar\eta_{\rm turb}=e^{-C_{\rm atm}s_R^{6/5}},\qquad
\eta\approx\eta_{\rm fixed}\bar\eta_{\rm turb}H,
\label{eq:empirical-turb}
\end{equation}
subject to boundedness and a declared joint model. Beam broadening and finite-aperture loss can change mean collected power even though normalized scintillation has mean one. We retain Eq.~\eqref{eq:empirical-turb} as an empirical backend, not as a derived consequence of Eq.~\eqref{eq:lognormal}. The single-repeater example uses $C_{\rm atm}=0.06$. The QUASAR mean-signal calculation does not apply an additional mean turbulence penalty; its turbulence quantities are separate diagnostics.

\subsubsection{Adaptive optics and conditional mode coupling}
Normalize the field before aperture loss so $\int|u|^2\dd^2r=1$, and the accepted receiver mode so $\int_A|g|^2\dd^2r=1$. Then
\begin{align}
\eta_{\rm ap}&=\int_A|u|^2\dd^2r,\\
\eta_{\rm cpl|ap}&=\frac{|\int_Ag^*(\boldsymbol r)u(\boldsymbol r)e^{-i\phi_{\rm AO}(\boldsymbol r)}\dd^2r|^2}{\eta_{\rm ap}}\leq1.
\label{eq:ao-overlap}
\end{align}
The overall accepted fraction is $\eta_{\rm ap}\eta_{\rm cpl|ap}$, with fixed mode-matching optics included consistently. This expression directly connects the corrected field modes to the single-photon transfer in Eq.~\eqref{eq:channel-mode-matrix}.

For a matched pupil with sufficiently small residual phase, a useful approximation is $\eta_{\rm cpl}\simeq\eta_{\rm ideal}e^{-\sigma_{\phi,\rm res}^2}$. Fitting error, servo lag, wavefront-sensor noise, anisoplanatism, and uncorrected modes enter $\sigma_{\phi,\rm res}^2$ according to the specified AO design. This is a regime-limited reduction, not a complete AO model. A relative gain $G_{\rm AO}=\eta_{\rm AO\,on}/\eta_{\rm AO\,off}$ may exceed one; a physical coupling probability may not. Factors $\eta_{\rm AO}$ and $G_{\rm AO}$ must be distinguished on that basis. Direct bucket detection and single-mode fiber injection do not have identical wavefront sensitivity \cite{ao}.

The empirical factors $\eta_{\rm HV}=e^{-(D_r/r_0)^{5/3}}$, $G_{\rm ap}=1-e^{-(D_r/D_{\rm sat})^2}$, and $G_{\rm bw}=1/(1+2\sigma_{\rm bw}^2/w_e^2)$ can be retained as named calibrated surrogates. They are not automatically independent of the aperture, pointing, and turbulence terms already evaluated. An independent product of mean efficiencies also differs from the mean of their correlated product.

\subsubsection{Fiber propagation and wavelength conversion}
For fiber length $L_f$ in kilometers, power-attenuation coefficient $\alpha_f$ in dB/km, and insertion loss $\ell_{\rm ins}$ in dB,
\begin{equation}
\eta_f=10^{-(\alpha_fL_f+\ell_{\rm ins})/10}.
\end{equation}
A mode-resolved narrowband propagation matrix can be written
\begin{align}
T_f(\omega)&=\sqrt{\eta_f(\omega)}\,J_f(\omega)e^{i\varphi_f(\omega)},\\
\varphi_f(\omega)&=L_f[\beta_0+\beta_1(\omega-\omega_0)\nonumber\\
&\hspace{18mm}+\tfrac12\beta_2(\omega-\omega_0)^2+\cdots],
\end{align}
using consistent length units in the phase. $J_f$ represents birefringence or polarization-dependent transfer. Group delay is $n_gL_f/c$; dispersion changes the outgoing temporal mode and therefore BSM overlap. Fiber attenuation alone is an erasure process, not a depolarizing map. Random birefringence, polarization-mode dispersion, and detector-conditioned background can reduce conditional fidelity through distinct mechanisms.

A downlink at 810\,nm and a telecom fiber or a Pr/alkali/ion memory are not automatically spectrally compatible. An explicit frequency-conversion interface must specify conversion efficiency, temporal/spectral transfer, added noise, and polarization preservation. Merely assigning the same token to those devices bypasses a physical interface that can dominate the experiment.

\subsubsection{Orbital state, visibility, and time of flight}
For a Keplerian benchmark with semimajor axis $a_o$, eccentricity $e_o$, and mean motion $n_o=\sqrt{\mu/a_o^3}$, solve
\begin{equation}
M(t)=M_0+n_o(t-t_0)=E_o-e_o\sin E_o.
\end{equation}
The inertial position follows from rotations of
\begin{align}
\boldsymbol r_{\rm orb}&=
\begin{pmatrix}a_o(\cos E_o-e_o)\\a_o\sqrt{1-e_o^2}\sin E_o\\0\end{pmatrix},\\
\boldsymbol r_I&=R_3(\Omega_o)R_1(i_o)R_3(\omega_o)\boldsymbol r_{\rm orb}.
\end{align}
This ideal propagation is not a replacement for TLE/OMM propagation with its appropriate SGP4 conventions. Epoch, reference frame, time scale, Earth rotation, and the propagator must be recorded \cite{vallado}. The single-repeater local tangent-plane trajectory is a separate controlled geometry used to isolate offset and crossing-angle effects; its labels must not be mistaken for a mission ephemeris.

For ground position $\boldsymbol r_g$ and local up unit vector $\boldsymbol u_g$ in the same Earth-fixed frame as the satellite,
\begin{equation}
L_g=\|\boldsymbol r_s-\boldsymbol r_g\|,\qquad
\epsilon_g=\arcsin\left[\boldsymbol u_g\cdot
(\boldsymbol r_s-\boldsymbol r_g)/L_g\right].
\end{equation}
A link is usable only if the elevation mask, tracking constraints, optical-terminal allocation, weather model, and hardware readiness all permit it. For spherical Earth radius $R_E$ and altitude $h$, the geometric slant-range check is
\begin{equation}
L_g(\epsilon)=\sqrt{(R_E+h)^2-R_E^2\cos^2\epsilon}-R_E\sin\epsilon.
\end{equation}
It gives $L_g(\pi/2)=h$ and $L_g(0)=\sqrt{(R_E+h)^2-R_E^2}$. An inter-satellite line segment must remain outside the selected Earth/atmosphere exclusion radius. Find the closest point $\boldsymbol r_1+u(\boldsymbol r_2-\boldsymbol r_1)$ with $u$ clamped to $[0,1]$ and test its norm, rather than applying a ground-elevation mask to an inter-satellite link.

Photon arrival time solves $t_a-t_e=\|\boldsymbol r_r(t_a)-\boldsymbol r_t(t_e)\|/c$ in vacuum, with optical-path corrections when required. The frozen-geometry approximation uses $L(t_e)/c$ and should be checked at the desired resolution. First-order Doppler is $\delta\omega/\omega\simeq-\dot L/c$; spectral compensation and synchronization affect memory acceptance and Eq.~\eqref{eq:gaussian-overlap-app}. The DES tracks propagation, storage, heralding, and correction delays separately.

\subsubsection{From source and count processes to state quality}
Source repetition $f_{\rm rep}$, mean pairs per excitation $\mu$, pair flux $R_p$, and the simulator's per-mode attempt rate have different units and meanings. For an ideal single two-mode-squeezed pair mode, $P_n=\mu^n/(1+\mu)^{n+1}$ and $R_p=f_{\rm rep}\mu$. A multimode polarization-entangled SPDC source needs its appropriate photon-number and joint spectral model. A ``1000 MHz emitted'' label must therefore be linked to either a pulse clock or an emitted pair flux before inferring multi-pair noise.

For a direct two-downlink coincidence benchmark in the low-occupancy regime,
\begin{align}
C_{\rm true}(t)&=R_p(t)\eta_A(t)\eta_B(t)A_t(w),\\
C_{\rm acc}(t)&\simeq S_A(t)S_B(t)w,
\label{eq:count-rates}
\end{align}
where $S_j$ are the relevant independent singles processes, including background and dark counts, and the accidental expression assumes a full-width relative gate with negligible saturation. Genuine paired counts must not be counted again as independent accidentals. At high brightness or appreciable $S_jw$, an explicit point-process/detector model replaces this approximation.

A dimensionally defined signal-to-noise ratio is $C_{\rm true}/C_{\rm acc}$, with both in s$^{-1}$ or both counted in the same window. If genuine pairs are ideal and accidental qubit states are $\id_4/4$, then
\begin{equation}
W=\frac{C_{\rm true}}{C_{\rm true}+C_{\rm acc}},\qquad
F=\frac{1+3W}{4}.
\end{equation}
This derives the SNR-to-Werner relation only for that mixture model. It does not say that pure propagation loss itself depolarizes a surviving photon. Source, background, and dark-count parameters must use an explicit counts-or-rates convention before evaluation. At a delivered-service boundary, memory and swapping are further operations and cannot be inferred from a direct coincidence rate.

\subsection{Space--time-based distribution and segmented nesting}
\label{app:stbd}
STBD and segmented nesting are adopted from the space--time distribution framework of Koudia et al.\ \cite{stbd2024}. This subsection formalizes their use with explicit two-ended resources, age-dependent noise, and execution constraints.
\subsubsection{Planning and physical execution}
STBD uses a space--time graph, product utilities, destination-time search, binary nesting, and early nesting of snapshot-consecutive segments. Causal edge durations, endpoint memory accounting, final-state quality, and failure handling connect that routing structure to an executable entanglement plan. Optimality of the selected additive route cost is distinct from global optimality for a stochastic, multi-request network.

Let the management horizon be divided by $t_0<t_1<\cdots<t_K$. A snapshot $G_k=(V,E_k)$ records candidate physical contacts. A lifted vertex $(v,k)$ denotes node $v$ at time $t_k$, not a qubit that has been copied into every layer. Two different edge types are required. A communication edge $(u,k)\to(v,k')$ represents a physical elementary-link opportunity whose generation and required heralding can complete no earlier than $t_{k'}$. A waiting edge $(v,k)\to(v,k+1)$ represents keeping a physical excitation in an allocated memory slot. Waiting edges require nonzero memory capacity and an appropriate survival/quality model.

The snapshot-to-space--time construction is described in Ref.~\cite{stbd2024}. Communication edges across layers specify link opportunities, not complete storage reservations or physical swaps. A communication edge creates a \emph{two-ended} entanglement resource; it does not forward an unknown qubit like a classical packet. Memory waiting, endpoint inventories, and operation prerequisites complete the execution model.

\subsubsection{Weights, horizons, and the shortest-path transformation}
For each eligible edge, define
\begin{align}
U_e&=(U_e^{\rm fid})^{\omega_f}(U_e^{\rm mem})^{\omega_m},\\
U_e^{\rm mem}&=e^{-\Delta t_e/t_c},\quad
\omega_f,\omega_m\geq0,\quad\omega_f+\omega_m=1,\\
U(P)&=\prod_{e\in P}U_e.
\label{eq:stbd-utility}
\end{align}
The function $U_e^{\rm fid}\in(0,1]$ is part of the policy configuration, not a unique consequence of the physical backend. Choosing a monotone function of predicted fidelity assigns a preference among routes.

For positive utilities,
\begin{equation}
\arg\max_P U(P)=\arg\min_P\sum_{e\in P}c_e,
\qquad c_e=-\log U_e\geq0.
\label{eq:logutility}
\end{equation}
Zero-utility edges are removed. This proves equivalence between multiplicative-utility routing and an additive shortest path for the chosen graph and weights. It does not include resource contention, correlated generation failures, or output-state constraints unless they are encoded in the state and feasible edge set. All target time layers can connect to a terminal sink; a single shortest-path search then replaces an explicit loop over destination times. If edges advance time, the graph is acyclic and a topological dynamic program suffices. Otherwise nonnegative-weight shortest-path methods are appropriate.

The planning horizon may be tied to a coherence scale $t_c$ or fixed independently. Automatically setting it to the swept physical $T_2$ changes both memory physics and planning foresight. A memory-sensitivity study therefore fixes the horizon, while a separate horizon sweep measures the routing-policy effect. A stretched-exponential $T_2$ is an operational coherence parameter, not a sharp deterministic failure boundary.

For independent Poisson opportunities of intensity $\lambda_e(t)$ with no resource blocking, an optional planning predictor is
\begin{equation}
p_{e,\rm window}=1-\exp\left[-\int_{t_a}^{t_b}\lambda_e(t)\dd t\right].
\end{equation}
This is the probability of at least one generation opportunity over that interval, conditional on the declared link model. 

\subsubsection{The physical plan and resource constraints}
Represent an elementary or swapped pair by
\begin{equation}
\mathsf P=(\mathrm{id},u,m_u,t_{w,u},v,m_v,t_{w,v},\boldsymbol p,
\mathrm{frame},\mathrm{request}),
\end{equation}
where $\boldsymbol p$ is a Bell-probability vector in the refined Pauli-diagonal model, and the two original write times are retained. A scalar $F/W$ record remains a compact approximation-specific adapter. A generation action declares the source, propagation direction, retained local half, remote loading operation, and successful-herald condition.

For an executable schedule $\pi$, physical memory occupancy obeys
\begin{equation}
\sum_{q:\,q\ \mathrm{stored\ at}\ v\ \mathrm{during}\ t}1\leq M_v(t).
\label{eq:memconstraint}
\end{equation}
Optical terminals, source flux, detector dead time, and BSM processors impose analogous constraints. A memory bank created dynamically for every link label must not create unlimited physical capacity: per-label reservations must sum to the hardware budget. If a source of total attempt rate $r_v$ serves several modes/links, their assigned rates sum to at most $r_v$, unless independent additional sources are explicitly modeled.

A swap at node $r$ can consume $\mathsf P_{ur}$ and $\mathsf P_{rv}$ only if their local qubits are distinct, present, unconsumed, not reserved by another operation, and causally known to the executor. Their remote endpoints become $u,v$ on success. The pair identity changes, but $t_{w,u},t_{w,v}$ do not. A failure consumes both local inputs and invalidates their pair correlations; cancellation/heralds reach the remote records according to the modeled classical delay.

\subsubsection{Exact age bookkeeping and the end-to-end quality constraint}
For a fitted coherence $\lambda(a)=e^{-\chi(a)}$ and an already applied age $a$, a later event at $b$ must apply
\begin{equation}
\lambda_{a\to b}=\frac{\lambda(b)}{\lambda(a)}
=e^{-[\chi(b)-\chi(a)]},\qquad b\geq a.
\label{eq:age-ratio-app}
\end{equation}
This telescopes over arbitrary event subdivisions. Replacing it by $e^{-\chi(b-a)}$ is generally wrong for $\chi(a)=(a/T_2)^n$ with $n\neq1$. Consequently a memory-wait edge has an age-dependent quality cost unless exponential decay is assumed. Exact optimization may require age-augmented states; a simpler planner must verify actual ages during execution. This bookkeeping describes the chosen monotone effective channel, not general multi-time non-Markovian noise under arbitrary interventions.

Under an isotropic independent-noise approximation, define $\mathcal H(\pi)$ as the set of all actual local-qubit storage histories, from write to local consumption or final delivery. Then
\begin{equation}
W_{\rm end}(\pi)=\prod_{e\in P}W_e\,
\prod_{q\in\mathcal H(\pi)}\lambda^{\rm iso}_{q}(\tau_q)
\prod_{s\in\pi}\nu_s,
\label{eq:history-product}
\end{equation}
where $\nu_s$ is included only if the BSM error is explicitly modeled as an independent isotropic channel. Every elementary link factor and every storage history appears once. Transition waiting is included through the actual surviving endpoint histories, not through an unspecified constant.

For pure-dephasing memory models, use the Bell convolution in Eq.~\eqref{eq:bell-convolution-app}, with each local phase-error distribution $[(1+\lambda_L)/2,(1-\lambda_L)/2,0,0]$ in the declared Bell ordering, instead of multiplying a Werner parameter by $\lambda_L$. Local amplitude damping or state-dependent loss requires a more general representation. The physical channel and the metadata update must agree.

When Eq.~\eqref{eq:history-product} is valid and $F_{\min}>1/4$, the final-quality condition is
\begin{equation}
-\ln W_{\rm end}\leq-\ln[(4F_{\min}-1)/3].
\end{equation}
Checking $F_e\geq F_{\min}$ for every elementary edge is insufficient. Four independent Werner links with $F_e=0.9$ have, even with perfect memories and swaps,
\begin{equation}
F_{\rm end}=\frac{1+3[(4\times0.9-1)/3]^4}{4}
\simeq0.673126.
\end{equation}
This explicit counterexample explains the need for final-state gating beyond per-link threshold filtering.

\subsubsection{Binary and segmented nesting}
A path of $n$ nodes has $\ell=n-1$ elementary links and $\ell-1$ repeater swaps. In a balanced binary schedule, adjacent available \emph{segments} are joined at their shared repeater, with an unmatched segment carried to the next level. The logical depth is
\begin{equation}
L_{\rm nest}=\lceil\log_2\ell\rceil\quad(\ell\geq1),
\end{equation}
so a direct elementary link needs zero rounds and a two-link path needs one. A BSM acts on two local qubits, not on two different repeater nodes. Parallel operations must use disjoint local operands and obey processor constraints.

For segmented nesting, split the selected route into snapshot-consecutive subpaths $P^{(1)},\ldots,P^{(K')}$. Establish elementary links in each segment, join ready local subsegments, store the remaining segment endpoints across the required transitions, and stitch neighboring segment resources when the prerequisites are available. The DES, not the coarse snapshot index, sets every generation, recall, BSM, and herald timestamp.

\begin{proposition}[Restricted memory-quality comparison]
Suppose a segmented and a late-nesting schedule use the same successful elementary pairs and the same accepted ideal swaps. Suppose their final delivery time and endpoint memories are fixed, and segmentation weakly shortens every internal-qubit storage history without introducing any extra waiting, noise, or resource conflict. Under independent monotone isotropic storage factors, $F_{\rm seg}\geq F_{\rm late}$. The inequality is strict only if a changed history has strictly decaying coherence and the remaining visibility product is nonzero.
\end{proposition}
\begin{proof}
Each unchanged factor in Eq.~\eqref{eq:history-product} cancels in the ratio. Every shortened history increases or preserves its nonnegative coherence factor. Hence $W_{\rm seg}\geq W_{\rm late}$ and $F=(1+3W)/4$ preserves the ordering. The stated conditions determine strictness.
\end{proof}
This proposition identifies the conditions under which early nesting preserves or improves memory-induced quality. 

\subsubsection{Executable request and retry procedure}
\begin{algorithm}[t]
\caption{STBD with resource-constrained physical execution}
\label{alg:stbd-expanded}
\begin{algorithmic}[1]
\Require Endpoints $(A,B)$, $F_{\min}$, deadline $D$, horizon $H$, physical budgets, model identifiers
\State Accept request; create request/plan epoch and deadline events
\State Build time-expanded contacts and feasible waiting edges from available or forecast information
\State Rank candidate paths by Eq.~\eqref{eq:logutility}; declare $U^{\rm fid}$
\For{candidate path in the configured search budget}
  \State Build segmented swap dependencies and endpoint reservations
  \If{predicted final quality, timing, and hardware budgets are feasible}
    \State Issue generation schedule; enter event-driven execution
    \State \textbf{break}
  \EndIf
\EndFor
\While{request active and simulation time $<D$}
  \State Process arrived storage/swap reports, not future outcomes
  \State Revalidate token identities, original ages, and reservations
  \If{two prerequisites for a permitted local swap are ready}
    \State Schedule recall, detector instrument, and result herald
  \EndIf
  \If{a pair reaches $(A,B)$ with an available correction frame}
    \State Evaluate final-state quality at the delivery boundary
    \If{$F\geq F_{\min}$} \State Deliver and release unused reservations \EndIf
  \EndIf
  \If{failure, expiry, or invalidated contact prevents completion}
    \State Apply declared local retry, cycle reset, or replanning policy
  \EndIf
\EndWhile
\State Abort unresolved request; notify endpoints and clear its resources
\end{algorithmic}
\end{algorithm}
Algorithm~\ref{alg:stbd-expanded} combines path selection, segmented nesting, feedback, and retry with explicit resource validation. Candidate enumeration, reservation method, replanning trigger, and retry locality are policy parameters.

With $K$ layers and $|V|$ physical nodes, the basic lifted graph has $O(K|V|)$ vertices and $O(\sum_k|E_k|+K|V|)$ edges. A nonnegative-weight heap-based shortest path costs $O((|\mathcal V|+|\mathcal E|)\log|\mathcal V|)$; an acyclic graph can be processed in $O(|\mathcal V|+|\mathcal E|)$. Age augmentation, multi-commodity reservations, and candidate backtracking add cost. Planning time and event-execution time should be reported separately.

\subsubsection{Throughput and control-plane KPIs}
A product of source, memory, swap, and transmission probabilities multiplied by a mode count is a per-trial yield scaling, not a rate. A trial-frequency or elapsed-time factor is required for units of pairs per second. The delivered-rate estimator is
\begin{equation}
R_{\rm usable}=\frac{1}{T_{\rm obs}}
\sum_j\mathbf1\{\mathrm{delivered}_j,\ F_j\geq F_{\min},\ t_j\leq D_j\}.
\label{eq:useful-rate-app}
\end{equation}
In a stationary two-link reset toy model with independent exponential readiness times $X_i\sim\mathrm{Exp}(r_i)$, ideal unlimited waiting, and each cycle consuming both links,
\begin{equation}
\E\max(X_1,X_2)=\frac1{r_1}+\frac1{r_2}-\frac1{r_1+r_2}.
\end{equation}
Adding a fixed swap/notification overhead $t_o$ and independent accepted success probability $p_b$ gives $R=p_b/[\E\max(X_1,X_2)+t_o]$. This is a useful limiting-case check, not an approximation for a fading, finite-buffer constellation without further analysis.

Memory-mode time per usable pair is $\int_0^{T_{\rm obs}}\sum_vB_v(t)\dd t/N_{\rm usable}$; retries and drops are separated by cause. Controller telemetry carries request/plan epoch, pair ID, endpoints, mode IDs, original write times, quality/model metadata, and event timestamp. Classical messages are scheduled with propagation, serialization, processing, queuing, and optional polling delay. Ryu HTTP wall-clock time is a runtime measurement, not automatically simulated network latency. Controller comparisons must hold information availability, policy, and physical random streams fixed when testing only the transport boundary.

\subsection{Scenario configurations and counting conventions}
\label{si:records}
\subsubsection{Single-repeater configurations}
The controlled S1 sweep and dated Paris--Vienna run have distinct geometric and capacity settings. Section~\ref{sec:s1} defines the two co-located limits and the two 500\,km separated-terminal cases; Table~\ref{tab:campaign-parameters} records their numerical defaults. The satellite stores one local half of each elementary pair and performs the BSM, while the remote halves belong to the two OGSs. Hardware implementation requires its own loading and spectral-interface specification.
\begin{table*}[!tbp]
\centering\small
\caption{Recorded numerical settings for the controlled S1 sweep and dated Paris--Vienna run. The sweep and geographical experiment use the same recorded optical and memory defaults, with the listed capacity and elevation-mask differences. All efficiencies are dimensionless; angular jitter and the static pointing factor are separate configuration fields.}
\label{tab:campaign-parameters}
\begin{tabularx}{\textwidth}{P{4.6cm}YY}
\toprule
Parameter & Four-archetype sweep & Dated Paris--Vienna case\\
\midrule
Execution / source rate & Local event execution; $5\times10^5$\,s$^{-1}$ & Local event execution; $5\times10^5$\,s$^{-1}$\\
Memory platform / depth / exponent & Cold atom; $d=20$; $n=1.3$ & Cold atom; $d=20$; $n=1.3$\\
$T_2$ and retrieval envelope & $T_2=100$\,s default; $T_{\rm mem}=1000$\,s, $m=1$ & $T_2=100$\,s; $T_{\rm mem}=1000$\,s, $m=1$\\
Initial mode efficiency / capacity & $0.9$; 8192 modes per label & $0.9$; 256 modes per label\\
Wavelength / apertures & 810\,nm; $D_t=0.15$\,m, $D_r=0.50$\,m & Same\\
Divergence / pointing & $12\,\mu$rad; static factor $0.97$; jitter $0.01\,\mu$rad & Same\\
Source / coupling / detection factors & $0.85/0.70/0.75$ & Same\\
Atmospheric fields & $t_{\rm abs}=t_{\rm sca}=0.94$, $c_{\rm atm}=0.06$ & Same\\
Turbulence / AO & $\sigma_R^2=0.08$ default; loss/quality coefficients $0.08/0.03$; AO gain $1$ & Same\\
Stochastic fading / jitter & Both enabled & Both enabled\\
Polarization / phase / temporal mismatch & $0.02$\,rad / $0.1$\,rad / $20$\,ps; coherence-width field $100$\,ps & Same\\
Count-model fields & $p_t=0.5$, $n_0=10^{-5}$, $n_{\rm dc}=120$\,s$^{-1}$, $n_{\rm conv}=0$, gate $1$\,ns & Same\\
BSM configuration & $V_{\rm HOM}=0.97$, $\eta_d=0.82$, jitter $50$\,ps, window $200$\,ps, $p_{\rm dc}=10^{-6}$ & Same\\
Elevation mask / bin width & $20^\circ$ / $2$\,s & $10^\circ$ / $2$\,s\\
Replication / seed & Four repeats per setting; base seed 62345 & One run; seed 103445\\
\bottomrule
\end{tabularx}
\end{table*}

\subsubsection{Run-bin statistics and model representation}
The complete controlled campaign contains 32 settings, four repeats per setting, 24\,320 run-bin rows, and 6080 ensemble rows. Every rate is a counted number of successes divided by its 2\,s bin width, followed by aggregation over repeats. Fidelity in the plotted sweep is the median of supported repeat-bin mean fidelities. Empty-support bins are undefined. The integrated count for each run is the sum of its bin counts. Means and sample standard deviations are calculated across the four full runs.

The supplied compact-state trajectories retain the state-update conventions of their numerical configuration. They are distinguished from the exact Bell-vector and general instrument constructions in Sections~\ref{app:memories} and \ref{app:bsm}. In particular, a low-fidelity scalar trajectory is not retrospectively assigned the pure-dephasing limit. The full-grid plots, original run-bin records, and parameter definitions remain together in the source-data package.

\subsubsection{S2 constellation counting boundary}
The S2 figure retains the four views of the same timed distribution experiment. Its 150\,s snapshots, 300\,s request spacing, 600\,s horizon, and 30\,s delivery cadence are separately specified. Counts are per event. Zero quality values between plotted windows are a no-delivery display convention. No controller-equivalence or routing-baseline gain is inferred from this single service trace.

\subsubsection{S3 observation and reporting boundary}
The S3 topology is specified in Fig.~\ref{fig:s3-layout}, and its replacement delivery and control results occupy the two parts of Fig.~\ref{fig:s3-campaign}. Raw delivery counts endpoint pairs before fidelity qualification; qualified delivery applies $F\geq0.8$ to the delivered resources. Raw rates use 1\,s bins. The conditional endpoint fidelity is count weighted, whereas the cumulative qualified curves accumulate only accepted deliveries. A count-weighted bin mean is not substituted for the individual qualification decision.

The reported diagnostics distinguish actual storage ages at BSM readout, storage-age distributions at the OGS and remote endpoints, reserved/stored OGS slots, and BSM busy-time fraction. Reserved/stored occupancy is not interpreted as the number of ready, high-fidelity pairs. The resource/policy cases are labeled FIFO with 100\,$\mu$s service, freshest-first control, 50\,ms BSM service, and 0.1\,s optical interpolation. The 100\,$\mu$s and 50\,ms values are service durations; 0.1\,s is an optical interpolation setting, not a delivery-count bin or a mandatory swap interval. The panel sets contain one counting realization and an across-realization mean but do not specify their replication count or confidence intervals; cumulative totals and other numerical summaries are therefore quoted only at plot-reading precision. The curves do not by themselves identify a microscopic memory or BSM instrument, whose selection belongs to the simulation configuration.

\subsection{Dated Paris--Vienna configuration, records, and geographical context}
\label{si:dated-context}\label{sec:dated-pass}
The dated design case uses city-center proxies at Paris $(48.8566^\circ\mathrm N,2.3522^\circ\mathrm E)$ and Vienna $(48.2082^\circ\mathrm N,16.3738^\circ\mathrm E)$, with zero station altitude. The assumed midpoint is 25 September 2026 at 17:00 UTC. The circular 360\,km orbit has inclination approximately $85.77^\circ$ and spans a 600\,s observation. The shared interval above $10^\circ$ lasts approximately 312.656\,s. Geometric, propagation, and counting timestamps are separately retained in the numerical archive.

The mean midpoint budget per arm includes about 26.770\,dB geometrical loss, 0.539\,dB atmospheric loss, 0.038\,dB configured mean turbulence loss, 0.132\,dB static pointing loss, 1.549\,dB coupling loss, 0.706\,dB source loss, and 1.249\,dB detection loss. The total is about 30.984\,dB per arm. The deterministic budget excludes memory recall and BSM.

For the same 300 two-second bins, the total failure column closes to 313\,914\,582 stage-specific increments: 313\,853\,542 link failures, 14\,937 recall failures, 45\,771 BSM failures, and 332 full-storage rejections. These categories are not divided by one common denominator to infer a composite success probability. The final retained inventory is excluded from delivered counts.

\subsubsection{Geographical and illumination overlays}
The global and European maps and all three pass snapshots are retained in Fig.~\ref{fig:dated-context}(a,b). Illumination is a diagnostic overlay: it does not gate the fixed-background channel in this design case. At the midpoint, modeled solar elevation is approximately $+6.45^\circ$ at Paris and $-2.78^\circ$ at Vienna.  Terrain, atmospheric refraction, and a satellite-eclipse calculation are outside these overlays.
\begin{figure*}[!tbp]
\centering
\begin{subfigure}[t]{\textwidth}
\centering
\includegraphics[width=.98\linewidth]{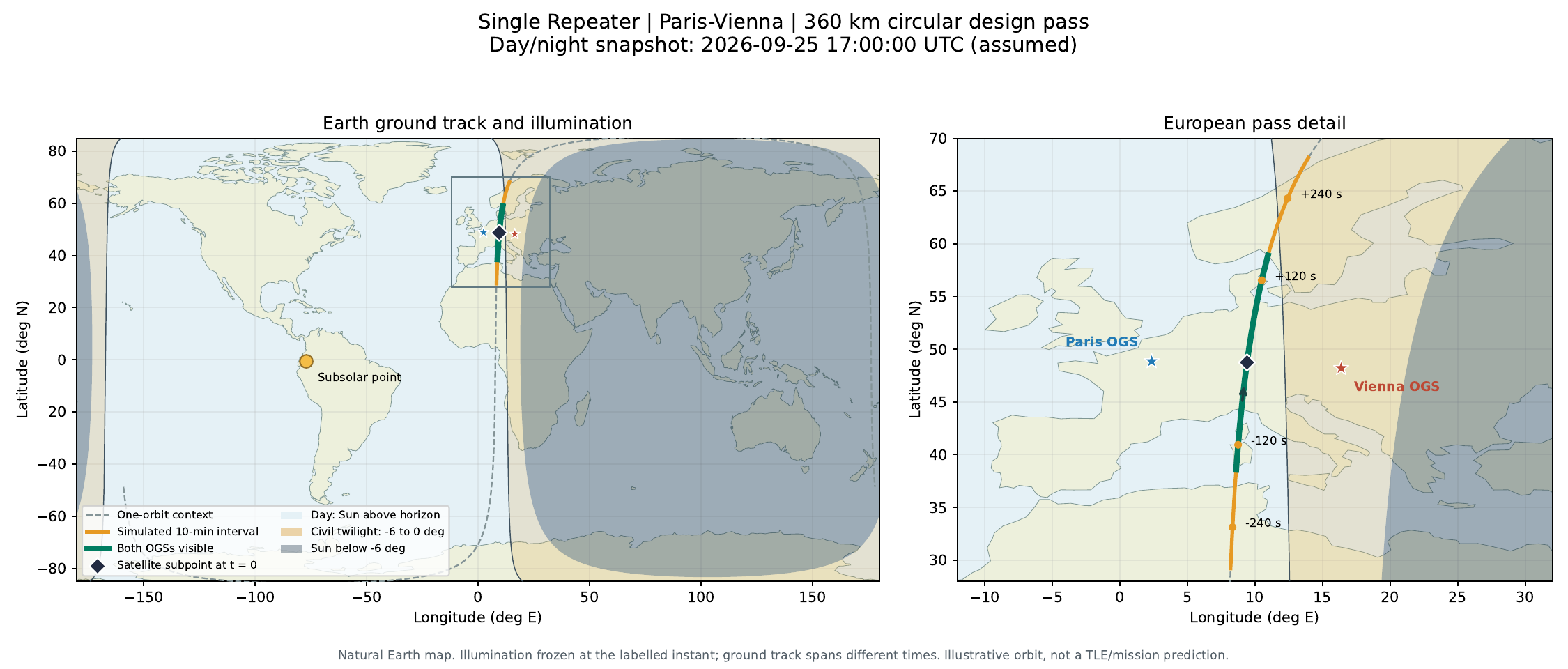}
\caption{Global illumination and European ground-track context at the labeled midpoint.}
\label{fig:dated-map}
\end{subfigure}
\par\vspace{4mm}
\begin{subfigure}[t]{\textwidth}
\centering
\includegraphics[width=.98\linewidth]{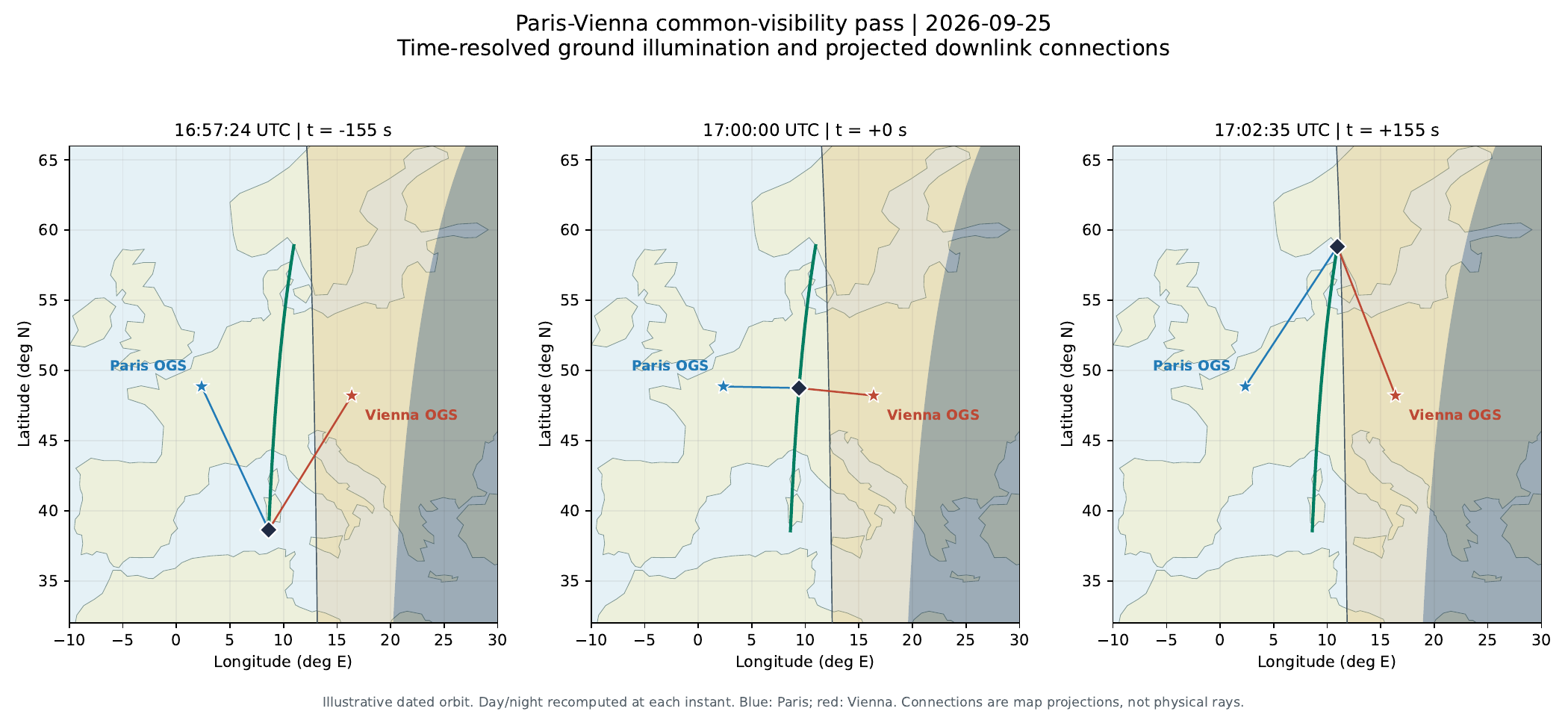}
\caption{Snapshots near the beginning, midpoint, and end of common visibility.}
\label{fig:dated-snapshots}
\end{subfigure}
\caption{\textbf{Dated S1 geographical and illumination context.} (a) The colored ground track spans different times, with common visibility highlighted, while illumination is evaluated at the labeled midpoint. (b) Ground illumination is recomputed for each of the three labeled instants; colored lines are projected connections to Paris and Vienna, not physical ray paths. The stations are city-center proxies and the orbit is a circular design case rather than a TLE or reconstructed spacecraft pass. Solar positions follow the NOAA approximate geometric equations; the base map is Natural Earth public-domain land data. Second-rounded UTC stamps are the timestamp references, while offset annotations are coarsely rounded. Illumination is diagnostic and does not modify the simulated background.}
\label{fig:dated-context}
\end{figure*}

\subsection{Direct Matera--Helmos inputs and extended receiver analysis}
\label{si:direct}
\subsubsection{Geometry, optical assumptions, and reference point}
The direct-distribution case uses one onboard entangled-pair source and two simultaneous downlinks, with one photon sent to Matera and the other to Helmos. It isolates optical collection and receiver count statistics from memory synchronization and swapping. Table~\ref{tab:quasar-inputs} defines the geometry and optical assumptions. This case corresponds to the memory-free resource stage. The 810 and 1550\,nm cases use equal wavelengths on their two arms, equal detector efficiency, and equal timing jitter; they compare wavelength-dependent propagation under a common receiver model.

\begin{table*}[!tbp]
\centering\small
\caption{Direct-downlink configuration. Geographic positions use WGS84. The two-body elements have a J2000-like frame and an assumed UTC epoch. Additional optical and receiver values are stated modeling assumptions.}
\label{tab:quasar-inputs}
\begin{tabularx}{\textwidth}{P{4.7cm}P{4.6cm}Y}
\toprule
Parameter & Matera / common & Helmos / interpretation\\
\midrule
Latitude, longitude & $40.6493^\circ$, $16.7041^\circ$ & $37.9854^\circ$, $22.1983^\circ$\\
Station altitude; receiver diameter & $536.9$\,m; $1.5$\,m & $2340$\,m; $2.3$\,m\\
Receiver blocked-area fraction & $0.1000$ & $0.1035$\\
Minimum elevation & $20^\circ$ & Both stations\\
Semimajor axis; eccentricity & $6939.1284$\,km; $0$ & Circular two-body reference\\
Inclination; RAAN & $97.6355^\circ$; $284.5^\circ$ & Inertial element convention\\
Argument of periapsis; initial anomaly & $0^\circ$; $54.0783^\circ$ & Phase at the stated epoch\\
Epoch; calculation spacing & 1 Jan. 2032, 22:59:28 UTC; $1$\,s & Assumed time/frame interpretation\\
Transmitter diameter; waist radius & $0.25$\,m; $0.125$\,m & Each arm; equivalent untruncated Gaussian\\
Beam quality; pointing RMS & $M^2=1$; $1\,\mu$rad per axis & Independent, zero-mean Gaussian pointing\\
Wavelengths; source fidelity & $810$ and $1550$\,nm; $F_0=0.99$ & One equal-wavelength case at a time\\
TX efficiency; RX efficiency; PDE & $0.80$; $0.50$; $0.60$ & Each applied once, both wavelengths\\
Zenith atmospheric transmission & $T_z=0.90$ & Slant transmission $T_z^{1/\sin\epsilon}$\\
Detector jitter; detectors per station & $50$\,ps RMS each; $4$ & Independent jitter; evenly loaded outputs\\
Dark rate; spectral bandwidth & $100$\,s$^{-1}$ per detector; $1$\,nm & $400$\,s$^{-1}$ per station\\
Field-stop half-angle; sky radiance & $10\,\mu$rad; $10^{-6}$\,W\,m$^{-2}$\,sr$^{-1}$\,nm$^{-1}$ & Same reference radiance at both sites/wavelengths\\
HV surface coefficient; wind & $1.7\times10^{-14}$\,m$^{-2/3}$; $21$\,m/s & Atmospheric-profile sensitivities in Section~\ref{sec:fieldstop}\\
\bottomrule
\end{tabularx}
\end{table*}

The reference two-body period is $95.88$\,min. Maximum elevations are $62.64^\circ$ at Matera and $63.21^\circ$ at Helmos. Shared samples above the elevation mask run from 23:20:00 to 23:24:52 on the assumed UTC date, corresponding to an approximately $293$\,s exposure on the $1$\,s grid.  This is shown in Fig. \ref{fig:direct-geometry}.

\begin{figure*}[!tbp]
\centering
\includegraphics[width=.92\textwidth]{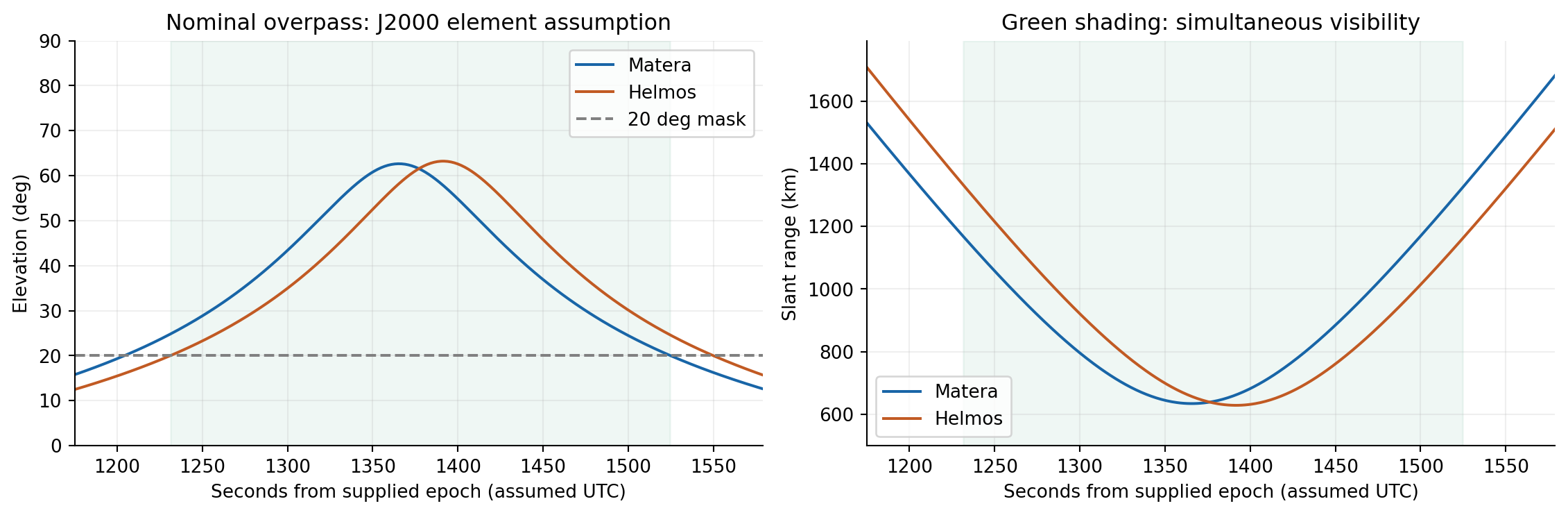}
\caption{\textbf{Direct baseline: Matera--Helmos pass geometry.} Elevations and slant ranges for the supplied circular-orbit interpretation. Shading identifies simultaneous visibility above the $20^\circ$ mask. The per-pass optical analysis and the frozen-peak operating-point calculation use different observation boundaries, stated in the text.}
\label{fig:direct-geometry}
\end{figure*}

The equivalent Gaussian model includes pointing through
\begin{align}
\theta_{\rm div}&=\frac{M^2\lambda_0}{\pi w_0},\\
w_{\rm eff}^2&=w_0^2+L^2(\theta_{\rm div}^2+4\sigma_\theta^2),
\label{eq:quasar-weff}\\
\eta_{\rm ann}&=e^{-2r_i^2/w_{\rm eff}^2}-e^{-2r_o^2/w_{\rm eff}^2},
\end{align}
where $r_o=D_r/2$ and $r_i=r_o\sqrt{o}$ for blocked-area fraction $o$. The per-arm detection probability before the separate field-stop sensitivity is
\begin{equation}
\eta_j=\eta_{\rm TX}\eta_{{\rm ann},j}T_z^{1/\sin\epsilon_j}\eta_{\rm RX}\eta_d.
\label{eq:quasar-detect}
\end{equation}
Gaussian pointing is already included in $w_{\rm eff}$ \cite{pirandola_sat}. Table~\ref{tab:quasar-loss} separates the incremental pointing contribution diagnostically.

\begin{table*}[!tbp]
\centering\small
\caption{Peak reference loss budget in dB. Collection and the incremental pointing contribution sum to the pointing-broadened annular collection. The total detection probability is used only once in the count equations.}
\label{tab:quasar-loss}
\begin{tabularx}{\textwidth}{Yrrrr}
\toprule
Contribution & 810 Matera & 810 Helmos & 1550 Matera & 1550 Helmos\\
\midrule
Gaussian annular collection & 3.86 & 1.97 & 8.43 & 5.24\\
Incremental pointing contribution & 2.17 & 1.35 & 0.91 & 0.80\\
Slant atmospheric loss & 0.52 & 0.52 & 0.52 & 0.52\\
TX and RX fixed loss & 3.98 & 3.98 & 3.98 & 3.98\\
Total optical loss & 10.53 & 7.82 & 13.83 & 10.53\\
Detection loss & 2.22 & 2.22 & 2.22 & 2.22\\
Total including detection & 12.74 & 10.04 & 16.05 & 12.75\\
\bottomrule
\end{tabularx}
\end{table*}

The frozen probabilities are $(\eta_M,\eta_H)=(0.053163022,0.099111603)$ at 810\,nm and $(0.024818782,0.053066425)$ at 1550\,nm. These numerical inputs retain guard digits for exact reproduction; they do not imply corresponding physical precision. The ideal linear-detector, $10^9$\,pairs/s stress reference gives approximately $5.269\times10^6$ and $1.317\times10^6$ peak genuine coincidences/s. Its station singles loading reaches $53.16$ and $99.11$\,Mcps at 810\,nm, motivating the fidelity-qualified and finite-live-time operating points below.

\subsubsection{Background, coincidence acceptance, and a fidelity constraint}
For entrance-plane sky spectral radiance $\mathcal L_\lambda$, clear area $A_{\rm clear}$, field solid angle $\Omega$, and bandwidth $\Delta\lambda$, the station sky rate is \cite{bonato}
\begin{equation}
B_j=\mathcal L_\lambda A_{{\rm clear},j}\Omega\Delta\lambda
\eta_{\rm RX}\eta_d\frac{\lambda_0}{h_{\rm P}c}.
\label{eq:quasar-background}
\end{equation}
Use wavelength in meters in the photon-energy factor and nm for the bandwidth when radiance is per nm. The reference sky rates are $(611.2,1431.4)$\,s$^{-1}$ at 810\,nm and $(1169.6,2739.2)$\,s$^{-1}$ at 1550\,nm. With the station dark rate, the total noise inputs are $(1011.2,1831.4)$\,s$^{-1}$ and $(1569.6,3139.2)$\,s$^{-1}$.

For CW Poisson pair brightness $R$, all-pairs time-tag association, and a full relative window $\tau$, define \cite{szlachetka}
\begin{align}
g(\tau)&=\operatorname{erf}\!\left(\frac{\tau}{2\sqrt2\sigma_\Delta}\right),
\quad \sigma_\Delta=\sqrt2\times50\,\mathrm{ps},\\
C&=R\eta_M\eta_Hg(\tau),\\
A&=(R\eta_M+N_M)(R\eta_H+N_H)\tau,\\
F_{\rm mix}&=\frac{F_0C+A/4}{C+A}.
\label{eq:quasar-fmix}
\end{align}
The accidental term includes different-pair coincidences and independent sky/dark contributions. The $1/4$ overlap declares a completely mixed effective accidental-qubit state. The statistic is the fidelity of that accepted mixture model.

For $1/4<F_\star<F_0$, write
\begin{equation}
\zeta=\frac{F_0-F_\star}{F_\star-1/4},\quad
b_j=\frac{N_j}{\eta_j},\quad K=\frac{\zeta g(\tau)}{\tau}-b_M-b_H.
\end{equation}
Then $F_{\rm mix}\ge F_\star$ is equivalent to $R^2-KR+b_Mb_H\le0$. If $K\ge2\sqrt{b_Mb_H}$, the feasible brightness interval is
\begin{equation}
R_-\le R\le R_+,\qquad R_\pm=\frac{K\pm\sqrt{K^2-4b_Mb_H}}2.
\label{eq:brightness-bound}
\end{equation}
Thus background sets a lower brightness limit while different-pair accidentals set an upper limit. At $R=10^9$\,s$^{-1}$ and $\tau=1$\,ns, $R\tau=1$ and $F_{\rm mix}\simeq0.62$; this stress point fails the selected $F_\star=0.90$ target.

\subsubsection{Detected singles, live time, and optimized operation}
For four evenly loaded nonparalyzable detectors per station and dead time $t_d$, the independent live-fraction approximation gives
\begin{equation}
\ell_j=\left[1+\frac{(R\eta_j+N_j)t_d}{4}\right]^{-1},
\qquad C_{\rm live}=C\ell_M\ell_H.
\label{eq:quasar-live}
\end{equation}
The argument is \emph{detected singles before dead-time suppression}. Applying the same live fractions to independent accidentals preserves Eq.~\eqref{eq:quasar-fmix} within this approximation. Table~\ref{tab:target-rates} reports fidelity-qualified brightness and genuine-coincidence contributions.

\begin{table*}[!tbp]
\centering\small
\caption{Frozen-peak operating points with full window $\tau=1$\,ns. The last two columns are genuine coincidences/s at the stated mixture-fidelity target. The live-time column assumes four detectors per station and $t_d=50$\,ns.}
\label{tab:target-rates}
\begin{tabularx}{\textwidth}{Yrrrr}
\toprule
Wavelength (nm) & $F_\star$ & Brightness $R_+$ (million pairs/s) & $C$ (pairs/s) & $C_{\rm live}$ (pairs/s)\\
\midrule
810 & 0.90 & 138.42 & 729,366 & 570,131\\
1550 & 0.90 & 138.34 & 182,199 & 160,008\\
810 & 0.95 & 57.11 & 300,892 & 270,728\\
1550 & 0.95 & 57.02 & 75,098 & 71,100\\
\bottomrule
\end{tabularx}
\end{table*}

Narrowing the window alone has a finite limit because
\begin{equation}
\frac{\tau}{g(\tau)}\longrightarrow\sqrt{2\pi}\sigma_\Delta
\quad (\tau\to0).
\end{equation}
At $10^9$\,pairs/s and the specified jitter, the zero-background limiting fidelity is $0.8786$. A $0.90$ target therefore requires a brightness reduction as well as timing selection. Optimizing brightness and a $20$--$5000$\,ps window, subject to $R\le10^9$\,s$^{-1}$, gives the values in Table~\ref{tab:joint-optimum}. These are peak stationary optima of the declared model; pass-integrated adaptive operation is a different optimization.

\begin{table}[!tbp]
\centering\small
\caption{Joint brightness/window optimization at $F_\star=0.90$ and $t_d=50$\,ns, retaining the frozen optical probabilities.}
\label{tab:joint-optimum}
\begin{tabular}{rrrr}
\toprule
$\lambda$ (nm)& $R$ ($10^6$ pairs/s)&$\tau$ (ps)&$C_{\rm live}$ (pairs/s)\\
\midrule
810&489.92&265.51&$1.139\times10^6$\\
1550&534.14&233.63&$4.017\times10^5$\\
\bottomrule
\end{tabular}
\end{table}

\subsubsection{Signal and background through the same field stop}
\label{sec:fieldstop}
The receiver sensitivity treats field acceptance consistently on the signal and noise sides. For a normalized angular point-spread function $P_j(\boldsymbol\alpha)$ conditional on entrance-pupil collection,
\begin{equation}
E_j(\theta)=\int_{|\boldsymbol\alpha|\le\theta}P_j(\boldsymbol\alpha)\,d^2\boldsymbol\alpha,
\quad\eta_j(\theta)=\eta_{j,\rm pre}E_j(\theta).
\end{equation}
The same half-angle sets $\Omega=2\pi(1-\cos\theta)$ in Eq.~\eqref{eq:quasar-background}. A uniform sky is integrated over that accepted solid angle.

For a uniformly illuminated annular pupil and no AO, the phase-only Kolmogorov long-exposure optical transfer function is \cite{fried1966,andrews2005}
\begin{equation}
T(\nu)=T_{\rm pupil}(\lambda_0\nu)
\exp[-3.44(\lambda_0\nu/r_0)^{5/3}],
\end{equation}
where $\nu$ is angular spatial frequency in cycles/radian and $T_{\rm pupil}$ is the normalized annular-pupil autocorrelation. Fourier--Bessel inversion gives
\begin{equation}
E(\theta)=2\pi\theta\int_0^{D/\lambda_0}T(\nu)J_1(2\pi\theta\nu)\,d\nu.
\label{eq:fieldstop-energy}
\end{equation}
The Fried parameter uses the HV profile integrated from station height to 30\,km, with $r_0=[0.423k^2\csc\epsilon\int C_n^2(h)\,dh]^{-3/5}$.

At 810\,nm, $20^\circ$ elevation, and a $10\,\mu$rad half-angle, a sea-level-referenced surface term yields $E_M=0.9251$, $E_H=0.9643$, and $0.50$\,dB combined field-stop loss. Rebasing only the surface exponential to each station height yields $0.5314$, $0.5723$, and $5.17$\,dB. Equation~\eqref{eq:fieldstop-energy} defines the angular sensitivity. These two prescribed profiles are separate physical assumptions, not fits to one dataset. Doubling the half-angle approximately quadruples the sky term, while the signal gain follows $E(\theta)$.


The receiver calculation is separate from Tables~\ref{tab:target-rates}--\ref{tab:joint-optimum}, which hold the original optical probabilities fixed. Multiplying their probabilities by $E_j$ and recomputing the bounds gives the corresponding combined sensitivity. The present calculation omits amplitude scintillation, nonuniform finite-beam pupil illumination, receiver tracking residuals, and measured AO errors; those model boundaries are consolidated in Section~\ref{sec:limitations}.



\subsection{Controller messages and service-accounting conventions}
\label{methods:control-contract}
Table~\ref{tab:api} specifies the messages exchanged between the controller and execution layer.
\begin{table*}[!tbp]
\centering\small
\caption{Controller--execution contract. Every message carries request/plan epoch and simulation timestamps; resource references include allocation versions.}
\label{tab:api}
\begin{tabularx}{\textwidth}{P{3.0cm}P{4.8cm}Y}
\toprule
Message/action & Payload & Resource meaning\\
\midrule
\code{ENT\_REQUEST} & Endpoints, service, $F_{\min}$, deadline & Creates request state and admission decision.\\
\code{StartSchedule} & Link/source ID, interval, rate/modes & Authorizes attempts within source and terminal budgets.\\
\code{REPORT\_STORAGE} & Endpoint/mode IDs, write and herald times, state tag & Updates the controller after signaling delay.\\
\code{InstructSwap} & BSM ID, two operands/versions, expiry & Atomically validates and reserves two local qubits.\\
\code{SWAP\_RESULT} & Outcome, consumed IDs, surviving endpoints, frame & Records destructive execution and surviving histories.\\
\code{StoreToken} & Derived pair identity and surviving endpoints & Associates metadata; an extra physical write is explicit.\\
\code{DeliverPair} & Pair ID, frame, quality, service-ready time & Announces the usable service at the endpoints.\\
\code{ClearRequest} & Request ID and epoch & Clears residual resources and invalidates stale actions.\\
\bottomrule
\end{tabularx}
\end{table*}
The full service definitions additionally include network access area, auxiliary-channel traffic, and synchronization precision. Stocked-service latency starts with suitable elementary inventory already available; cold-start latency includes generation. Both physical-ready and frame-known completion can be retained. Contact availability is geometrical; service availability includes quality, timing, and resource requirements.

\section*{Data availability}

The numerical data supporting the findings of this study will be made
publicly available upon publication through the QNETSIM repository at
\url{https://github.com/seidkoudia-max/QNETSIM}. The release will include
the simulation outputs underlying the figures and tables, scenario
configurations, physical-model parameters, random seeds, and analysis
and plotting scripts. Software-version and execution-environment
information will accompany the data to support reproduction of the
reported results.

\section*{Code availability}
The QNETSIM repository is \url{https://github.com/seidkoudia-max/QNETSIM}. The reproducibility release will contain the simulator, scenario drivers, physical-model configurations, environment information, and commands for regenerating the results. The source package includes the numerical reanalysis and reference-calculation scripts used in this article. 

\section*{Author contributions}

S.K. developed the quantum-memory models. S.K. and H.P. developed the
free-space optical (FSO) channel models, and H.P. modeled the
Hong--Ou--Mandel (HOM) interference. F.V. developed the models for
segmented entanglement swapping and scheduling. J.C. developed the
application programming interfaces (APIs). I.M. and I.B. supervised
the development of the software-defined networking (SDN) component.
All authors contributed to the discussion of the QNETSIM results.
S.K. and S.C. supervised the overall work.

\section*{Acknowledgements}

This work was carried out as part of the European Space Agency (ESA)
project \emph{Quantum information NETworks SIMulator} (QNETSIM),
under ESA Contract No.~40001 49440125/NL/FGL/lf.
\bibliographystyle{IEEEtran}
\begingroup
\small
\bibliography{references-2}

@article{parny2023,
  author  = {de Forges de Parny, L. and others},
  title   = {Satellite-based quantum information networks: use cases, architecture, and roadmap},
  journal = {Communications Physics},
  volume  = {6},
  pages   = {12},
  year    = {2023},
  doi     = {10.1038/s42005-022-01123-7}
}

@article{netsquid,
  author  = {Coopmans, T. and others},
  title   = {{NetSquid, a NETwork Simulator for QUantum Information using Discrete events}},
  journal = {Communications Physics},
  volume  = {4},
  pages   = {164},
  year    = {2021},
  doi     = {10.1038/s42005-021-00647-8}
}

@misc{qunetsimdocs,
  author       = {{QuNetSim authors}},
  title        = {Introduction},
  howpublished = {Official {QuNetSim} documentation},
  url          = {https://tqsd.github.io/QuNetSim/intro.html},
  urldate      = {2026-09-24}
}

@misc{qnetdocs,
  author       = {{Baidu Quantum}},
  title        = {{qcompute-qnet}},
  howpublished = {{PyPI}, official-maintainer package description and tutorial inventory},
  url          = {https://pypi.org/project/qcompute-qnet/},
  urldate      = {2026-09-24}
}

@misc{sequence,
  author        = {Wu, X. and Kolar, A. and Chung, J. and Jin, D. and Zhong, T. and Kettimuthu, R. and Suchara, M.},
  title         = {{SeQUeNCe: A Customizable Discrete-Event Simulator of Quantum Networks}},
  year          = {2020},
  howpublished  = {arXiv:2009.12000},
  eprint        = {2009.12000},
  archivePrefix = {arXiv},
  note          = {Revised 2021}
}

@misc{euroqci2025,
  author       = {{European Commission}},
  title        = {Commission and {European Space Agency} sign
                  {EuroQCI} implementation agreement},
  howpublished = {Shaping Europe's Digital Future},
  year         = {2025},
  month        = jan,
  url          = {https://digital-strategy.ec.europa.eu/en/news/commission-and-european-space-agency-sign-euroqci-implementation-agreement},
  
  urldate      = {2026-09-26}
}

@misc{eagle1,
  author       = {{European Space Agency}},
  title        = {{Eagle-1}},
  howpublished = {Official mission overview},
  url          = {https://www.esa.int/Applications/Connectivity_and_Secure_Communications/Eagle-1},
  
  urldate      = {2026-09-26}
}

@techreport{etsi_qkd014,
  author      = {{European Telecommunications Standards Institute}},
  title       = {Quantum Key Distribution ({QKD});
                 Protocol and data format of {REST}-based
                 key delivery {API}},
  institution = {ETSI},
  type        = {Group Specification},
  number      = {ETSI GS QKD 014 V1.1.1},
  year        = {2019},
  month       = feb,
  url         = {https://www.etsi.org/deliver/etsi_gs/QKD/001_099/014/01.01.01_60/gs_qkd014v010101p.pdf}
}

@techreport{etsi_qkd015,
  author      = {{European Telecommunications Standards Institute}},
  title       = {Quantum Key Distribution ({QKD});
                 Control Interface for Software Defined Networks},
  institution = {ETSI},
  type        = {Group Specification},
  number      = {ETSI GS QKD 015 V2.1.1},
  year        = {2022},
  month       = apr,
  url         = {https://www.etsi.org/deliver/etsi_gs/QKD/001_099/015/02.01.01_60/gs_QKD015v020101p.pdf}
}

@techreport{etsi_qkd020,
  author      = {{European Telecommunications Standards Institute}},
  title       = {Quantum Key Distribution ({QKD});
                 Protocol and data format of {REST}-based
                 Interoperable Key Management System {API}},
  institution = {ETSI},
  type        = {Group Specification},
  number      = {ETSI GS QKD 020 V1.1.1},
  year        = {2026},
  month       = jun,
  url         = {https://www.etsi.org/deliver/etsi_gs/QKD/001_099/020/01.01.01_60/gs_QKD020v010101p.pdf}
}

@techreport{rfc9340,
  author      = {Kozlowski, Wojciech and
                 Wehner, Stephanie and
                 Van Meter, Rodney and
                 Rijsman, Bruno and
                 Cacciapuoti, Angela Sara and
                 Caleffi, Marcello and
                 Nagayama, Shota},
  title       = {Architectural Principles for a Quantum Internet},
  institution = {RFC Editor},
  type        = {RFC},
  number      = {9340},
  year        = {2023},
  month       = mar,
  doi         = {10.17487/RFC9340},
  url         = {https://www.rfc-editor.org/rfc/rfc9340.html},
  note        = {Informational; Internet Research Task Force}
}

@article{pirandola2017,
  author  = {Pirandola, Stefano and Laurenza, Riccardo and
             Ottaviani, Carlo and Banchi, Leonardo},
  title   = {Fundamental limits of repeaterless quantum communications},
  journal = {Nature Communications},
  volume  = {8},
  pages   = {15043},
  year    = {2017},
  doi     = {10.1038/ncomms15043}
}

@article{wehner2018,
  author  = {Wehner, Stephanie and Elkouss, David and Hanson, Ronald},
  title   = {Quantum internet: A vision for the road ahead},
  journal = {Science},
  volume  = {362},
  number  = {6412},
  pages   = {eaam9288},
  year    = {2018},
  doi     = {10.1126/science.aam9288}
}

@inproceedings{quisp,
  author        = {Satoh, R. and others},
  title         = {{QuISP: a Quantum Internet Simulation Package}},
  booktitle     = {IEEE International Conference on Quantum Computing and Engineering},
  year          = {2022},
  pages         = {353--364},
  eprint        = {2112.07093},
  archivePrefix = {arXiv},
  note          = {Preprint: arXiv:2112.07093 (2021)}
}

@article{boone2015,
  author  = {Boone, K. and Bourgoin, J.-P. and Meyer-Scott, E. and Heshami, K. and Jennewein, T. and Simon, C.},
  title   = {Entanglement over global distances via quantum repeaters with satellite links},
  journal = {Physical Review A},
  volume  = {91},
  pages   = {052325},
  year    = {2015},
  doi     = {10.1103/PhysRevA.91.052325}
}

@article{liorni2021,
  author  = {Liorni, Carlo and Kampermann, Hermann and Bru{\ss}, Dagmar},
  title   = {Quantum repeaters in space},
  journal = {New Journal of Physics},
  volume  = {23},
  pages   = {053021},
  year    = {2021},
  doi     = {10.1088/1367-2630/abfa63}
}

@article{khatri2021,
  author  = {Khatri, Sumeet and Brady, Anthony J. and Desporte, Ren{\'e}e A. and Bart, Manon P. and Dowling, Jonathan P.},
  title   = {Spooky action at a global distance: analysis of space-based entanglement distribution for the quantum internet},
  journal = {npj Quantum Information},
  volume  = {7},
  pages   = {4},
  year    = {2021},
  doi     = {10.1038/s41534-020-00327-5}
}

@article{gundogan2021,
  author  = {G{\"u}ndo{\u g}an, Mustafa and Sidhu, Jasminder S. and Henderson, Victoria and Mazzarella, Luca and Wolters, Janik and Oi, Daniel K. L. and Krutzik, Markus},
  title   = {Proposal for space-borne quantum memories for global quantum networking},
  journal = {npj Quantum Information},
  volume  = {7},
  pages   = {128},
  year    = {2021},
  doi     = {10.1038/s41534-021-00460-9}
}

@article{wallnofer2022,
  author  = {Walln{\"o}fer, J. and others},
  title   = {Simulating quantum repeater strategies for multiple satellites},
  journal = {Communications Physics},
  volume  = {5},
  pages   = {169},
  year    = {2022},
  doi     = {10.1038/s42005-022-00945-9}
}

@article{yehia,
  author  = {Yehia, R. and Schiavon, M. and Marulanda Acosta, V. and Coopmans, T. and Kerenidis, I. and Elkouss, D. and Diamanti, E.},
  title   = {Connecting quantum cities: simulation of a satellite-based quantum network},
  journal = {New Journal of Physics},
  volume  = {26},
  pages   = {073015},
  year    = {2024},
  doi     = {10.1088/1367-2630/ad5b13}
}

@article{meister,
  author  = {Meister, J. and Kleinpa{\ss}, P. and Orsucci, D.},
  title   = {Simulation of satellite and optical link dynamics in a quantum repeater constellation},
  journal = {EPJ Quantum Technology},
  volume  = {12},
  pages   = {5},
  year    = {2025},
  doi     = {10.1140/epjqt/s40507-025-00307-8}
}

@misc{stbd2024,
  author        = {Koudia, S. and ur Rehman, J. and Chatzinotas, S.},
  title         = {{Space-Based Quantum Internet: Entanglement Distribution in Time-Varying LEO Constellations}},
  year          = {2024},
  howpublished  = {arXiv:2409.17032},
  eprint        = {2409.17032},
  archivePrefix = {arXiv},
  doi           = {10.48550/arXiv.2409.17032}
}

@misc{shi2026,
  author        = {Shi, Yaliang and Wang, Zi and Yuan, Bangguo and Wu, Gaojie and Zhao, Zhiwei},
  title         = {{QUASAR}: Quantum satellite architecture and routing simulator},
  year          = {2026},
  howpublished  = {arXiv:2609.04920v1},
  eprint        = {2609.04920v1},
  archivePrefix = {arXiv},
  note          = {Preprint}
}

@article{yin2017,
  author  = {Yin, Juan and others},
  title   = {Satellite-based entanglement distribution over 1200 kilometers},
  journal = {Science},
  volume  = {356},
  pages   = {1140--1144},
  year    = {2017},
  doi     = {10.1126/science.aan3211}
}

@misc{netsquiddocs,
  author       = {{NetSquid authors}},
  title        = {Official simulator overview and interface examples},
  howpublished = {Official {NetSquid} website},
  url          = {https://netsquid.org/},
  urldate      = {2026-09-24}
}

@misc{ryu2026,
  author       = {{Ryu maintainers}},
  title        = {{Ryu} project {README}},
  year         = {2026},
  url          = {https://github.com/faucetsdn/ryu},
  urldate      = {2026-09-25}
}

@article{fleischhauer,
  author  = {Fleischhauer, M. and Lukin, M. D.},
  title   = {Quantum memory for photons: Dark-state polaritons},
  journal = {Physical Review A},
  volume  = {65},
  pages   = {022314},
  year    = {2002},
  doi     = {10.1103/PhysRevA.65.022314}
}

@article{gorshkov_prl,
  author  = {Gorshkov, A. V. and Andr{\'e}, A. and Fleischhauer, M. and S{\o}rensen, A. S. and Lukin, M. D.},
  title   = {Universal approach to optimal photon storage in atomic media},
  journal = {Physical Review Letters},
  volume  = {98},
  pages   = {123601},
  year    = {2007},
  doi     = {10.1103/PhysRevLett.98.123601}
}

@article{gorshkov_free,
  author        = {Gorshkov, A. V. and Andr{\'e}, A. and Lukin, M. D. and S{\o}rensen, A. S.},
  title         = {Photon storage in {$\Lambda$}-type optically dense atomic media. {II}. Free-space model},
  journal       = {Physical Review A},
  volume        = {76},
  pages         = {033805},
  year          = {2007},
  doi           = {10.1103/PhysRevA.76.033805},
  eprint        = {quant-ph/0612083},
  archivePrefix = {arXiv},
  note          = {arXiv:quant-ph/0612083}
}

@article{crib_kraus,
  author  = {Kraus, B. and Tittel, W. and Gisin, N. and Nilsson, M. and Kr{\"o}ll, S. and Cirac, J. I.},
  title   = {Quantum memory for nonstationary light fields based on controlled reversible inhomogeneous broadening},
  journal = {Physical Review A},
  volume  = {73},
  pages   = {020302(R)},
  year    = {2006},
  doi     = {10.1103/PhysRevA.73.020302}
}

@article{crib_sangouard,
  author  = {Sangouard, N. and Simon, C. and Afzelius, M. and Gisin, N.},
  title   = {Analysis of a quantum memory for photons based on controlled reversible inhomogeneous broadening},
  journal = {Physical Review A},
  volume  = {75},
  pages   = {032327},
  year    = {2007},
  doi     = {10.1103/PhysRevA.75.032327}
}

@article{dd_filter,
  author  = {Bylander, J. and others},
  title   = {Noise spectroscopy through dynamical decoupling with a superconducting flux qubit},
  journal = {Nature Physics},
  volume  = {7},
  pages   = {565--570},
  year    = {2011},
  doi     = {10.1038/nphys1994}
}

@article{ion_memory,
  author  = {Drmota, P. and others},
  title   = {Robust quantum memory in a trapped-ion quantum network node},
  journal = {Physical Review Letters},
  volume  = {130},
  pages   = {090803},
  year    = {2023},
  doi     = {10.1103/PhysRevLett.130.090803}
}

@article{reiter,
  author        = {Reiter, F. and S{\o}rensen, A. S.},
  title         = {Effective operator formalism for open quantum systems},
  journal       = {Physical Review A},
  volume        = {85},
  pages         = {032111},
  year          = {2012},
  doi           = {10.1103/PhysRevA.85.032111},
  eprint        = {1112.2806},
  archivePrefix = {arXiv},
  note          = {arXiv:1112.2806}
}

@article{bsm,
  author  = {L{\"u}tkenhaus, N. and Calsamiglia, J. and Suominen, K.-A.},
  title   = {Bell measurements for teleportation},
  journal = {Physical Review A},
  volume  = {59},
  pages   = {3295},
  year    = {1999},
  doi     = {10.1103/PhysRevA.59.3295}
}

@article{bsm_limit,
  author  = {Calsamiglia, J. and L{\"u}tkenhaus, N.},
  title   = {Maximum efficiency of a linear-optical {Bell}-state analyzer},
  journal = {Applied Physics B},
  volume  = {72},
  pages   = {67--71},
  year    = {2001},
  doi     = {10.1007/s003400000484}
}

@article{pirandola_sat,
  author  = {Pirandola, S.},
  title   = {{Satellite Quantum Communications: Fundamental Bounds and Practical Security}},
  journal = {Physical Review Research},
  volume  = {3},
  pages   = {023130},
  year    = {2021},
  doi     = {10.1103/PhysRevResearch.3.023130}
}

@book{andrews2005,
  author    = {Andrews, Larry C. and Phillips, Ronald L.},
  title     = {{Laser Beam Propagation through Random Media}},
  edition   = {2nd},
  publisher = {SPIE Press},
  address   = {Bellingham, WA},
  year      = {2005},
  doi       = {10.1117/3.626196}
}

@article{fried1966,
  author  = {Fried, David L.},
  title   = {Optical resolution through a randomly inhomogeneous medium for very long and very short exposures},
  journal = {Journal of the Optical Society of America},
  volume  = {56},
  number  = {10},
  pages   = {1372--1379},
  year    = {1966},
  doi     = {10.1364/JOSA.56.001372}
}

@article{alhabash2001,
  author  = {Al-Habash, M. A. and Andrews, L. C. and Phillips, R. L.},
  title   = {Mathematical model for the irradiance probability density function of a laser beam propagating through turbulent media},
  journal = {Optical Engineering},
  volume  = {40},
  number  = {8},
  pages   = {1554--1562},
  year    = {2001},
  doi     = {10.1117/1.1386641}
}

@article{vasylyev,
  author  = {Vasylyev, D. and Semenov, A. A. and Vogel, W.},
  title   = {Atmospheric quantum channels with weak and strong turbulence},
  journal = {Physical Review Letters},
  volume  = {117},
  pages   = {090501},
  year    = {2016},
  doi     = {10.1103/PhysRevLett.117.090501}
}

@article{ao,
  author        = {Marulanda Acosta, V. and others},
  title         = {Analysis of satellite-to-ground quantum key distribution with adaptive optics},
  journal       = {New Journal of Physics},
  volume        = {26},
  pages         = {023039},
  year          = {2024},
  doi           = {10.1088/1367-2630/ad231c},
  eprint        = {2111.06747},
  archivePrefix = {arXiv},
  note          = {arXiv:2111.06747}
}

@techreport{vallado,
  author      = {Vallado, D. A. and Crawford, P. and Hujsak, R. and Kelso, T. S.},
  title       = {Revisiting {Spacetrack Report \#3}},
  institution = {AIAA},
  type        = {AIAA Paper},
  number      = {2006-6753},
  year        = {2006},
  note        = {Reference implementation and test cases maintained through {CelesTrak}}
}

@misc{skyfield,
  author       = {Rhodes, B.},
  title        = {{Skyfield}: high precision research-grade positions for planets and {Earth} satellites},
  howpublished = {Software documentation, ``Coordinates'' and ``Positions''},
  url          = {https://rhodesmill.org/skyfield/}
}

@article{heralded_memory,
  author  = {Tanji, H. and Ghosh, S. and Simon, J. and Bloom, B. and Vuleti{\'c}, V.},
  title   = {Heralded single-magnon quantum memory for photon polarization states},
  journal = {Physical Review Letters},
  volume  = {103},
  pages   = {043601},
  year    = {2009},
  doi     = {10.1103/PhysRevLett.103.043601}
}

@article{szlachetka,
  author  = {Szlachetka, J. and Joarder, K. and Kolenderski, P.},
  title   = {Ultrabright source of non-degenerate polarization-entangled photon pairs based on off-the-shelf polarization optics},
  journal = {Applied Physics Letters},
  volume  = {123},
  pages   = {144001},
  year    = {2023},
  doi     = {10.1063/5.0159000}
}

@article{bonato,
  author  = {Bonato, C. and Tomaello, A. and Da Deppo, V. and Naletto, G. and Villoresi, P.},
  title   = {Feasibility of satellite quantum key distribution},
  journal = {New Journal of Physics},
  volume  = {11},
  pages   = {045017},
  year    = {2009},
  doi     = {10.1088/1367-2630/11/4/045017}
}

@article{afc,
  author  = {Afzelius, M. and Simon, C. and de Riedmatten, H. and Gisin, N.},
  title   = {Multimode quantum memory based on atomic frequency combs},
  journal = {Physical Review A},
  volume  = {79},
  pages   = {052329},
  year    = {2009},
  doi     = {10.1103/PhysRevA.79.052329}
}

@article{ma2007,
  author  = {Ma, X. and Fung, C.-H. F. and Lo, H.-K.},
  title   = {Quantum key distribution with entangled photon sources},
  journal = {Physical Review A},
  volume  = {76},
  pages   = {012307},
  year    = {2007},
  doi     = {10.1103/PhysRevA.76.012307}
}
\endgroup

\end{document}